\documentclass{svjour2}                    % onecolumn
\smartqed  % flush right qed marks, e.g. at end of proof
\usepackage{graphicx}
\usepackage{amsfonts,amssymb} 

\usepackage{latexsym}
\usepackage{rawfonts}
\usepackage{multirow}
\usepackage{rotating}
\usepackage{lscape}
\usepackage{graphicx}\graphicspath{{figs/}}
\usepackage{subfigure}
\usepackage{epsfig}
\usepackage{fancybox}
\usepackage[usenames,dvipsnames]{color}
\usepackage{bbm}
\input{prepictex}
\input{pictex}
\input{postpictex}

\definecolor{Brown}{rgb}{0.7,0.3,0}
\definecolor{Tan}{rgb}{0.823,0.707,0.549}

\newcommand{\Claim}{\\ \noindent{\it Claim:\quad}}
\newcommand{\Claimproof}{\noindent {\it Proof of claim:\quad}}

\newcommand{\q}{\quad}
\newcommand{\qq}{\qquad}

\newcommand{\Ref}[1]{(\ref{#1})}

\newcommand{\NatN}{{\mathbb{N}}}

\newcommand{\IntN}{{\mathbb{Z}}}

\newcommand{\Lattice}{{\mathbb{L}}}
\newcommand{\Edges}{{\mathbb{E}}}

\def\L{\left(}
\def\R{\right)}
\def\LH{\left[}
\def\RH{\right]}
\def\LC{\left\{}
\def\RC{\right\}}

\def\LV{\left|}
\def\RV{\right|}

\def\lcl{\left\lceil}
\def\rcl{\right\rceil}
\def\vert{{\,\hbox{\large$|$}\,}}
\def\Vert{\hbox{\LARGE$|$}}

\def\colM{\color{Maroon}}

\def\O#1{\overline{#1}}
\def\W#1{\widetilde{#1}}

\def\C#1{{\cal #1}}

\def\betas#1{{\beta_#1}}

\def\eps{\epsilon}
\def\omegas#1{{\omega_#1}}

\def\Bi#1#2{ \L {{#1}\atop{#2}} \R }

\def\n#1{{n_#1}}
\def\m#1{{m_#1}}

\def\z#1{{z_#1}}

\def\sstack#1#2{\hbox{\tiny$\begin{array}{l} \hbox{\scriptsize$#1$} \\ \hbox{\scriptsize$#2$} \end{array}$}}

\def\indic{\hbox{\large$\displaystyle\mathbbm{1}$}}

\def\sfrac#1#2{\hbox{\normalsize $\frac{#1}{#2}$}}

\begin{document}

\title{Phase diagram of inhomogeneous percolation with a defect plane
\thanks{Research supported by grants from NSERC (Canada)}
}


\author{G.K. Iliev$^1$,
E.J. Janse van Rensburg$^2$
and N. Madras$^3$}


\institute{
$^1$ Chemistry, 
University of Toronto, Toronto, M5S~1A1, Canada\\
$^2$ Mathematics and Statistics, 
York University, Toronto, M3J~1P3, Canada \\
$^3$ Mathematics and Statistics, 
York University, Toronto, M3J~1P3, Canada\\
              \email{madras@yorku.ca}           %  \\
}

\date{Received: date / Accepted: date}

\maketitle

\begin{abstract}
Let $\Lattice$ be the $d$-dimensional hypercubic lattice
and let $\Lattice_0$ be an $s$-dimensional sublattice,
with $2 \leq s < d$.  We consider a model of inhomogeneous bond
percolation on $\Lattice$ at densities $p$ and $\sigma$,
in which edges in $\Lattice\setminus \Lattice_0$
are open with probability $p$, and edges in $\Lattice_0$ open
with probability $\sigma$.  We generalizee several 
classical results of (homogeneous) bond percolation 
to this inhomogeneous model.
The phase diagram of the model is presented, and it is 
shown that there is a subcritical regime for $\sigma< \sigma^*(p)$ 
and $p<p_c(d)$ (where $p_c(d)$ is
the critical probability for homogeneous percolation in $\Lattice$),
a bulk supercritical regime for $p>p_c(d)$, and a surface 
supercritical regime for $p<p_c(d)$ and  $\sigma>\sigma^*(p)$.  
We show that $\sigma^*(p)$ is a strictly decreasing function 
for $p\in[0,p_c(d)]$, with a jump discontinuity at $p_c(d)$.  
We extend the Aizenman-Barsky differential inequalities for homogeneous
percolation to the inhomogeneous model and use them to 
prove that the susceptibility is finite inside the subcritical phase.   
We prove that the cluster size distribution decays exponentially
in the subcritical phase,
and sub-exponentially in the supercritical phases.  
For a model of lattice animals with a defect plane, the free energy 
is related to functions of the 
inhomogeneous percolation model, and we show how
the percolation transition implies a non-analyticity in the free
energy of the animal model.  Finally, we present simulation estimates
of the critical curve $\sigma^*(p)$.
\end{abstract}

\section{Introduction}

Percolation \cite{BH57} in $\IntN^d$ is a lattice model of 
polymeric gelation \cite{PWRWO89,R89} and of chemical gelation
due to polymerisation of monomers or comonomers \cite{GRD02}.  In a 
percolation model the phenomenon of gelation is understood as a 
critical phenomenon \cite{CSK79} with characteristic scaling about a 
critical point called the percolation threshold.  Studies of gelation
from a percolation point of view are now classical \cite{S79,SCA82,KD01}, 
and were reviewed in references \cite{S79A,E72,HW83,SA94,E00}.

Surface phenomena in percolation have also received considerable attention
\cite{CCRS81,DB79,DB00,DBE81,DBE00,DBL93}.  This is a model of gelation along 
a defect plane or surface, and was also interpreted as a model of branch polymer 
adsorption in bulk \cite{DBE00} -- results in references \cite{DB79,DBE00} 
suggest that a surface transition is absent in two dimensional models.

Consider a model of inhomogeneous percolation in the hypercubic
 lattice $\IntN^d$ with an $s$-dimensional hyperplane $\IntN^s$
as a defect plane (where $2\leq s < d$).  This model 
has received some attention in the mathematical literature
\cite{NW97} (see reference \cite{MSS94} for a model of inhomogeneous
percolation with defect lines in the bulk lattice). 

Percolation along a defect plane may be considered as a model
of gelation along a surface defect, and it is known that
this phenomenon is associated with a \textit{surface transition} 
in addition to the usual bulk percolation 
phenomenon \cite{DBE81,DB00,DBE00,NW97}. 

There is a significant number of results known for (homogeneous) 
bond percolation \cite{G99,K82}.  Known results include
the location of the critical bond-percolation threshold in the 
square lattice \cite{K80} (see also \cite{H60}), the uniqueness of the critical 
point \cite{AB87,M86} and the decay rate of the clusters 
at the origin in the sub- and supercritical phases \cite{ADS80,AN84}.  
Analogous results for models of inhomogeneous percolation are incomplete,
and in this paper our aim is to provide some mathematical results to 
extend the standard theorems of homogeneous percolation to a 
model of inhomogeneous percolation.  This requires generalisation of
several of the classical results for homogeneous bond percolation.
A secondary goal is examine the connection between lattice models of
branch polymers close to a surface or defect plane and percolation
along a defect plane.

\subsection{Homogeneous percolation}
\label{sec-homog}

In this section we define some terms and notation, and we briefly 
review homogeneous percolation.   

The $d$-dimensional hypercubic lattice $\Lattice$ with vertices in
$\IntN^d$ has unit length edges joining nearest neighbour
vertices (or points) in $\IntN^d$. The set of edges of $\Lattice$ is 
denoted by $\Edges$.  We shall write $x{\sim}y$ to denote the 
edge that joins the vertices $x$ and $y$.

In bond percolation models, each edge $e\in\Edges$ has an 
associated random variable $\omega(e)$ with possible 
values $0$ and $1$. We say that the edge $e$ is \textit{open} if 
$\omega(e)=1$, and that $e$ is \textit{closed} if $\omega(e)=0$.  
In the present paper we always assume that the random variables 
$\omega(e)$ are independent.  

In the \textit{homogeneous} percolation model, the probability 
that $\omega(e)=1$ is the same for every $e$, and we denote 
this common value by $p$.

We call $p$ the \textit{density} of the  model.
We denote by $P^H_p$ the homogeneous (bond)-percolation
measure on $\Lattice$ at density $p$, and by $E^H_p$ the expectation
with respect to $P^H_p$. (The superscript ``H" 
will be used for functions describing homogeneous percolation, 
in contrast to the inhomogeneous model to be introduced below.)

%%%%%%%%%%%%%%%%%%%%%%
\begin{figure}[t!]
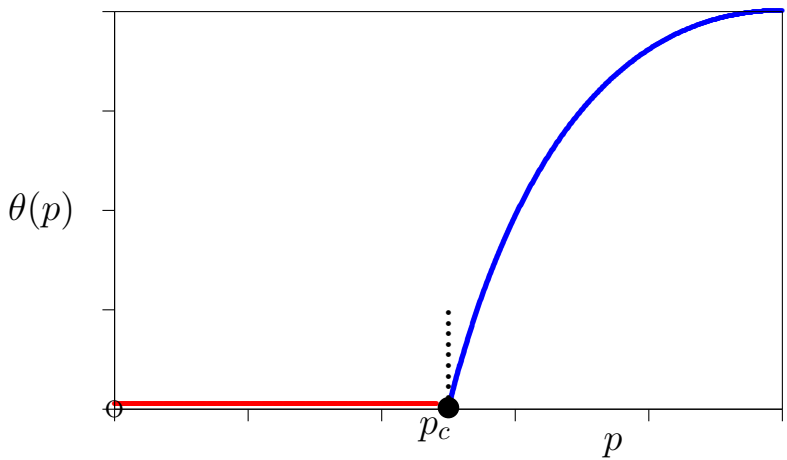

\centering
\input figure3.tex
\caption{A schematic graph of the probability that the cluster 
at the origin is infinite as a function of $p$ in homogeneous 
bond percolation.  This probability is zero for $p<p_c(d)$ 
and positive for $p>p_c(d)$.}
\end{figure}
%%%%%%%%%%%%%%%%%%%%%%

The union of open edges is a subgraph $G$ of $\Lattice$.  In general, $G$
is not a connected graph, but is the union of a collection
of connected subgraphs of open edges.  For a vertex $x$, let 
$C(x)$ be the connected component of open edges containing $x$.  
We call $C(x)$ the open cluster at $x$.  

When $x$ is the origin, we write $C$ instead of $C(0)$.

The collection of closed edges incident with $C(x)$ is the 
\textit{perimeter} of $C(x)$.   

The \textit{size} of the cluster $C$ is the number of 
vertices in $C$, and is denoted $|C|$.  We shall also work
with the number of edges in $C$, which we denote by $\|C\|$.

The probability that the origin is in a cluster of infinite 
size is given by
\begin{equation}
 \theta_d^H(p) \;\equiv \;\theta^H(p)
  \;=\;\lim_{n\to\infty} P^H_{p}(|C|{\geq} n)\;=\; P^H_p(|C|{=}\infty)
\label{eqn1}   %%ZXZ[eqn1]
\end{equation}
(we shall often omit the subscript $d$).

The fundamental property of percolation is that there
exists a critical density $p_c(d)\in(0,1]$ in the 
$d$-dimensional lattice (see reference \cite{G99}, section 1.4) 
such that
\begin{equation}
 \theta_d^H(p)\q 
\cases{
\;=\;0 & \mbox{if $p<p_c(d)$}, \cr
\;> \;0 & \mbox{if $p>p_c(d)$}.
}
\end{equation}
It is easy to see that $p_c(1)=1$; however the result that 
$p_c(2)=\sfrac{1}{2}$ requires considerably more effort \cite{K80}.
In general, it can be shown that $0< p_c(d+1) < p_c(d)$
\cite{K82,M87}.

The expected value of the size of the cluster at the origin is the 
\textit{susceptibility} defined by
\begin{equation}
\chi^H (p) \;=\; E^H_p|C| \;=\;\LH  \infty \cdot P^H_p(|C|{=}\infty)\RH
+ \sum_{n=1}^\infty n\,P^H_p(|C|{=}n) 
\end{equation}
(interpreting $\infty\cdot 0 = 0$).

If $p>p_c(d)$, then obviously $\chi^H(p)=\infty$.  It is also known that
$\chi^H(p)  < \infty$ whenever $p<p_c(d)$ 
(see references \cite{AB87,M86}, and also for 
example \cite{G99}).  This property is often referred to as the uniqueness of
the critical point.  The finite component of the susceptibility is given by
\begin{equation}
\chi^{f,H} (p) \;=\; E^H_p\left(|C|\, \indic_{|C|{<}\infty}\right)\;
=\; \sum_{n=1}^\infty n\,P^H_p(|C|{=}n) .
\end{equation}
Clearly $\chi^{f,H} (p) \leq \chi^{H} (p)$.

%It is known that $\chi^H(p)<\infty$ and $\chi^{f,H}(p) < \infty$ when $p< p_c(d)$. 
%These show that $P^H_p(|C| = n)$ approaches zero with 
%increasing $n$ in the subcritical regime. 
It is known that the limit 
\begin{equation}
 \zeta^H(p) \;=\; - \lim_{n \to \infty} \sfrac{1}{n} \log P^H_p(|C|{=}n)
\label{eqn4z}   %%ZXZ[eqn4z]
\end{equation}
exists and that $\zeta^H(p)>0$ if $p < p_c(d)$ \cite{KS78}.  Hence,
we have exponential decay of $P^H_p(|C|=n)$ in the subcritical regime.  
%There exists a $\lambda^H(p)>0$ such that 
%\begin{equation}
%  P^H_p(|C|=n) \;\leq \; P^H_p(|C|\geq n)\; \leq\; e^{-n\lambda^H(p)} 
%\end{equation}
%if $p < p_c(d)$. 
More explicitly, $P^H_p(|C|=n)$ is bounded from above by \cite{AN84}
\begin{equation}
 %\flushleft \hbox{\hspace{6mm}}
 P^H_p(|C|{=}n) \; \leq \; 
P^H_p(|C|{\geq} n)\; \leq \;2\, e^{- \sfrac{n}{2[\chi^H(p)]^2} }
\q\hbox{for all $n$, if $p < p_c(d)$.}
\end{equation}
%0Since $\chi^H(p) <\infty$ if $p<p_c(d)$, this bound is
%true for finite $n$ which are large enough.

In the supercritical phase the cluster size distribution of the
cluster at the origin has sub-exponential decay \cite{ADS80}:
\begin{equation}
P_p^H(|C| {=} n) \;\geq\; e^{-\gamma(p) \, n^{(d-1)/d}}
\label{eqn444z}   %%ZXZ[eqn444z]
\end{equation}
where $\gamma(p)$ is a finite function of $p\in(p_c(d),1]$. By taking
logarithms, dividing by $n$ and letting $n\to\infty$, this shows that 
$\zeta^H(p) = 0$ for $p>p_c(d)$.

%%%%%%%%%%%%%%%%%%%%%%
\begin{figure}[t!]
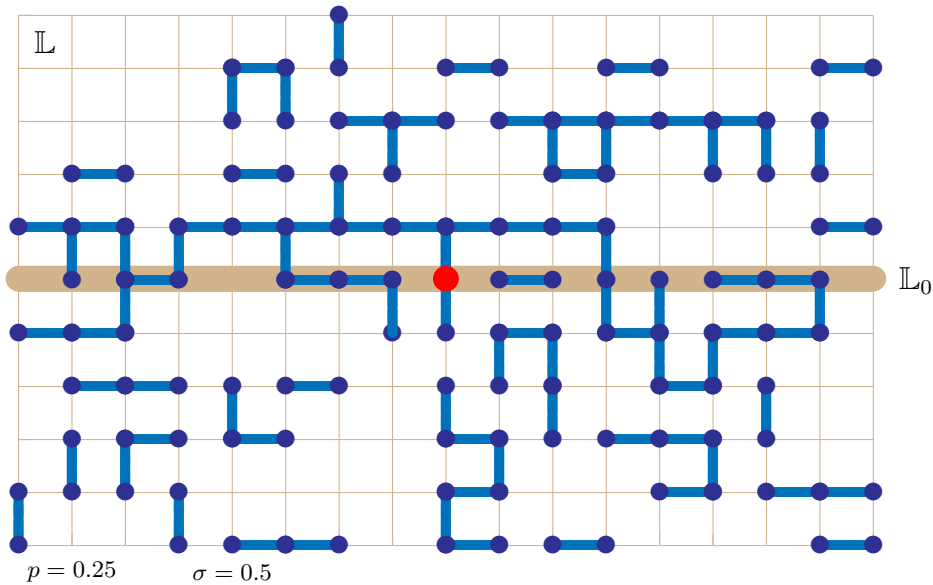

\centering\hfill
\input figure4.tex
\caption{Bond percolation in the square lattice $\Lattice$ with a 
(one dimensional) defect line $\Lattice_0$.  Edges in
$\Lattice_0$ are open with probability $\sigma$,  and 
(bulk) edges in $\Lattice\setminus\Lattice_0$ are open
with density $p$.  In this illustration a section of $\Lattice$
centered at the origin is shown and the densities were set
at $p=0.25$ and $\sigma=0.5$.}
\end{figure}
%%%%%%%%%%%%%%%%%%%%%%

\subsection{Inhomogeneous percolation}

Let $\{\vec{e}_1,\vec{e}_2,\ldots,\vec{e}_d\}$ be 
the standard basis of unit vectors 
in the $d$-dimensional hypercubic lattice $\Lattice$.  
Choose an integer $s$ such that $2\leq s < d$ and let
$\Lattice_0$ be the $s$-dimensional sublattice of $\Lattice$
which contains the origin and has basis vectors
$\{\vec{e}_1,\vec{e}_2,\ldots,\vec{e}_s\}$. 

We shall view $\Lattice_0$ as a ``defect plane" 
or ``adsorbing surface" in $\Lattice$.  The set 
of edges or bonds with both endpoints in $\Lattice_0$ 
is $\Edges_0$, and we shall write that 
$\Lattice_0 \subseteq \Lattice$, since $\Edges_0 \subseteq \Edges$.

Inhomogeneous bond percolation is set up in $\Lattice$ 
with one density $\sigma$ for the defect plane $\Edges_0$ 
and another density $p$ for the bulk ($\Edges \setminus \Edges_0$).
Given $p,\sigma\in[0,1]$, the inhomogeneous percolation 
probability measure $P^I_{p,\sigma}$ is given by
\begin{equation}
P^I_{p,\sigma} ( \omega(e){=}1) \;=\;
P^I_{p,\sigma} ( \hbox{$e$ is open} ) \;=\; \cases{
p &\hbox{if $e\in \Edges\setminus\Edges_0$},\cr
\sigma & \hbox{if $e\in \Edges_0$},
} 
\label{eqnZ88}     %%ZXZ[eqnZ88]
\end{equation}
with all edges independent.  The corresponding expectation 
is $E^I_{p,\sigma}$.

Open clusters $C(x)$ are defined as before.
The probability that $C$ (the cluster at the origin)
has infinite size is given by
\begin{equation}
\theta^I(p,\sigma) \;=\; \lim_{n\to\infty} P^I_{p,\sigma} (|C|{\geq} n)
 \;=\; P^I_{p,\sigma} (|C| = \infty)
\label{eqnZ99}     %%ZXZ[eqnZ99]
\end{equation}
(we suppress the dimensions $d$ and $s$ in this notation).  We say that
percolation occurs if $\theta^I(p,\sigma)>0$.  Clearly,
$\theta^I(p,p) = \theta^H(p)$, and $\theta^I(p,\sigma)$ is a non-decreasing
function of its arguments---that is, $\theta^I(p,\sigma)
\leq \theta^I(p^\prime,\sigma')$ if $p \leq p^\prime$ and
$\sigma \leq \sigma^\prime$.

Similarly, the susceptibility is defined by
\begin{equation}
%\fl \hspace{6mm}
\chi^I(p,\sigma) \;=\; E^I_{p,\sigma}|C|  
\;=\; \left[ \infty\cdot P_{p,\sigma}^I(|C|{=}\infty)\right]  \,+\, 
\sum_{n=1}^\infty n\,P_{p,\sigma}^I(|C|{=}n)
\label{eqnZ14}     %%ZXZ[eqnZ14]
\end{equation}
and we define
\begin{equation}
\chi^{f,I}(p,\sigma)\; =\; E^I_{p,\sigma}\left(|C| 
\,\indic_{|C|<\infty} \right) \;=\;
\sum_{n=1}^\infty n\,P_{p,\sigma}^I(|C|{=}n) .
\label{eqnZ15}     %%ZXZ[eqnZ15]
\end{equation}

%%%%%%%%%%%%%%%%%%%%%%
\begin{figure}[t!]
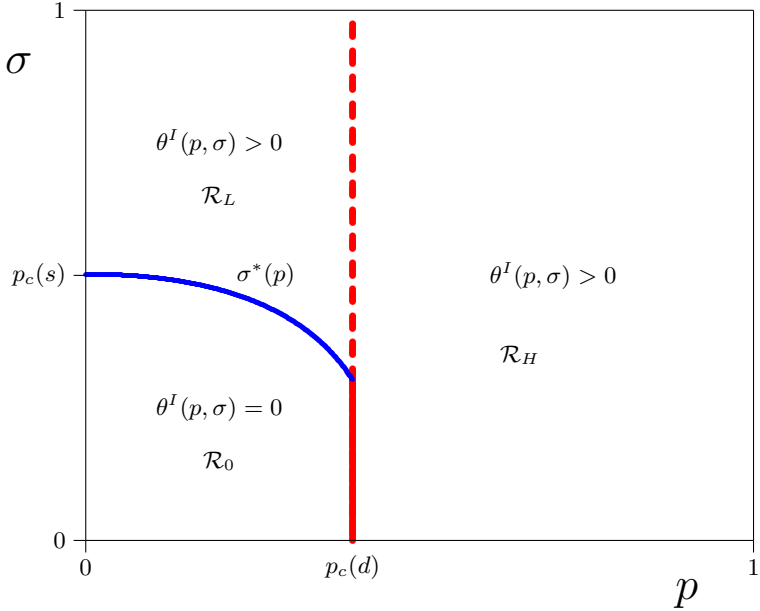

\centering\hfil
\input figure2N.tex
\caption{The phase diagram of inhomogeneous percolation.  The subcritical
phase is labeled by $\C{R}_0$, and we distinguish two supercritical phases:  The
regime labeled by $\C{R}_L$ is a surface supercritical phase, where percolation
has occurred along $\Lattice_0$ but has not penetrated into the bulk
of the lattice $\Lattice$.  In the regime marked by $\C{R}_H$ percolation 
occurs throughout the $d$-dimensional lattice $\Lattice$, since $p>p_c(d)$.}
\label{Diagram}    %%ZXZ[Diagram]
\end{figure}
%%%%%%%%%%%%%%%%%%%%

Figure \ref{Diagram} shows the three regimes of this model.
We shall begin with a formal definition of the three regimes $\C{R}_0$, $\C{R}_L$
and $\C{R}_H$, and then we shall describe their properties.

We define a critical curve $\sigma=\sigma^*(p)$ by
\begin{equation}
\label{eqn-sigmastardef}
\sigma^*(p) \;=\; \inf%_{\sigma\in[0,1]} 
\{\sigma\in[0,1] \, \vert \theta^I(p,\sigma) \;>\; 0\}  \hspace{5mm}
\hbox{for $p\in [0,1]$}.
\end{equation}
It is not hard to show that $\sigma^*(p)=0$ for $p>p_c(d)$, and $0<\sigma^*(p)<1$
for $0\leq p<p_c(d)$ (see proposition \ref{prop-sigstar}).
Accordingly, we define
\begin{eqnarray*}
    \C{R}_0  & = &  \{(p,\sigma)\in[0,1]^2 \,:\, \hbox{ $ p<p_c(d)$ and $\sigma<\sigma^*(p)$ } \}, \\
      \C{R}_L  & = &  \{(p,\sigma)\in[0,1]^2 \,:\, \hbox{ $ p<p_c(d)$ and $\sigma>\sigma^*(p)$ } \},
        \hbox{ and} \\
        \C{R}_H  & = &  \{(p,\sigma)\in[0,1]^2 \,:\, \hbox{ $ p>p_c(d)$  } \} .
\end{eqnarray*}

By definition, we see that $\theta^I(p,\sigma)=0$ at every point of  $\C{R}_0$.
Thus $\C{R}_0$ is the subcritical phase, in which every cluster is finite.  Also by 
the definition (\ref{eqn-sigmastardef}), we have
\begin{equation}
   \label{eqn-sigma0pcs}
     \sigma^*(0) \;=\; p_c(s) \,.
 \end{equation}

In $\C{R}_H$, the infinite cluster extends throughout the bulk.  
Indeed, for given $(p,\sigma)$ in $\C{R}_H$, the probability 
$P^I_{p,\sigma}(\vec{v}\in C)$ is bounded away from 
$0$ uniformly for all $\vec{v}$ in $\Lattice$ (see proposition \ref{prop-connect}).  
In contrast, for $(p,\sigma)$ in 
$\C{R}_L$, the probability $P^I_{p,\sigma}(v\in C)$ decays to 0 exponentially rapidly in
the distance from $v$ to $\Lattice_0$ (see lemma \ref{lemma888}).  Thus we
call $\C{R}_L$ the surface supercritical phase, where the infinite cluster stays close
to the defect plane rather than penetrating into the bulk, and we call $\C{R}_H$ 
the bulk supercritical phase, where the 
infinite cluster spreads throughout the whole lattice.
Alternatively, for large $N$ we see that $E(|C|\cap [-N,N]^d)$ is proportional to
$N^d$ when $(p,\sigma)$ is in $\C{R}_H$, and is proportional to $N^s$ when 
$(p,\sigma)$ is in $\C{R}_L$. 

From a slightly different perspective, 
if we fix $p<p_c(d)$, then increasing $\sigma$ takes the model through 
a percolation transition at $\sigma^*(p)$ from $\C{R}_0$ into  $\C{R}_L$.  
This transition is often 
referred to as a ``surface phase transition" 
in the model (\cite{DBL93}, see references \cite{DBE81,DB00} as well).

(We remark here that the case $s=1$, a defect line, has simpler behaviour:  for
every $\sigma$, we have 
$\theta^I(p,\sigma)=0$ if $p<p_c(d)$ and 
$\theta^I(p,\sigma)>0$ if $p>p_c(d)$ \cite{MSS94}.  That is, there is no 
$\C{R}_L$ phase.  For this reason, we assume $s\geq 2$ throughout this paper.)

If $p$ is increased so that $p>p_c(d)$ then the model goes 
through a bulk percolation threshold into a bulk super-critical 
regime $\C{R}_H$ --- see proposition \ref{prop-sigstar}.

%The regimes $\C{R}_0$ (which is the sub-critical regime), or $\C{R}_L$
%(the super-critical surface regime) and $\C{R}_H$ (the super-critical
%bulk percolation regime) may be separated by critical curves.  In particular,
%the sub-critical phase $\C{R}_0$ is separated from the surface super-critical
%phase $\C{R}_L$ by $\sigma=\sigma^* (p)$. 

The boundary between the regimes $\C{R}_0$ and $\C{R}_L$
will be denoted by $\C{R}_{0,L}$.
% and we shall write $v \in \C{R}_{0,L}$ for a point $v$ which is in the boundary.  
Obviously, the critical curve
$\sigma =\sigma^*(p)$  for $p\in[0,p_c(d))$ is a subset of $\C{R}_{0,L}$.  We expect
that they are equal (except perhaps for a point at $p=p_c(d)$), but we do not know this
rigorously because we cannot prove that the curve is continuous on $[0,p_c(d))$.  The curve
is obviously non-increasing, and Proposition \ref{prop.cubic} shows that it
is strictly decreasing on $[0,p_c(d)]$.

There is a jump discontinuity in $\sigma^*(p)$ at $p_c(d)$.
Indeed, proposition \ref{prop-sigstar}  in section \ref{section2} 
shows that $p_c(s)>\sigma^*(p)>p_c(d)$ for every $p$ in $(0,p_c(d))$
and that $\sigma^*(p)=0$ if $p>p_c(d)$.  Thus, the boundary between
$\C{R}_0$ and $\C{R}_H$ is a vertical line segment at $p=p_c(d)$.
For large enough $d$, Newman and Wu \cite{NW97} prove
the stronger result that $p_c(s) > \sigma^*(p_c(d)) > p_c(d)$ 
for $2\leq s \leq d-3$.  
They conjecture that this is true whenever $2\leq s<d$.
We prove that $p_c(s) > \sigma^*(p_c(d))$ in general (see corollary \ref{cor-ltpcs}).  
It would be much 
harder to prove $\sigma^*(p_c(d)) > p_c(d)$, since this would imply the
longstanding conjecture that $\theta_d^H(p_c(d))=0$.  In general, 
it seems hard to say much about $\sigma^*(p_c(d))$.

In section \ref{section3} we consider the uniqueness of the 
critical point.  This requires the generalisation of differential 
inequalities for homogeneous percolation in reference 
\cite{AB87} to the model in this paper.  This is done
in \ref{appendix1}, and the resulting modified inequalities 
are used to show that if $\chi^{I} (p,\sigma) =\infty$, 
then
%either $\theta^I(p,\sigma)>0$ or else $\theta^I(p,\sigma)=0$ 
%and $\theta^I(p+\Delta,\sigma+\Delta)>0$ for all small 
%positive $\Delta$.
$(p,\sigma)$ cannot be in the interior of $\C{R}_0$ (see theorem
\ref{theorem5}).

In sections \ref{section4} and \ref{section5}, we consider the distribution of 
the size of the cluster $C$ at the origin.
In the subcritical regime $\C{R}_0$, Theorem \ref{theoremZ6} shows
that the distribution 
of $|C|$ has exponential tails; more precisely, for every $n$
\begin{equation}
P_{p,\sigma}^I (|C| {=} n) 
\;\leq \; 2\,e^{-n/(2\,\chi^H(p)\chi^I(p,\sigma))} .
\end{equation}

In the supercritical regime $\C{R}_H$, there exists a $\gamma(p)>0$ such that
\begin{equation}
P_{p,\sigma}^I (\infty {>} |C| {\geq} n) 
\; \geq \;
P_{p,\sigma}^I (|C| {=} n) 
\; \geq \; e^{-\gamma(p) n^{(d-1)/d}} .
\label{eqn13AA}   %%ZXZ[eqn13AA]
\end{equation}
%provided that $\sigma \geq p > p_c(d)$. 
See theorem \ref{theorem5-12} in section \ref{section5}.
This result should be compared with the situation in regime 
$\C{R}_L$, where we show in theorem \ref{thm999} that there
exist positive $\betas1$ and $\betas2$ (depending on $p$ and $\sigma$)
such that
\begin{equation}
P_{p,\sigma}^I (\infty {>} |C| {\geq} n) 
\;\geq\; \betas1 e^{-\betas2 n^{(s-1)/s}
(\log^2 n)^{d-s}}. 
\end{equation}
It follows that
\begin{equation}
\lim_{n\to\infty} \frac{1}{n^{(d-1)/d}}
\log P_{p,\sigma}^I (\infty {>} |C| {\geq} n) 
\;=\; 0 \q\hbox{in regime $\C{R}_L$}.
\end{equation}
This suggests two different behaviours for the tails of
$P_{p,\sigma}^I (\infty {>} |C| {\geq} n) $ in
the regimes $\C{R}_L$ and $\C{R}_H$. 

In section \ref{section6} we consider briefly the relation of 
$P^I_{p,\sigma}(|C|=n)$ to a lattice animal model of polymer collapse 
near a defect plane.  We show that there is a limiting free energy for the 
lattice animals which implies the existence of the limits
\begin{equation}
%\fl
\psi^I(p,\sigma) \;=\; \lim_{n\to\infty} \sfrac{1}{n} 
\log P^I_{p,\sigma}(\|C\|=n) , \q
\zeta^I(p,\sigma) \;=\; \lim_{n\to\infty} \sfrac{1}{n} 
\log P^I_{p,\sigma}(|C|=n) .
\end{equation}
Using our knowledge of the percolation transition, we show 
that the lattice animal free energy is non-analytic
on certain curves. 

We coded the numerical algorithm of Newman and Ziff \cite{NZ00}  
for the inhomogeneous percolation model and collected data to
determine the location of the critical curve $\sigma^*(p)$
for low dimensions.  We present some results in section \ref{sectionNum}, 
including data on the case with the bulk density fixed at density $p$
near $p_c(d)$, where we obtain estimates of the curve $\sigma^*$, 
consistent with reference \cite{CCRS81}.

We conclude the paper in section \ref{sectionConc} with a summary
and some final remarks on the model.

\section{The Phase Boundaries}
\label{section2}     %%ZXZ[section2]

This section proves properties of the critical curve $\sigma=\sigma^*(p)$, which 
is defined in equation (\ref{eqn-sigmastardef}) by
\[
\sigma^*(p) \;=\; \inf%_{\sigma \in [0,1]} 
\{ \sigma \in [0,1]\,\vert \hbox{$\theta^I(p,\sigma) > 0$} \}.
%\label{eqn18}   %%ZXZ[eqn18]
\]

Let $P_{p,\sigma}^I$ and $\theta^I(p,\sigma)$ be defined as in equations
\Ref{eqnZ88} and \Ref{eqnZ99}, with the homogeneous analogues
 $P_p^H$ and $\theta^H(p)$ as defined in section 1.1.
Observe that $P^I_{p,p}=P^H_p$,  $\theta^I(p,p)=\theta^H(p)\equiv \theta^H_d(p)$ and
$\theta^I(0,\sigma)= \theta^H_s(\sigma)$.

The following proposition verifies the basic structure of 
Figure \ref{Diagram}.  

\begin{proposition}
   \label{prop-sigstar}
The critical curve $\sigma=\sigma^*(p)$ satisfies
\begin{itemize}
\item[(a)] $\sigma^*(0)\;=\;0$ for $p\,>\,p_c(d)$, and 
\item[(b)]$0\,<\,p_c(d)\,\leq\,\sigma^*(p) \,\leq \, p_c(s)\,=\,\sigma^*(0)$ if $0\,\leq \, p\,< \,p_c(d)$.
\end{itemize}
In particular, $\sigma^*$ is discontinuous at $p=p_c(d)$.
\end{proposition}

It is hard to say much about the value of $\sigma^*(p)$ 
at $p=p_c(d)$; see reference \cite{NW97}.

\begin{proof}
(a) 
Consider percolation in the $d$-dimensional half-space 
$\Lattice_+\subset \Lattice$ with vertices $\IntN^{d-1} 
\times \NatN$ (where $\NatN = \{1,2,3\ldots\}$) and the corresponding 
edges of $\Edges$.  
  
It is known that the critical density for homogeneous 
percolation  in $\Lattice_+$ is equal to $p_c(d)$ \cite{BGN91}.  
Thus, if $C_1$ is the connected component
containing the vertex  $(0,\ldots,0,1)$ in the subgraph 
induced by the open edges of $\Lattice_+$, then 
$P^H_p(|C_1|{=}\infty)\,>\,0$ for $p>p_c(d)$.  Moreover, 
by also considering the status of the single edge from the 
origin to $(0,\ldots,0,1)$, we have $P_{p,0}^I(|C(0)|{=}\infty) 
\geq p\,P^H_p(|C_1|{=}\infty)$.  Hence $\theta^I(p,0)\,>\,0$
if $p\,>\,p_c(d)$.  Part (a) follows.

\smallskip
\noindent
(b) Fix $p\,<\,p_c(d)$.  If $\sigma<p_c(d)$, then by monotonicity,
$\theta^I(p,\sigma) \,\leq \,\theta^H(\max\{p,\sigma\}) =0$ [since 
$\max\{p,\sigma\}\,<\,p_c(d)$].  This shows that $\sigma^*(p)\,\geq \,p_c(d)$.
We obtain $\sigma^*(p)\leq \sigma^*(0)=p_c(s)$ from equation (\ref{eqn-sigma0pcs})
and the fact that $\sigma^*$ is non-increasing in $p$.
%Also, if $\sigma>p_c(s)$, then $\theta^I(p,\sigma)\,\geq \,\theta^I(0,\sigma) =
%\theta^H_s(\sigma)\,>\,0$.  This shows that $\sigma^*(p)\leq p_c(s)$.
\end{proof}

The phase boundary $\sigma^*(p)$ may be estimated for small $p$
in a mean field approximation using the approach in reference \cite{CCRS81}.
Consider percolation in the defect lattice $\Lattice_0$, which
has critical density $\sigma_c = p_c(s)$.  An infinite cluster
can grow in $\Lattice_0$ either along edges $x{\sim}y \in 
\Lattice_0$, or if such an edge is closed, then along a ``bridge"
of three edges in $\Lattice\setminus\Lattice_0$ in a $\sqcap$-shape,
and joining $x$ to $y$.  That is, a bridge of $x{\sim}y$ is
a sequence of three edges $x{\sim}r{\sim}t{\sim}y$ with
$r,t\in\Lattice\setminus\Lattice_0$.

The probability that $x{\sim}y$ is open is $\sigma$, and the
probability that a particular bridge of $x{\sim}y$ is open is
$p^3$.  Since $x{\sim}y$ is bridged by $2(d-s)$ bridges, the 
probability that at least one of them is open is
$1-(1-p^3)^{2(d-s)}$.   Hence, the probability that either
$x{\sim}y$ is open, or that it is closed and at least one of its bridges is
open, is given by $\sigma+(1-\sigma)(1-(1-p^3)^{2(d-s)})$.

An approximation is made by assuming that bridges of different
edges in $\Lattice_0$ are open or closed independently.  In this
approximation a cluster will grow to infinity along $\Lattice_0$
using bridges if the density of open edges or closed edges with
an open bridge is greater than $\sigma_c$,  i.e.\  if
\begin{equation}
\sigma_c \;<\; \sigma + (1-\sigma)\L 1-(1-p^3)^{2(d-s)}\R .
\end{equation}
Solving for $\sigma$ gives an estimate of $\sigma^*(p)$ for small $p$ :
\begin{equation}
\hspace{0mm}
\sigma^*(p) \;\simeq\; \frac{\sigma_c + (1-p^3)^{2(d-s)}-1}{(1-p^3)^{2(d-s)}}
\;= \;\sigma_c - 2(d-s)(1-\sigma_c)p^3 + O(p^6)
\label{eqnsigapp}
\end{equation}
(the approximation of \cite{CCRS81} is different only because they consider
a half-space with  $\Lattice_0$ being the boundary plane).
This result should be
good for small values of $p$ in particular, because the 
assumption that bridges are independent is better at low densities
of edges in the bulk lattice.

\subsection{Strict monotonicity of $\sigma^*(p)$}

The next proposition serves two purposes.  Firstly, it shows that $\sigma^*$
is a strictly decreasing function for $p\in [0,p_c(d)]$.   Secondly, it proves that
the cubic form of the mean-field approximation of equation \Ref{eqnsigapp}
(see reference \cite{CCRS81})
is a rigorous upper bound for $\sigma^*(p)$ when $p$ is close to 0.

\begin{proposition}
   \label{prop.cubic}
Assume $2\leq s<d$.
Fix $0\leq p<p_c(d)$.  Then there is a positive constant $A$ (possibly depending
on $p$) such that 
\[   \sigma^*(p+\delta) \;\leq\; \sigma^*(p)\,-\,A \,\delta^3   \hspace{5mm}
   \hbox{for all sufficiently small positive $\delta$.}
\]
In particular, $\sigma^*$ is strictly decreasing on $[0,p_c(d)]$.
\end{proposition}

\begin{proof}  Consider a modified lattice obtained by adding some new edges 
to $\Lattice$ as follows (see figure \ref{Diagram3}).
For each vertex ${u}\in \Lattice_0$, let $\tilde{b}_0[{u}]$ be the edge of $\Edges$
joining ${u}$ to ${u}+\vec{e}_d$.  Introduce $2s$ new edges,
$\tilde{b}_1[{u}],\ldots,\tilde{b}_{2s}[{u}]$, each joining 
${u}$ to ${u}+\vec{e}_d$; thus we have $2s+1$ parallel edges 
joining these two vertices.   For each edge $\beta$ in $\Edges_0$, let 
$\beta^+ \,=\,\beta+\vec{e}_d$ (i.e., the edge in $\Edges$ obtained by translating
$\beta$ one unit in the $d^{th}$ coordinate direction).  Also, let $\tilde{\beta}^+$
be a new edge parallel to $\beta^+$ (i.e., having the same endpoints).
Let $\tilde{\Lattice}$ be the (inhomogeneous) lattice with sites $\IntN^d$ and edges
\[    \Edges \,\cup\, \{ \tilde{b}_j[{u}] \,:\, {u}\in \Lattice_0, 1\leq j\leq 2s\}  \,\cup\,
   \{ \tilde{\beta}^+ \,:\,\beta\in \Edges_0\}  \,.\]

%%%%%%%%%%%%%%%%%%%%%%
\begin{figure}[t!]
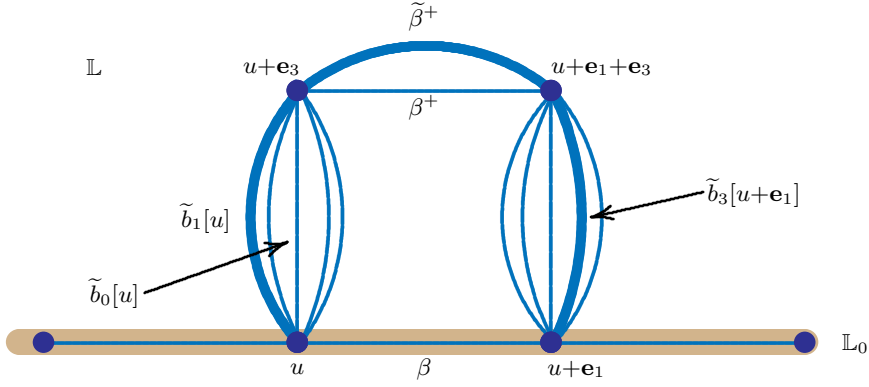

\centering\hfil
\input figure2P.tex
\caption{$\Lattice$ is modified in the proof of proposition
\ref{prop.cubic} by adding edges next to edges $\beta\in \Lattice_0$ as 
shown above for the case $d=3$, $s=2$ and $J=1$. The set $\C{S}(\beta)$
consists of the three bold edges.}
\label{Diagram3}    %%ZXZ[Diagram3]
\end{figure}
%%%%%%%%%%%%%%%%%%%%

Let $P^*_{p,\sigma,t}$ be the probability measure for bond percolation on $\tilde{\Lattice}$ 
with parameters $p,\sigma,t\in [0,1]$ so that edges are independent and 
\[    
%\fl  \hspace{5mm}
P^*_{p,\sigma,t}(\hbox{$e$ is open}) \;=\;  \cases{
p  & \hbox{if $e\in \Edges\setminus \Edges_0$}   \cr
\sigma & \hbox{if $e\in \Edges_0$}   \cr
t & \hbox{if $e=\tilde{\beta}^+$ for $\beta\in \Edges_0$}  \cr
1-(1-t)^{1/2s}  & \hbox{if $e=\tilde{\beta}_j[{u}], {u}\in \Lattice_0, 1\leq j\leq 2s$.}
}
\]

With each edge $\beta\in \Edges_0$, we associate a set $\C{S}(\beta)$ of three
edges in $\tilde{\Lattice}$ as follows.  Let ${u}\in\Lattice_0$ and $J\in\{1,\ldots,s\}$
be such that $\beta$ has endpoints ${u}$ and ${u}+\vec{e}_J$.   Then define
\[    \C{S}(\beta) \;=\;  \{  \tilde{\beta}^+, \,\tilde{b}_J[{u}],\, \tilde{b}_{J+s}[{u}+\vec{e}_J]\} \,.
\]
Thus the edges of $\C{S}(\beta)$ form a three-step path from one endpoint of 
$\beta$ to the other.  Note that for two distinct edges $\beta_1$ and $\beta_2$ in 
$\Edges_0$, the sets $\C{S}(\beta_1)$ and $\C{S}(\beta_2)$ are disjoint.  For 
each $\beta\in\Edges_0$, define the random variable
$Y(\beta)$ to be 1 if either $\beta$ is open, or if all three edges in $\C{S}(\beta)$ 
are open; and define $Y(\beta)$ to be 0 otherwise.   Then
\begin{eqnarray*}
   P^*_{p,\sigma,t}(Y(\beta)=1) & \;=\; &   \sigma \,+\, (1-\sigma)t(1-(1-t)^{1/2s})^2  \\
        & \;\geq\; & \sigma \,+\, \frac{(1-\sigma)t^3}{(2s)^2}   
\end{eqnarray*}
where we used the fact that $tw\leq 1-(1-t)^w$ for $t,w\in (0,1)$.  Since $\C{S}(\beta)$ is 
disjoint from $\Edges\setminus\Edges_0$, it follows that
\begin{equation}
    \label{eq-starI}
      P^*_{p,\sigma,t}(|C|{=}\infty) 
\;\geq \;   \theta^I\left(p,\sigma+\frac{(1-\sigma)t^3}{(2s)^2} \right)\,.
\end{equation}
Next, given a small positive $\delta$, let $t=\delta/(1-p)$, so that $1-(1-p)(1-t) =p+\delta$.
Then for each ${u}\in \Lattice_0$, we have
\begin{eqnarray}
 % \fl    \hspace{8mm}
  \nonumber
&  P^*_{p,\sigma,t}&(\hbox{ $\tilde{b}_j[{u}]$ is empty for all $j=0,\ldots,2s$ }) \\
& & \;=\;  (1-p)\left( (1-t)^{1/2s}\right)^{2s}   \nonumber  \\
        \label{eq.Pstar}
& &\;=\; (1-p) \,(1-t)  \;\;\;=\;\;\; 1-(p+\delta)\,. 
\end{eqnarray}
In addition,
\begin{equation}
   \label{eq.Pstar2}
   %\fl    \hspace{9mm}
      P^*_{p,\sigma,t}(\hbox{ $\beta^+$ or $\tilde{\beta}^+$ is occupied }) \;=\; 
        1 \,-\,(1-p)(1-t) \;=\;p+\delta\,.
\end{equation}
Combining equations \Ref{eq.Pstar} and \Ref{eq.Pstar2} shows that
percolation on 
$\tilde{\Lattice}$ governed by $P^*_{p,\sigma,t} $ is essentially the same as 
percolation on $\Lattice$ governed by a modification of $P_{p,\sigma}$ in which
some edges of $\Edges\setminus \Edges_0$ have their density raised from $p$
to $p+\delta$.  This shows that
\begin{equation}
    \label{eq-starH}
    \theta^I\left( p+\delta,\sigma \right)  
\;\geq \;  P^*_{p,\sigma,t}(|C|{=}\infty) \;\,.
\end{equation}
Combining equations \Ref{eq-starI} and \Ref{eq-starH} shows that 
$\theta^I(p+\delta,\sigma)\,>\,0$
if $\sigma+(1-\sigma)t^3/(2s)^2\,>\sigma^*(p)$, which can be rearranged to yield the inequality
\begin{equation} 
   \label{eq.sigbd}
     \sigma  \; >  \;   \frac{\sigma^*(p)-t^3/(2s)^{2}}{1-t^3/(2s)^2}   \,.
\end{equation}
Therefore the right-hand side of equation \Ref{eq.sigbd} is an upper bound for 
$\sigma^*(p+\delta)$.  The proposition follows (using the fact that $\sigma^*(p)<1$).
\end{proof}

\begin{corollary}
  \label{cor-ltpcs}
$\sigma^*(p_c(d))\,<\,p_c(s)$.
\end{corollary}
\begin{proof}
By the strict monotonicity of $\sigma^*$ (proposition \ref{prop.cubic}) and 
equation (\ref{eqn-sigma0pcs}), we have $\sigma^*(p_c(d))\,<\,\sigma^*(0)\,=\,p_c(s)$.
\end{proof}

\section{Uniqueness of the critical point}
\label{section3}     %%ZXZ[section3]

It follows trivially from equation \Ref{eqnZ14} that $\theta^I(p,\sigma)>0$
implies that $\chi^I(p,\sigma) = \infty$.  
To show a converse is more difficult. 
We would like to prove that the percolation transition is at the same place as
the transition from finite to infinite susceptibility.  Such an assertion, often
called the ``uniqueness of the critical point", is the subject of the 
following theorem, whose proof is the goal of this section.

%%%%%%%%%%%%%%%%
\begin{theorem}
Suppose 
% that $\Delta > 0$ is small and that $(p_1,\sigma_1)$ is such
that $\chi^{I}(p_1,\sigma_1) = \infty$.  Then either
\begin{itemize}
\item[a)] $\theta^I(p_1,\sigma_1) > 0$, or
\item[b)] $\theta^I(p_1,\sigma_1) = 0$ and
$\theta^I(p_1+\Delta,\sigma_1+\Delta) > 0$ for all small positive $\Delta$; in
particular, $(p_1,\sigma_1)$ is a boundary point of $\C{R}_{0}$ .
\end{itemize}
\label{theorem5}    %%ZXZ[theorem5]
\end{theorem}

By proposition \ref{prop-sigstar}(a), theorem \ref{theorem5}
holds whenever $p_1\geq p_c(d)$.  
If $p_1\geq \sigma_1$, then $\chi^{I}(p_1,\sigma_1) = \infty$ implies that
$\chi^H(p_1)=\infty$ (by monotonicity), which in turn 
implies that $p_1\geq p_c(d)$ by our knowledge of the 
homogeneous case.  
Thus theorem \ref{theorem5} holds whenever $p_1\geq \sigma_1$.

Consequently, for the rest of this section we 
shall assume that $p_1<p_c(d)$, $p_1<\sigma_1$, and $\chi^{I}(p_1,\sigma_1) = \infty$.  
If we also have $\chi^{f,I}(p_1,\sigma_1) < \infty$, then (by comparing 
equations (\ref{eqnZ14}) and (\ref{eqnZ15})),
we must have $\theta^I(p_1,\sigma_1)>0$, and we are done.  Therefore
we shall also assume that $\chi^{f,I}(p_1,\sigma_1) = \infty$.

The proof for the homogeneous case in reference \cite{AB87} (see also reference \cite{G99})
relies on augmenting the model to include a ghost vertex $g$.  
We follow a similar approach in the inhomogeneous case.

Thus we introduce the ghost vertex $g$ and edges $\Edges_g = 
\{g{\sim} x\vert \,x\in \Lattice\}$. 
Edges in $\Edges_g$ are open with probability $\gamma\in (0,1)$.
Define $G$ to be the (random) collection of vertices in $\Lattice$
adjacent to $g$ through open edges in $\Edges_g$.  

Percolation on $\Edges \cup \Edges_g$ has parameters
$(p,\sigma,\gamma)$. Since edges in $\Edges_g$ are open with 
probability $\gamma$, it follows that
\begin{equation}
\theta^I (p,\sigma,\gamma) \;=\; 1 - \sum_{n=1}^\infty (1-\gamma)^n\,
P^I_{p,\sigma}(|C| {=} n) 
\label{eqn34}    %%ZXZ[eqn34]
\end{equation}
is the probability that there is an open path from the origin
to the ghost vertex $g$.  Observe that by Abel's theorem
\begin{eqnarray*}
\lim_{\gamma\to 0^+} \theta^I (p,\sigma,\gamma)
&\;=\; 1 - \sum_{n=1}^\infty P^I_{p,\sigma}(|C| {=} n) \\
&\;=\; 1 - P^I_{p,\sigma}(|C|{<}\infty) 
\;=\; P^I_{p,\sigma}(|C|{=}\infty) = \theta^I(p,\sigma). 
\end{eqnarray*}
Similarly, it is the case that
\begin{equation}
\chi^I (p,\sigma,\gamma) \;:=\; \sum_{n=1}^\infty n\,(1-\gamma)^{n} 
P^I_{p,\sigma} (|C|{=}n) \;=\; (1-\gamma) \,
 \frac{\partial \theta^I}{\partial \gamma}
\label{eqn35}    %%ZXZ[eqn35]
\end{equation}
for $\gamma\in(0,1)$. We also have
\begin{equation}
\lim_{\gamma \to 0^+} \chi^I (p,\sigma,\gamma) \;=\; \chi^{f,I}(p,\sigma) .
\label{eqn35a}
\end{equation}
If $p < \sigma$, then theorem \ref{theoremA9} in the Appendix shows that
the functions $\theta^I(p,\sigma,\gamma)$ 
and $\chi^I(p,\sigma,\gamma)$ satisfy the differential inequalities 
\begin{eqnarray}
%\fl  \hspace{2mm}
\vec{q} \cdot \nabla \theta^I (p,\sigma,\gamma)
& \;\leq\; &2d\, \chi^H(p)\,\theta^I(p,\sigma,\gamma) 
(1-\gamma) \frac{\partial \theta^I}{\partial \gamma} ,
\label{eqn37}    %%ZXZ[eqn37]
\\
%
%\fl  
\theta^I(p,\sigma,\gamma) 
&\;\leq\; &\gamma\, \frac{\partial\theta^I}{\partial \gamma}
+  \L\theta^I(p,\sigma,\gamma)\R^2 \nonumber \\
& & \qq + \chi^H(p)\, \theta^I(p,\sigma,\gamma)
\L \vec{p}\cdot \nabla \theta^I(p,\sigma,\gamma) \R  
\label{eqn38}    %%ZXZ[eqn38]
\end{eqnarray}
where $\nabla = (\sfrac{\partial}{\partial p},
\sfrac{\partial}{\partial \sigma})$, $\vec{p}=(p,\sigma)$ 
and $\vec{q} = (1-p,1-\sigma)$.  

\begin{theorem}
Assume that $p< p_c(d)$, $p\leq \sigma$, and  $\chi^{f,I} (p,\sigma) = \infty$.
Then there exists an  $\alpha=\alpha(p, \sigma)$ such that
\[ \theta^I(p,\sigma,\gamma) \;\geq\; \alpha\, \gamma^{1/2} \]
for all small positive values of $\gamma$.
\label{thm3}    %%ZXZ[thm3]
\end{theorem}

\begin{proof}
Suppose that $p<p_c(d)$,  $0 <p \leq \sigma < 1$, and
$\chi^{f,I}(p,\sigma) = \infty$.
If $\theta^I(p,\sigma) >0$ then the theorem is trivial.
Thus, assume that $\theta^I(p,\sigma) = 0$.  This implies that 
%$\theta^I(p,\sigma,\gamma) \geq 0$ and
$\lim_{\gamma\to 0^+} \theta^I(p,\sigma,\gamma) =0$.

%Observe that $\theta^I(p,\sigma,\gamma)$ is strictly increasing  with 
%$\gamma$ and continuously differentiable with respect to $\gamma$ on $(0,1)$.

Observe that since $p\leq \sigma$ and $1-p\geq 1-\sigma$,
\begin{eqnarray}
   \vec{p}\cdot \nabla \theta^I
   \;=\; p\frac{\partial \theta^I}{\partial p} + 
      \sigma\frac{\partial \theta^I}{\partial \sigma}\;&\leq &\; \frac{\sigma}{1-\sigma} 
    \L (1-p)\frac{\partial\theta^I}{\partial p} + 
      (1-\sigma)\frac{\partial \theta^I}{\partial \sigma} \R \nonumber \\
   \;&=&\; \frac{\sigma}{1-\sigma}\, \vec{q}\cdot \nabla \theta^I .
   \label{eq.pdottheta}
\end{eqnarray}
With $p$ and $\sigma$ fixed, put  $\theta^I(p,\sigma,\gamma) = f(\gamma)$.  The properties of 
$f(\gamma)$ are such that $f$ is strictly increasing and continuously differentiable 
on $(0,1)$ with $f(0)=0$ and $f(1) = 1$.  
Using equation \Ref{eq.pdottheta} to eliminate $\nabla \theta^I$ from the
differential inequalities \Ref{eqn37} and \Ref{eqn38}, we obtain
\begin{equation}
%\fl   \hspace{15mm}
 f(\gamma) 
  \;\leq\; \gamma\, \frac{\partial f(\gamma)}{\partial \gamma}
+  f^2(\gamma)
+ \frac{2d\,\sigma}{1-\sigma}
\LH \chi^H(p) f(\gamma) \RH^2 (1-\gamma) 
\frac{\partial f(\gamma)}{\partial \gamma} \,.
\label{star4}
\end{equation} 
By the mean value theorem there exists a $\psi\in (0,\gamma)$ such that 
\[ f^\prime(\psi) \;=\; \frac{1}{\gamma} f(\gamma) .\]
As $\gamma\to 0^+$, $\psi\to 0^+$, so that by equations
\Ref{eqn35} and \Ref{eqn35a}, 
\[  \lim_{\gamma \to 0^+} \frac{\gamma}{f(\gamma)} \;=\; 0. \]
Define the inverse function of $f$ to be $h$. Then $h$ is
strictly increasing and continuously differentiable with
$h(0)=0$ and $h(1)=1$, and satisfying
\begin{equation}
\lim_{\phi \to 0^+} \frac{h(\phi)}{\phi} \;=\; 0  . 
\label{starstar4}
\end{equation} 
This shows that $h^\prime(\phi)$ is bounded on $(0,\Phi]$ 
for some $\Phi>0$ . Also, note that $\sfrac{d\,h}{d\phi} 
= \L \sfrac{d\,f}{d\gamma} \R^{-1}$.

By substituting $\gamma = h(\phi)$ and $f(\gamma)=\phi$ in 
equation \Ref{star4} and simplifying, we get
\begin{equation}
   \label{ineq.dhdphi}
  \frac{1}{\phi} \, \frac{d\,h}{d\phi} - \frac{h}{\phi^2}
   \;\leq\; \frac{2d\,\sigma}{1-\sigma} \L \chi^H (p) \R^2
    \L 1- h\R + \frac{d\,h}{d\phi} .
\end{equation}
Observe that $h$ is a strictly increasing function with bounds 
$0 \leq h(\phi) \leq 1$ where
$h(0)=0$ and $h(1)=1$.  Furthermore, since $h^\prime (\phi)$ is bounded
on $(0,\Phi]$, equation (\ref{ineq.dhdphi}) implies that 
there exists a $\beta(p,\sigma)>0$ such that
\[ 
\frac{1}{\phi} \, \frac{d\,h}{d\phi} - \frac{h}{\phi^2}
\;\leq\; \beta,\q\hbox{if $0 < \phi \leq \Phi$.} \]
Integrate this over $\phi \in (0,x)$ where $x \leq \Phi$ to get
\[ \frac{h(\phi)}{\phi} \Vert_0^x \;\leq\; \beta x . \]
By equation \Ref{starstar4}, this gives $h(x) \leq \beta x^2$ for all
$x\in (0,\Phi]$.  In terms of $f$ this says $\gamma \leq \beta\,
\L f(\gamma) \R^2$ for $\gamma \leq h(\Phi)$, or
$f(\gamma) \geq \alpha \gamma^{1/2}$ where
$\alpha = \beta^{-1/2}$.
\end{proof}

We shall now complete the proof of theorem \ref{theorem5} in the remaining 
situation of interest, namely 
that $p_1<p_c(d)$, $p_1<\sigma_1$ and $\chi^{f,I}(p_1,\sigma_1) = \infty$.
Let $\theta^I \equiv \theta^I(p,\sigma,\gamma)$ and $\kappa=\kappa(p) = \chi^H(p)$.
Write equation \Ref{eqn38} as
\begin{equation}
0 \;\leq\; \frac{1}{\theta^I} \,\frac{\partial \theta^I}{\partial\gamma}
+ \frac{1}{\gamma} \L
\theta^I - 1 + \kappa \, \vec{p} \cdot \nabla\theta^I \R .
\label{star5}
\end{equation}
If $\theta^I (p_1,\sigma_1) > 0$
then the proof is done, so assume that $\theta^I (p_1,\sigma_1) =0$.
Let $\Delta>0$ be small, and 
define $\O{p} = p_1 + \Delta/2$ and $\O{\sigma} = \sigma_1 + \Delta/2$.
Put $p\,=\,p(t)\, =\, p_1+t/2$ and $\sigma =\sigma(t)=\sigma_1+t/2$, and define
$\vec{p}(t) = \L p(t),\sigma(t) \R$.

%%%%%%%%%%%%%%%%%%%%%%
\begin{figure}[t!]
\centering\hfill
\input figure9.tex
\caption{Integrating equation \Ref{sstar5} for $t\in(0,\Delta)$ takes
$(p_1,\sigma_1)$ to $(\O{p},\O{\sigma})$.}
\end{figure}
%%%%%%%%%%%%%%%%%%%%

We shall follow the method as presented in reference \cite{G99} to show that
$\theta^I (\O{p},\O{\sigma}) > 0 $.
We begin with the following claim.  (Notice that $\kappa = \chi^H(p(t))$ and 
$\theta^I=\theta^I(p(t),\sigma(t))$ now depend on $t$.)
\smallskip
\Claim 
$\displaystyle
\theta^I - 1 + \kappa\, \vec{p}(t)\cdot \nabla \theta^I
\; \leq \;\frac{d}{dt} \L (p(t)+\sigma(t))(2\kappa\,\theta^I -1) \R$.

\Claimproof
\begin{eqnarray*}
%\fl
& & \sfrac{d}{dt} \L (p(t)+\sigma(t))(2\kappa\,\theta^I -1) \R
 \;=\; 2\kappa\,\theta^I - 1
    + (p(t)+\sigma(t)) \sfrac{d}{dt} \L 2\kappa\,\theta^I \R  \\
 & \;=\; & 2\kappa\,\theta^I  - 1 + (p(t)+\sigma(t))
    \L \kappa^\prime \theta^I 
     + \kappa \L \sfrac{\partial}{\partial p}\theta^I
                 + \sfrac{\partial}{\partial \sigma} \theta^I\R \R  \\
& \;\geq\; & \theta^I - 1 + \kappa (p(t)+\sigma(t))
  \L \sfrac{\partial}{\partial p}\theta^I
   + \sfrac{\partial}{\partial \sigma} \theta^I\R  \\
& \;\geq\; & \theta^I - 1 + \kappa \, \vec{p}(t) \cdot \nabla \theta^I
\end{eqnarray*}
since $\kappa\geq 1$, $\kappa^\prime \geq 0$ and $\theta^I
\in [0,1]$ and is non-decreasing with $p$ and $\sigma$.  
This completes the proof of the claim. 

By the claim, equation \Ref{star5} implies
\begin{equation}
0 \;\leq\; \frac{1}{\theta^I} \,\frac{\partial \theta^I}{\partial\gamma}
+ \frac{1}{\gamma} \, \frac{d}{dt} 
\L (p(t)+\sigma(t))(2\kappa\,\theta^I -1) \R .
\label{sstar5}
\end{equation}
Integrate the first term with respect to $\gamma\in(\delta, \eps)$ to obtain
\[ \int_\delta^\eps  \frac{1}{\theta^I} 
\,\frac{\partial \theta^I}{\partial\gamma} \, d\gamma
\;=\; \log \theta^I(p(t),\sigma(t),\eps) 
      - \log \theta^I(p(t),\sigma(t),\delta).\]
Now integrate the result with respect to $t\in (0,\Delta)$, and use the bounds 
\[ \theta^I(p_1,\sigma_1,\gamma) 
\;\leq\; \theta^I(p(t),\sigma(t),\gamma) 
\;\leq\; \theta^I(\O{p},\O{\sigma},\gamma) \]
on $\theta^I$ to get
\begin{equation}
 \int_0^\Delta \LH \int_\delta^\eps  \frac{1}{\theta^I} 
\,\frac{\partial \theta^I}{\partial\gamma} \, d\gamma \RH dt
\;\leq\; \Delta \L \log \theta^I(\O{p},\O{\sigma},\eps) 
      - \log \theta^I(p_1,\sigma_1,\delta)\R .
\label{ssstar5}
\end{equation}

We integrate and bound the second term in equation \Ref{sstar5} similarly, as follows.
The integral with respect to $t$ is straightforward, and integrating the
result with respect to $\gamma$ gives
\begin{eqnarray*}
%\fl  \hspace{5mm}
\int_\delta^\eps \LH (\O{p} + \O{\sigma}) 
                 \L 2\kappa(\O{p}) \theta^I(\O{p},\O{\sigma},\gamma)-1 \R
                 -  ({p_1} + {\sigma_1}) 
                 \L 2\kappa(p_1)\, \theta^I({p_1},{\sigma_1},\gamma)-1 \R \RH
                 \frac{d\gamma}{\gamma}& &\\
\; \leq \; \int_\delta^\eps \LH (\O{p} + \O{\sigma}) 
                 \L 2\kappa(\O{p}) \theta^I(\O{p},\O{\sigma},\eps)-1 \R
                 +  ({p_1} + {\sigma_1}) \RH 
                 \frac{d\gamma}{\gamma}& &\\
\;=\; \log(\eps/\delta) \L 2(\O{p} + \O{\sigma}) \,
         \kappa(\O{p}) \theta^I(\O{p},\O{\sigma},\eps)  - \Delta \R & &
\end{eqnarray*}
%where $\kappa = \kappa_{\O{p}}$ and 
since $(\O{p}+\O{\sigma})-(p_1+\sigma_1) = \Delta$.
Using this and equation \Ref{ssstar5}, the inequality in
\Ref{sstar5} becomes
\[ %\fl \hspace{5mm}
 0 \;\leq \;
\frac{\Delta \L \log \theta^I(\O{p},\O{\sigma},\eps) 
      - \log \theta^I(p_1,\sigma_1,\delta)\R}{\log \eps - \log\delta} 
       +  2(\O{p} + \O{\sigma}) \,
         \kappa(\O{p})\, \theta^I(\O{p},\O{\sigma},\eps)  - \Delta .\]
By theorem \ref{thm3}, $\theta^I (p_1,\sigma_1,\gamma) 
\geq \alpha \gamma^{1/2}$.  Hence $-\log \theta^I(p_1,\sigma_1,\delta)
\;\leq\; - \log \alpha - \sfrac{1}{2} \log \delta$ for small $\delta$.  Substituting this into the
above and taking $\delta \to 0^+$ gives
\[ 0 \;\leq \;\frac{\Delta}{2} +  2(\O{p} + \O{\sigma}) 
        \, \kappa(\O{p})\, \theta^I(\O{p},\O{\sigma},\eps)  - \Delta \]
which can be rearranged to give
\[ \frac{\Delta}{2} \;\leq\; 4 \kappa(\O{p})\, \theta^I(\O{p},\O{\sigma},\eps) . \]
Taking $\eps \to 0^+$ completes the proof.

\section{Exponential decay of the cluster size distribution}
\label{section4}     %%ZXZ[section4]

In this section the exponential decay of $P_{p,\sigma}^I(|C|=n)$ in the
subcritical phase is examined (corresponding to region $\C{R}_0$ in 
figure \ref{Diagram}).  Proving this follows the same general outline 
as for the similar result in homogeneous percolation, with 
minor modifications.

The two-point connectivity function is defined by
\begin{equation}
\tau^I_{p,\sigma} (x,y) 
\;=\; E^I_{p,\sigma} \indic_{ \{x\leftrightarrow y\} }
\label{eqnZ39}    %%ZXZ[eqnZ39]
\end{equation}
where $\indic_{\{x\leftrightarrow y\}}$ is the indicator function of 
the event that there exists an open path joining vertices $x$ and $y$
in $\IntN^d$ (or the indicator function of the event that $x$ and $y$ 
belong to the same open cluster).  

Naturally, the number of vertices in the cluster at the origin is
$|C| = \sum_x \indic_{\{0\leftrightarrow x\}}$ and so the
susceptibility defined in equation \Ref{eqnZ14} may be expressed
in terms of the two point connectivity function via
\[
%\fl   \hspace{5mm}
\chi^I(p,\sigma) \;=\;
E^I_{p,\sigma} |C| \;=\; E^I_{p,\sigma} \sum_x 
\indic_{\{0\leftrightarrow x\}} 
\;=\; \sum_x E^I_{p,\sigma} \indic_{\{0\leftrightarrow x\}} 
\;=\; \sum_x \tau^I_{p,\sigma} (0,x)\,.
\]

More generally, consider the open cluster $C(y)$ 
%with size $|C(y)|$ vertices 
at the site $y$ and define
\begin{eqnarray*}
%\fl
\chi^I(p,\sigma;y) 
\;=\; E^I_{p,\sigma} | C(y) | & \;=\; & E^I_{p,\sigma} \sum_x 
\indic_{\{y\leftrightarrow x\}}  \\
&\;=\;& \sum_x E^I_{p,\sigma} \indic_{\{y\leftrightarrow x\}} 
\;=\; \sum_x \tau^I_{p,\sigma} (y,x) .
\end{eqnarray*}

By lemma \ref{lemmaA6} and equation \Ref{eqnAchi} in the appendix,
%$\chi^I(p,\sigma;y) \;\leq\; \chi^H(p) 
%\, \chi^I (p,\sigma;0) \;=\; \chi^H(p) \, \chi^I (p,\sigma)$; that is,
\begin{equation}
%\fl  \hspace{15mm}
\sum_x \tau^I_{p,\sigma} (y,x) 
  \; =\;  \chi^I(p,\sigma;y)
  \;\leq \;\chi^H(p) \, \chi^I (p,\sigma;0)   
  \hspace{5mm} \hbox{for every $y\in \IntN^d$}.
\label{eqnZ41}    %%ZXZ[eqnZ41]
\end{equation}

A similar and generalised bound on the $(n{+}1)$-point connectivity 
function $\tau^I_{p,\sigma} (y,x_1,x_2,x_3,\ldots,x_n)$ defined by
\begin{eqnarray}
%\fl  \hspace{6mm}
& & \tau^I_{p,\sigma} (y,x_1,x_2,x_3,\ldots,x_n) \nonumber\\
& & \;=\; E^I_{p,\sigma} \L
\indic_{\{y\leftrightarrow x_1\}} 
\indic_{\{y\leftrightarrow x_2\}}
\indic_{\{y\leftrightarrow x_3\}} \ldots 
\indic_{\{y\leftrightarrow x_n\}}
\R
\end{eqnarray}
should be determined.  This will give an upper bound on 
$E^I_{p,\sigma} |C|^n$ since
\begin{equation}
E^I_{p,\sigma} |C|^n \;=\; \sum_{x_1,x_2,\ldots,x_n}
\tau^I_{p,\sigma} (0,x_1,x_2,x_3,\ldots,x_n) .
\end{equation}

In the case of the three-point connectivity function, the arguments 
given in Chapter 6 of reference \cite{G99} can be extended to 
apply to the inhomogeneous model considered in this paper.  As such, we can 
state the following lemma without proof.

\begin{lemma}
For all values of $p$ and $\sigma$ and vertices $x_0$, $x_1$ and $x_2$,
\[ \tau^I_{p,\sigma}(x_0,x_1,x_2) \;\leq\;
\sum_y \tau^I_{p,\sigma} (y,x_0) \, \tau^I_{p,\sigma} (y,x_1) \,
 \tau^I_{p,\sigma} (y,x_2). \]
Thus, by equation \Ref{eqnZ41}, 
\[ %\fl
E^I_{p,\sigma}  |C|^2  \;\leq\;
\sum_{y,x_1,x_2} \tau^I_{p,\sigma} (y,0) \, \tau^I_{p,\sigma} (y,x_1)\,
\, \tau^I_{p,\sigma} (y,x_2) 
\;\leq\; \L \chi^H(p) \, \chi^I (p,\sigma;0)  \R^3 . \]
\label{lemma5Z}    %%ZXZ[lemma5Z]
\end{lemma}

The generalisation of the above bound for the inhomogeneous model
proceeds along the same line as the argument given by Aizenman and 
Newman \cite{AN84} for the case of homogeneous percolation, 
involving the characterisation of connectivity functions 
as sums over labeled skeletons (trees with all interior vertices of degree 
three) \cite{G99}. 

Following the arguments for the homogeneous case one arrives
at the bound
\begin{equation}
\tau_{p,\sigma}^I(x_0,x_1,\ldots,x_n) \;\leq\;
\sum_S \sum_{\psi_x} \prod_{u{\sim} v \in S}
\tau_{p,\sigma}^I (\psi_x(u),\psi_x(v)) 
\label{eqn44}    %%ZXZ[eqn44]
\end{equation}
in the notation of reference \cite{G99}.  The summation over $S$ is
over all labeled skeletons with $n+1$ exterior vertices 
(or \textit{end vertices}).  The summation over $\psi_x$ is a sum
over all admissible mappings from the vertex set of a skeleton $S$
into $\IntN^d$ (this is a summation over all possible $\psi_x(v)$
as $v$ takes on values in the interior vertices of $S$).  The product
is over all branches $u{\sim} v$ (edges joining adjacent vertices $u$ and $v$
in the graph theoretic sense) of $S$.

Equation \Ref{eqn44} must be summed over $x_j$ for $1\leq j \leq n$
to obtain an upper bound on $E_{p,\sigma}^I |C|^n$.  Since
the $x_i$ are end-vertices in $S$, and are vertices in the two-point
functions, one may use equation \Ref{eqnZ41} to bound these summations
from above.  That reduces equation \Ref{eqn44} to
\[
E_{p,\sigma}^I |C|^n \;\leq\; \L \chi^H(p)\,\chi^I(p,\sigma) \R^n\,
\sum_S \sum_{\psi} {\prod_{u{\sim} v}}^{\hskip 0mm \hbox{\Large$\prime$}} \,
 \tau_{p,\sigma}^I (\psi(u),\psi(v)) 
\]
and the primed product is only over branches $u{\sim} v$ where
$u$ and $v$ are vertices in the skeleton which are either the origin, or 
are interior vertices of $S$.  Performing the
summation over $\psi$ and using equation \Ref{eqnZ41} as a bound
gives
\begin{equation}
E_{p,\sigma}^I |C|^n \;\leq\; N_{n+1}\,
\L \chi^H(p)\chi^I(p,\sigma) \R^{2n-1}
\label{eqn46}     %%ZXZ[eqn46]
\end{equation}
where
\begin{equation}
N_{n+1} \;=\; \frac{(2n-2)!}{2^{n-1}(n-1)!} 
\end{equation}
is the number of labeled skeletons with $n+1$ exterior (or end-)vertices.
This is the generalisation of lemma \ref{lemma5Z}.

The bound in equation \Ref{eqn46} is sufficient for the following theorem.

\begin{theorem}
Suppose that $\chi^I(p,\sigma)<\infty$.
Then for every $n$, 
\begin{equation}
  \label{eq.cltexp}
 P_{p,\sigma} (|C|\geq n) \;\leq\; 2 \, e^{-n/\L 2 \,(\chi^H(p)\,\chi^I(p,\sigma))^2\R}. 
\end{equation}
\label{theoremZ6}     %%ZXZ[theoremZ6]
\end{theorem}

\begin{proof}
The proof follows the approach in Grimmett \cite{G99}.  Use equation
\Ref{eqn46} to see that
\begin{eqnarray*}
%\fl
E_{p,\sigma}^I\L |C|\,e^{t|C|} \R
& \;=\; & \sum_{n=0}^\infty \sfrac{t^n}{n!}\,
          E_{p,\sigma}^I\L |C|^{n+1}\R  \\
&  \;\leq\; & \L \chi^H(p)\,\chi^I(p,\sigma)\R\,
      \LH 1+\sum_{n=1}^\infty \sfrac{t^n}{n!}
          N_{n+2} \L \chi^H(p)\, \chi^I(p,\sigma) \R^{2n}\RH \\ 
& \;=\; & \frac{\chi^H(p)\,\chi^I(p,\sigma)}{
\sqrt{1-2t\L \chi^H(p)\,\chi^I(p,\sigma) \R^2}} 
\end{eqnarray*}
whenever $t\in\LH 0,\sfrac{1}{2} \L \chi^H(p)\,\chi^I(p,\sigma) \R^{-2} \R$.
Markov's inequality \cite{GS01} then shows that
\[ P_{p,\sigma}^I(|C|{\geq} n) 
\;=\; P_{p,\sigma}^I (|C|e^{t|C|} {\geq} n\,e^{tn} ) 
\;\leq\; \frac{1}{n\,e^{tn}} \,
E_{p,\sigma}^I (|C| e^{t|C|} ) . \]
This shows that
\[ P_{p,\sigma}^I(|C|{\geq} n) \;\leq\; 
\frac{\chi^H(p)\,\chi^I(p,\sigma)}{n\,e^{tn} 
\sqrt{1-2t\L \chi^H(p)\,\chi^I(p,\sigma) \R^2}} .\]

The final step in the proof is to choose an appropriate value 
for $t$.  The last inequality is valid for 
$0 \leq t < \sfrac{1}{2}
 \L \chi^H(p)\,\chi^I(p,\sigma) \R^{-2}$, so put
\[ t \;=\; \frac{1}{2\L \chi^H(p)\,\chi^I(p,\sigma)\R^2} - \frac{1}{2n} . \]
If $n> \L \chi^H(p)\,\chi^I(p,\sigma)\R^2$ then $t>0$,
and with this choice of $t$ one gets
\[ P_{p,\sigma}^I(|C|{\geq} n) \;\leq\; 
\sfrac{\sqrt{e}}{\sqrt{n}} \,
e^{- \sfrac{n}{2\L \chi^H(p)\chi^I(p,\sigma) \R^2 }} \]
and equation (\ref{eq.cltexp}) follows.  Finally, (\ref{eq.cltexp}) is 
trivially true for $n\leq  \L \chi^H(p)\,\chi^I(p,\sigma)\R^2$ because $2\,e^{-1/2}>1$.
This completes the proof.
\end{proof}

Since for each cluster $C$ one has $\sfrac{1}{d} \|C\|
\leq | C| \leq \|C\|+1$, it follows that
$P_{p,\sigma}^I(\|C\| {\geq} n)\,\leq \,
 P_{p,\sigma}^I(|C| {\geq} \,n/d )$
and we get the following a corollary of theorem \ref{theoremZ6}:

\begin{corollary}
Suppose that %$p\in (0,p_c(d))$ and $\sigma \in (0,\sigma^*(p))$.  
$(p,\sigma)$ is in the interior of $\C{R}_0$.
Then there exists a function $\lambda^I (p,\sigma) > 0$ such that
\[ P_{p,\sigma}^I(\|C\|{=}n) \;\leq\; P_{p,\sigma}^I(\|C\|{\geq} n)
\;\leq\; e^{-n\,\lambda^I(p,\sigma)} .\]
\label{Cor7X}     %%ZXZ[Cor7X]
\end{corollary}

\section{The supercritical region}
 \label{section5}   

\subsection{Subexponential decay of the supercritical cluster size distribution}
 \label{subsec5.1}   %%ZXZ[section5]

It is a result of homogeneous percolation that $P_p^H(|C|{=}n)$ does
not decay exponentially with $n$ in the supercritical phase.
Instead, it is known that there exists a $\gamma_H(p)>0$ such that
\begin{equation}
P_p^H(|C|{=}n)\;\geq \; e^{-\gamma_H(p)\, n^{(d-1)/d} }\q
\hbox{if $p>p_c(d)$}.
\end{equation}

A similar result [equation (\ref{eqn13AA})] can be shown for the inhomogeneous model 
with $p>p_c(d)$ by considering percolation in the half-space
$\Lattice^+$ \cite{M86}.  We state this as the following theorem
and defer its proof to section \ref{section6-2}.

\begin{theorem}
If $p>p_c(d)$, then there exists a $\gamma_I(p)>0$ such that
\[
 P^I_{p,\sigma} (\infty {>} |C| {\geq}  n) \;\geq \;
  P^I_{p,\sigma} (|C| {=}  n) \;\geq \;
e^{-\gamma_I(p)\, n^{(d-1)/d} }.
\]
\label{theorem5-12}   %%ZXZ[theorem5-12]
\end{theorem}

%This also generalises the lower bound in equation \Ref{eqn13AA} to all of $\C{R}_H$.

In the case that $p< p_c(d)$ and $\theta^I(p,\sigma) >0$, the
decay of the cluster size distribution has a different subexponential lower bound, which 
we state in theorem \ref{thm999}.  We prove it 
using a variation of the method for homogeneous percolation
due to Aizenman, Delyon and Souillard \cite{ADS80}.   

\begin{theorem}
Assume that $0<p<p_c(d)$ and that $\theta^I(p,\sigma)>0$.
Then there exist positive constants $\betas1$ and $\betas2$
(which are functions of $(p,\sigma)$) such that for all sufficiently large $n$,
\begin{equation}
  \label{eq.supsubth1}
   P^I_{p,\sigma} \L \infty {>} |C| {\geq} n \R \;\geq\;
\betas1 \,e^{-\betas2 n^{(s-1)/s} ( \log^2 n )^{d-s}} . 
\end{equation}
\label{thm999}    %%ZXZ[thm999]
\end{theorem}

To prove this result, we shall show that if
$p< p_c(d)$ and $\theta^I(p,\sigma) > 0$, then there exist positive constants
$\alpha_1$, $\alpha_2$  and $\alpha_3$ (depending on $p$ and $\sigma$) such that 
\begin{eqnarray}
   \label{eq.supsubth2}
   %\fl   \hspace{5mm}
    P^I_{p,\sigma} \L\infty {>} |C| \geq \alpha_1\, m^s \R
    &\;\geq\;& \alpha_2\, e^{-\alpha_3m^{s-1}\,h(m)^{d-s}} \nonumber \\
   & &    \hspace{12mm}\hbox{for all sufficiently large $m$,}
\end{eqnarray}
where $h(m) = \lcl \log^2 m \rcl$. 
Theorem \ref{thm999} follows from this by putting $m = (n/\alpha_1)^{1/s}$.
The rest of this subsection is devoted to proving equation \Ref{eq.supsubth2}, with 
$h(m)$ being any function that grows faster than $\log m$ and slower than $m$.

Assume that $p< p_c(d)$ and $\theta^I(p,\sigma) > 0$.

Let $h: \NatN\rightarrow\NatN$ be a specified function satisfying $h(m)=o(m)$ and 
$\log m = o(h(m))$.  For each $m\in \NatN$, 
define the rectangular box $B^*(m)$ centered at the origin in $\Lattice$ by
\begin{eqnarray}
%\fl 
& & B^*(m) \;=\;\left([-m,m]^s \times [-h(m),h(m)]^{d-s}\right) \cap\Lattice \label{eqnA60}
  \\
%\fl 
& \;=\; &\LC {z}\in \IntN^d \Vert \hbox{$|\z{i}| \leq m$ for $i=1,\ldots,s$
 and $|\z{i}| \leq h(m)$ for $i=s{+}1,\ldots,d$}\RC . \nonumber
\end{eqnarray}
We separate the boundary of  $B^*(m)$ into
a vertical part $\partial_{\hbox{vert}} (m)$ and a horizontal part
$\partial_{\hbox{hor}}(m)$:
\begin{eqnarray}
%\fl   \hspace{5mm}
\partial_{\hbox{vert}}(m)
&\;=\;& \LC {v} \in B^*(m) \Vert \hbox{$|v_i| = m$ for at least one $i\leq s$} \RC, 
  \hspace{3mm}\hbox{and} \label{eqnA61} \\
 % \fl   \hspace{5mm}
\partial_{\hbox{hor}}(m)
 &\;=\;& \LC {v} \in B^*(m)\setminus \partial_{\hbox{vert}}(m) 
\Vert \hbox{$|v_i| = h(m)$ for at least one $i>s$} \RC. \nonumber
\end{eqnarray}
Also, let $\partial_e(m)$ be the set of edges outside $B^*(m)$ incident on 
$\partial_{\hbox{vert}} (m)$:
\[
     \partial_e(m) \;=\; \{ x{\sim} y \in \Edges \,:\, x\in \partial_{\hbox{vert}} (m), \,y\not\in B^*(m) \}\,.
\]

Given $v\in\IntN^d$ and $A\subseteq \IntN^d$, denote the event that $v$ is
connected to a point of $A$ by an open path by $\LC v \leftrightarrow A\RC
= \cup_{x\in A} \LC v\leftrightarrow x\RC$. 
For every $v\not\in \Lattice_0$,
\begin{equation}
P_{p,\sigma}^I \LC v \leftrightarrow \Lattice_0\RC \;=\; 
P_p^H \LC v \leftrightarrow \Lattice_0\RC.
\label{eqn57AA}     %%ZXZ[eqn57AA]
\end{equation}
Notice that if $p < p_c(d)$, then there exists a positive constant $\delta_p$ such that
\begin{equation}
P_p^H  \LC v \leftrightarrow \Lattice_0\RC \;\leq\;
e^{-\delta_p \, \hbox{dist}(v,\Lattice_0)}
\label{eqn58AA}     %%ZXZ[eqn58AA]
\end{equation}
where $\hbox{dist}(v,\Lattice_0) = \min\{ \| v - x \| \vert x\in \Lattice_0\}$
(and $\|\cdot\|$ is the Euclidean norm).  This is a consequence of
the exponential decay of the cluster size distribution in homogeneous percolation
if $p < p_c$ (see for example reference \cite{AN84}).

Since $\hbox{dist}(v,\Lattice_0) \geq h(m)$ if $v\in\partial_{\hbox{hor}}(m)$,
the following lemma for the probability of the event $\{ v \leftrightarrow \Lattice_0\}$
follows from equations \Ref{eqn57AA} and  \Ref{eqn58AA}.

\begin{lemma}
If $p<p_c(d)$, then there is a $\delta_p>0$
such that 
\[ P_{p,\sigma}^I \LC v \leftrightarrow \Lattice_0 \RC \;\leq\; e^{-\delta_p h(m) } 
 \hspace{5mm}\hbox{for every $v\in \partial_{\hbox{hor}}(m)$}. \]
\label{lemma888}     %%ZXZ[lemma888]
\end{lemma}

%%%%%%%%%%%%%%%%%%%%%%
\begin{figure}[t!]
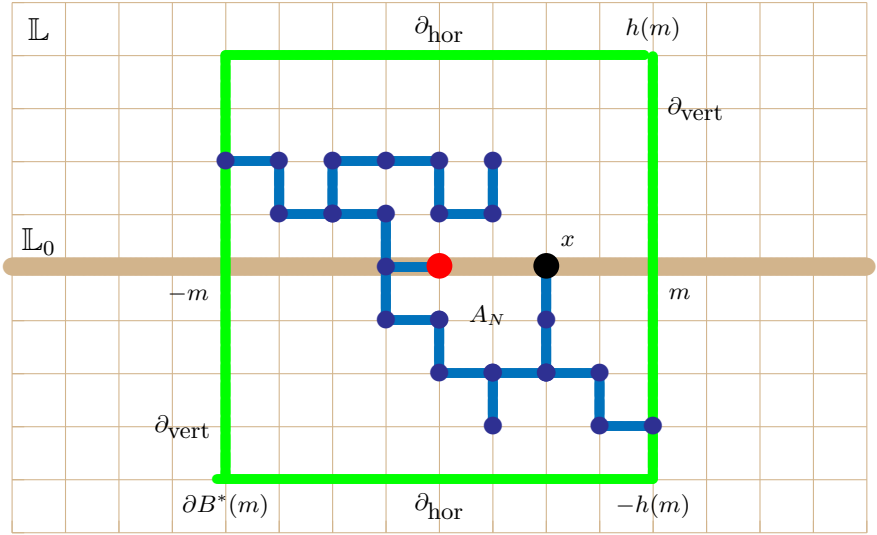

\centering\hfill
\input figure6.tex
\caption{The event $Q_m(x)$ that there is an open path in
$B^*(m)$ from $x\in\Lattice_0$ to $\partial_{\hbox{vert}}(m)$.}
\end{figure}
%%%%%%%%%%%%%%%%%%%%%%

For $x\in \Lattice_0 \cap B^*(m)$, let $Q_m(x)$ be the event that there is an
open path in $B^*(m)$ from $x$ to $\partial_{\hbox{vert}}(m)$.
Let $F_m$ be the event that there is no open path 
% in $B^*(m)\setminus \partial_{\hbox{vert}} (m)$ from the origin $0$ to $\partial_{\hbox{hor}} (m)$.  
from $\partial_{\hbox{hor}}(m)$ to $\Lattice_0$ whose edges are all outside 
of $B^*(m)\cup \partial_e(m)$.

\begin{lemma}
Assume $p<p_c(d)$ and $\theta^I(p,\sigma)>0$. 
Then for sufficiently large $m$, 
\begin{equation}
  \label{eq.999a}
   P_{p,\sigma}^I (Q_m(x))  \;>\; \sfrac{1}{2} \theta^I(p,\sigma)
	 \hspace{4mm}\hbox{for all $x\in \Lattice_0 \cap B^*(m)$.}
\end{equation}
Also,
\begin{equation}
  \label{eq.999b}
     \lim_{m\rightarrow\infty}P^I_{p,\sigma}(F_m) \;=\; 1.
\end{equation}
\label{lemma999}     %%ZXZ[lemma999]
\end{lemma}

\begin{proof}
Let $Z_m= \{v\leftrightarrow \Lattice_0 \hbox{ for some $v\in\partial_{\hbox{hor}}(m)$}\}$.
Using $\LV \partial_{\hbox{hor}} (m)\RV = o(m^d)$ and $\log m = o(h(m))$, we observe that 
as $m\rightarrow\infty$,
\begin{equation}
      \label{eq.gapbound}
  %\fl   \hspace{8mm}
     P_{p,\sigma}^I(Z_m)
       \; \leq  \;    \sum_{v\in \partial_{\hbox{hor}}(m)}
        P_{p,\sigma}^I(v \leftrightarrow \Lattice_0)   
      \; \leq \; \LV \partial_{\hbox{hor}} (m)\RV \,e^{-\delta_p h(m)}   
      \; = \; o(1)     \,.
 %   \hspace{9mm}\hbox{as $m\rightarrow\infty$}.        
\end{equation}
Equation \Ref{eq.999b} follows since $F_m^c\subseteq Z_m$.   For every 
$x\in \Lattice_0\cap B^*(m)$ we have 
$\LC C(x) \cap \partial_{\hbox{hor}}(m)\neq\emptyset \RC \;\subseteq\;Z_m$, so equation (\ref{eq.gapbound}) implies that 
\begin{equation}
\label{eq.limmaxp}
\lim_{m\to\infty} \L  \max_{x \in \Lattice_0 \cap B^*(m)} 
P_{p,\sigma}^I \L C(x) \cap  \partial_{\hbox{hor}}(m) \;\not=\;\emptyset\R \R \;=\;0 . 
\end{equation}
Next, for every $x\in \Lattice_0\cap B^*(m)$ we clearly have
\[  
    \LC |C(x)| \;=\; \infty \RC \cap
    \LC C(x) \cap \partial_{\hbox{hor}}(m)=\emptyset \RC \;\subseteq\; Q_m(x) . 
\]
Hence by equation \Ref{eq.limmaxp}, for sufficiently large $m$ we have
\[ %\fl  \hspace{5mm}
   P_{p,\sigma}^I( Q_m (x))
\; >\; \sfrac{1}{2} P_{p,\sigma}^I \L |C(x)| {=} \infty \R
\;=\; \sfrac{1}{2} \theta^I(p,\sigma)   
  \hspace{5mm}\hbox{for all $x\in \Lattice_0\cap B^*(m)$}, 
\]
which proves Equation \Ref{eq.999a}.
\end{proof}

%A corollary of lemma \ref{lemma999} is the following.
%\begin{corollary}
%\[ \fl
%\lim_{m\to\infty} \L  \max_{x \in \Lattice_0 \cap B^*(m)} 
%P_{p,\sigma}^I \L C(x) \cap  \partial_{\hbox{hor}}(m) \not=\emptyset\R \R =0 . \]
%\label{cor999}     %%ZXZ[cor999]
%\end{corollary}

Let $L(m)$ be the set of vertices in $B^*(m)$ where
$Q_m(x)$ occurs:
\begin{equation}
L(m) \;=\; \{ x\in \Lattice_0 \cap B^*(m) \Vert \, \hbox{$Q_m(x)$ occurs} \}.
\label{eqn5959}     %%ZXZ[eqn5959]
\end{equation}
Then the next lemma shows that 
$P_{p,\sigma}^I\L  \LV L(m) \RV \geq \sfrac{1}{2}m^s \theta^I(p,\sigma)\R $
is not small in the supercritical regime.

\begin{lemma}
Assume $p<p_c(d)$ and $\theta^I(p,\sigma)>0$.  Then
for large $m$ we have
$P_{p,\sigma}^I\L  \LV L(m) \RV 
 {\geq} \sfrac{1}{2}m^s \theta^I(p,\sigma) 
\R \;\geq\; \sfrac{1}{4} \theta^I(p,\sigma)$.
\label{lemma1010}     %%ZXZ[lemma1010]
\end{lemma}

\begin{proof}
Notice that
$|L(m)| \, \leq \,| \Lattice_0 \cap B^*(m) | \,= \,(2m+1)^s$.  Moreover,
\[ E_{p,\sigma}^I |L(m)| \; =\; \sum_{x\in \Lattice_0 \cap B^*(m)}
 P_{p,\sigma}^I (Q_m(x)) \;\geq\; \sfrac{1}{2} \theta^I(p,\sigma)\,
       (2m+1)^s \]
for large $m$, by lemma \ref{lemma999}. Hence,
\begin{eqnarray*}
%\fl   \hspace{5mm}
\sfrac{1}{2} \theta^I(p,\sigma)\, (2m+1)^s 
&  \;\leq\; & E_{p,\sigma}^I |L(m)|  \\
& \;\leq\; & \sfrac{1}{2} m^s \theta^I(p,\sigma)P_{p,\sigma}^I
  \L \LV L(m) \RV {<} \sfrac{1}{2} m^s \theta^I(p,\sigma) \R \\
& &\qquad\qquad\quad    \,+\, (2m+1)^s P_{p,\sigma}^I
  \L \LV L(m) \RV {\geq}  \sfrac{1}{2} m^s \theta^I(p,\sigma) \R \\
& \;\leq\;& \sfrac{1}{2} m^s \theta^I(p,\sigma) \,+\, (2m+1)^s
  P_{p,\sigma}^I\L  
\LV L(m) \RV {\geq}  \sfrac{1}{2} m^s \theta^I(p,\sigma) \R .
\end{eqnarray*}
Solving for $P_{p,\sigma}^I\L  
\LV L(m) \RV {\geq} \sfrac{1}{2} m^s \theta^I(p,\sigma) \R $ then gives
\[ %\fl  \hspace{5mm}
P_{p,\sigma}^I\L  
\LV L(m) \RV {\geq}  \sfrac{1}{2} m^s \theta^I(p,\sigma) \R 
\;\geq\; \sfrac{1}{2} \theta^I(p,\sigma) \L 1 - \L \sfrac{m}{2m+1}\R^s \R
\;\geq\; \sfrac{1}{4} \theta^I(p,\sigma) . \]
This completes the proof of the lemma.
\end{proof}

Let $A_m$ be the event that all edges in $\partial_{\hbox{vert}} (m)$ 
are open, and all edges of 
$\partial_{e} (m)$  are closed.    Then, for some $\beta_2>0$,
\begin{equation}
P_{p,\sigma}^I(A_m) \;=\; e^{-O\L |\partial_{\hbox{vert}} (m)|\R} 
\;\geq \; e^{-\beta_2 m^{s-1} (h(m))^{d-s}} .
\label{eqn6060}     %%ZXZ[eqn6060]
\end{equation}
%Let $F_m$ be the event that there is no open path in
%$B^*(m)\setminus \partial_{\hbox{vert}} (m)$ from the origin $0$ to 
%$\partial_{\hbox{hor}} (m)$.  Observe in particular by corollary
%\ref{cor999} that for every $\eps>0$ there is an $M$ such that for all $m\geq M$,
%\begin{equation}
%\fl
%P_{p,\sigma}^I \L C(0) \cap  \partial_{\hbox{hor}}(m) \not=\emptyset\R \leq
% \max_{x \in \Lattice_0 \cap B^*(m)} 
%P_{p,\sigma}^I \L C(x) \cap  \partial_{\hbox{hor}}(m) \not=\emptyset\R< \eps .
%\end{equation}
%Then, for $m$ large enough,
%\begin{equation}
%P_{p,\sigma}^I \L F_m \R > \sfrac{1}{2} .
%\label{eqn6262}     %%ZXZ[eqn6262]
%\end{equation}
Let $D_m$ be the event that the number of vertices in
$L(m)$ is at least $\sfrac{1}{2}m^s \theta^I(p,\sigma)$, i.e.
\begin{equation}
D_m \;=\; \LC  \LV L(m) \RV \;\geq\;  \sfrac{1}{2} m^s \theta^I(p,\sigma)  \RC 
\end{equation}
(recall that $L(m)$ is the set of vertices in $\Lattice_0\cap B^*(m)$
that are connected by open paths in $B^*(m)$ to
 $\partial_{\hbox{vert}} (m)$).
Observe the following:
\begin{itemize}
\item[(\textit{i})]
\[ %\fl 
A_m \cap Q_m(0) \cap D_m 
\,\subseteq \, \LC   |C(0)| \geq \sfrac{1}{2} m^s \theta^I(p,\sigma)\RC .\]
In other words, if all the edges in $\partial_{\hbox{vert}} (m)$ are open
and all edges of  $\partial_{e} (m)$
are closed, and if $Q_m(0)$ and $D_m$ occur, then
the cluster at the origin has size at least 
$ \sfrac{1}{2} m^s \theta^I(p,\sigma)$.
\item[(\textit{ii})]
\[ %\fl 
A_m \cap F_m \,\subseteq \, \{C(0)\cap \Lattice_0 \subseteq B^*(m)\} \,\subseteq \,
 \{ |C(0)| < \infty \}\cup Y_m, \hbox{ where }\]
$Y_m=\{|C(0)|=\infty\hbox{ and }C(0)\cap \Lattice_0 \subseteq B^*(m)\}$.
\item[(\textit{iii})]
The events $A_m$, $F_m$, and $D_m\cap Q_m(0)$ are independent.
%$A_m$ is independent of $D_m$, $Q_m(0)$ and $F_m$.
\item[(\textit{iv})] By the FKG Inequality, and Lemmas \ref{lemma999} and \ref{lemma1010}, for
large $m$ we have
\[ %\fl \hspace{8mm}
P_{p,\sigma}^I(Q_m(0)\cap D_m)
\;\geq \;P_{p,\sigma}^I(Q_m(0)) \, P_{p,\sigma}^I(D_m)
\;\geq \;\sfrac{1}{8} \L \theta^I(p,\sigma) \R^2 . \]
%This, in particular, follows from lemmas \ref{lemma999} and \ref{lemma1010}
%as well as the bound in equation \Ref{eqn6262}.
\item[(\textit{v})]
$P^I_{p,\sigma}(Y_m)=0$, where $Y_m$ was defined in (\textit{ii}).  To see this, 
consider the new percolation measure $P^*$ in which each bond of $\Edges_0\cap B^*(m+1)$ 
is open with probability $\sigma$, and every other bond is open with probability $p$.
Then $P^I_{p,\sigma}(Y_m)=P^*(Y_m)$.  Moreover, since $p{<}p_c(d)$, we have
$P^H_p(Y_m)=0$, and hence $P^*(Y_m)$ is also 0 because $P^H_p$ and $P^*$ differ
on only finitely many edges.
\end{itemize}
Thus we conclude
\begin{eqnarray*}
P_{p,\sigma}^I&\L \infty >  | C(0) | \geq \sfrac{1}{2} m^s \theta^I(p,\sigma) \R  &  \\
&\;\geq\; P_{p,\sigma}^I\L A_m \cap Q_m(0) \cap D_m \cap F_m\cap Y_m^c\R  \q &
   \hbox{(by (\textit{i}) and  (\textit{ii}))}\\
   &\;=\; P_{p,\sigma}^I\L A_m \cap Q_m(0) \cap D_m \cap F_m\R  \q &
   \hbox{(by (\textit{v}))}\\
&\;=\; P_{p,\sigma}^I(A_m) \, P_{p,\sigma}^I ( D_m \cap Q_m(0) ) \,
    P_{p,\sigma}^I(F_m)\q  &
\hbox{(by (\textit{iii}))} .
% \\
% &\;\geq  \; P_{p,\sigma}^I\L A_m \R \, \left(\sfrac{1}{8}  \theta^I(p,\sigma)^2-P(F_m^c)\right)
% \q\hbox{(by (\textit{iv}))} .
\end{eqnarray*}
The proof of theorem \ref{thm999} is completed by comparing this last lower bound 
% on $P_{p,\sigma}^I\L \infty > | C(0) | \geq  \sfrac{1}{2} m^s \theta^I(p,\sigma) \R$
with equations \Ref{eqn6060} and \Ref{eq.999b}, as well as (\textit{iv}).

\subsection{Long-range connectivity above $p_c(d)$}

Recall that the two-point connectivity function is 
$\tau^I_{p,\sigma}(x,y) \,=\, P^I_{p,\sigma}(x\leftrightarrow y)$.
The next result shows that this function is bounded away from 0 for any given $(p,\sigma)$
in $\C{R}_H$.

\begin{proposition}
  \label{prop-connect}
Fix $p>p_c(d)$ and $\sigma\in[0,1]$.  Then $\inf\{\tau^I_{p,\sigma}(x,y) :  x,y\in \Lattice\} \,>\,0$.
\end{proposition}

\begin{proof}
Let $\vec{0}$ be the origin in $\Lattice$.  Since $P^I_{p,\sigma}(x\leftrightarrow y) \,\geq\,
P^I_{p,\sigma}(x\leftrightarrow \vec{0})\,P^I_{p,\sigma}(x\leftrightarrow \vec{0})$ by the 
FKG inequality, it suffices to prove that 
$\inf\{P^I_{p,\sigma}(v\leftrightarrow\vec{0}) :  v\in \Lattice\} \,>\,0$.

For $\vec{v}\in\IntN^d$, denote the $d$-th coordinate of
$\vec{v}$ by $v_d$.  Define the half-lattice
\begin{equation}
\Lattice_+ (1) = \{ \vec{v}\in \Lattice \vert v_d \geq 1 \}.
\end{equation}
Choose the origin in $\Lattice_+(1)$ at $\vec{1} = \vec{0}+\vec{e}_d$
and let $P_p^+$ be the (usual homogeneous) percolation measure in 
the half-lattice $\Lattice^+$ (see for example reference \cite{BGN91}).  
%Suppose that $p>p_c(d)$ and consider homogeneous percolation in $\Lattice_+(1)$.  
Since $p{>}p_c(d)$,  with probability 1 there is an infinite cluster
$C_+$ in $\Lattice^+$ \cite{GM90}, 
which is unique by the corollary to theorem 1.1 in \cite{BGN91}.
Let $C$ be the cluster containing $\vec{0}$ in $\Lattice$.
By noting that the edge $\vec{0}{\sim}\vec{1}$ is open with 
probability $p$, and using the FKG inequality and the fact that $P_p^+(\vec{v}\in C_+)$ is an 
increasing function of $v_1$, we see that for any 
$\vec{v} \in \Lattice_+(1)$ we have
\begin{eqnarray}
P_{p,\sigma}^I(\hbox{$\vec{v}\in C$}) 
&\geq &p\, P_p^+(\hbox{$\vec{v}\in C(\vec{1})$})\cr
&\geq &p\, P_p^+(\hbox{$\vec{v}\in C_+$ and $\vec{1}\in C_+$}) \cr
&\geq &p\, P_p^+(\hbox{$\vec{v} \in C_+$}) \, P_p^+(\hbox{$\vec{1}\in C_+$}) \cr
&\geq & p\, \L P_p^+(\vec{1}\in C_+) \R^2   \;>\;0.
\end{eqnarray}
This is uniform for all $\vec{v} \in \Lattice_+(1)$.  A similar bound
follows if $-\vec{v} \in \Lattice_+(1)$.  If $\vec{v} \in \Lattice_0$, then
since the edge $\vec{v}{\sim}(\vec{v}{+}\vec{e}_d)$ is open with
probability $p$, we see that
\[
P_{p,\sigma}^I (\hbox{$\vec{v} \in C$})
\; \geq \; p\, P^I_{p,\sigma}(\hbox{$\vec{v}+\vec{e}_d \in C$}) 
\; \geq \; p^2\, \L P_p^+( \vec{1}\in C_+) \R^2
\]
since $\vec{v}+\vec{e}_d \in \Lattice_+(1)$.  Therefore
$p^2\L P_p^+( \vec{1}\in C_+) \R^2$ is a positive lower bound for 
$P^I_{p,\sigma}(\vec{v}\leftrightarrow\vec{0})$ that is uniform in $\vec{v}\in \Lattice$. 
\end{proof}

\section{Collapsing animals, and the function $\zeta^I(p,\sigma)$}
\label{section6}     %%ZXZ[section6]

\subsection{Lattice animals, collapse and (homogeneous) percolation}
\label{sec6.1}

A \textit{lattice animal} is a connected and finite subgraph of
$\Lattice$.  All animals will be rooted at the origin, unless
otherwise indicated.

The \textit{size} of the animal is its number of vertices,
and the \textit{perimeter} of the animal is the collection of lattice 
edges which are incident with the animal but are not in the animal.
The \textit{perimeter size} is the number of edges in the perimeter.

Let $a_n(t)$ denote the number of distinct animals 
containing the origin, having $n$ edges, and having 
perimeter size $t$.  For example, in $\IntN^d$, 
$a_0(2d)=1$, $a_1(2d+2)=2d$, and so on. 

As before, denote the the cluster at the origin by $C$,
and let $|C|$ denote the number of vertices in $C$ and $\|C\|$
be the number of edges in $C$.  It is known that the limits
\begin{equation}
%\fl   \hspace{4mm}
\zeta^H(p) = - \lim_{n\to\infty} \sfrac{1}{n} \log
P_p^H(|C|{=}n)\q\hbox{and}\;
\psi^H(p)= - \lim_{n\to\infty} \sfrac{1}{n} \log
P_p^H(\|C\|{=}n)
\end{equation}
exist \cite{G99}.  Moreover, since 
$\sfrac{1}{d} \|C\| \leq |C|\leq \|C\|+1$ for all clusters $C$, it follows
that $\zeta^H(p) = 0$ if and only if $\psi^H(p) = 0$.

The weight of the open cluster $C$ at the origin in homogeneous
percolation is $p^{\|C\|}q^t$ (where $q=1-p$).  The probability that $C$ 
has $n$ edges is 
\begin{equation}
P_p^H(\|C\|{=}n) \;=\; \sum_{t\geq 0} a_n(t)\,p^nq^t .
\label{eqn7P}    %%ZXZ[eqn7P]
\end{equation}
This shows that
\begin{equation}
\psi^H(p) \;=\; - \log p - \lim_{n\to\infty} \sfrac{1}{n} \log  \sum_{t\geq 0} a_n(t)\,q^t .
\label{eqn59}    %%ZXZ[eqn59]
\end{equation}

A \textit{contact} of an animal is a lattice edge that is not in the animal but whose
endpoints are both in the animal.
Contacts are part of the perimeter of a cluster --- they are closed
edges with both endpoints in the open cluster. 

An edge is in a \textit{cycle} in the open cluster at the origin if the cluster
stays connected when the state of the edge is changed to \textit{closed}.  In 
the context of the lattice animal, an edge is in a cycle if deleting it does not
disconnect the animal.  The \textit{cyclomatic index} $c$ of a lattice animal
is the maximum number of edges which can be deleted without disconnecting the animal.

A model of lattice animals in the \textit{cycle-contact ensemble} is constructed
by counting lattice animals with respect to \textit{cyclomatic index} 
and \textit{contacts} \cite{MSW88}.  Hence, let $a_n(c,k)$ be the number of animals 
containing the origin
with $n$ edges, cyclomatic index $c$, and $k$ contacts.
The partition function of the model is 
\begin{equation}
Z_n^A(x,y) \;=\; \sum_{\sstack{c\geq 0}{k\geq 0}}
a_n(c,k)\, x^c y^k .
\label{eqn9z}    %%ZXZ[eqn9z]
\end{equation}
The parameters $x$
and $y$ are the \textit{cycle} and \textit{contact activities} (or generating
variables) in the model. The free energy of this 
model is known to exist \cite{FGSW99}, and is defined by
\begin{equation}
\C{F}^A (x,y) \;=\; \lim_{n\to\infty} \sfrac{1}{n} \log Z_n^A(x,y) .
\label{eqn10z}    %%ZXZ[eqn10z]
\end{equation}

For animals in $\IntN^d$, we have $2d\,v = 2n + t + k$ (where $v$ is the number
of vertices), while from Euler's 
relation we get $c=n-v+1$.  Eliminating $v$ from these two relations 
implies that the cyclomatic index and the number of contacts are related 
to the perimeter by
\begin{equation}
t \;=\; 2d + 2(d-1)n - k - 2d\, c
\end{equation}

% The cyclomatic index and number of contacts in an animal is related to 
% its perimeter by
% \begin{equation}
% t \;=\; 2d + 2(d-1)n - k - 2d\, c
% \end{equation}
% since if an animal has $v$ vertices and $n$ edges, then
% $2d\,v = 2n + t + k$ and by Euler's relation, $c=n-v+1$, from which one
% may eliminate $v$.  

Hence, write equation \Ref{eqn7P} as
\begin{equation}
P_p^H(\|C\|{=}n) \;=\; p^n \sum_{\sstack{c\geq 0}{k\geq 0}}
a_n(c,k)\, q^{2d + 2(d-1)n - k - 2d\, c}
\end{equation}
Comparing the above expression to equation \Ref{eqn9z} shows that
\begin{equation}
P_p^H(\|C\|{=}n) \;=\; q^{2d}\L pq^{2(d-1)}\R^n Z_n^A(q^{-2d},q^{-1}) .
\end{equation}
Taking logarithms of both sides, dividing by $n$ and letting $n\to\infty$ gives
\begin{equation}
\psi^H(p) \;=\; -2(d-1)\log q - \log p - \C{F}^A(q^{-2d},q^{-1}) . 
\end{equation}
Since $\psi^H(p) = 0$ if $p>p_c(d)$ and $\psi^H(p) > 0$ for
$p < p_c(d)$ (see section \ref{sec-homog}), this proves that 
\begin{equation}
\C{F}^A(q^{-2d},q^{-1})  \cases{
\;< \; -2(d-1)\log q - \log p, & if $p< p_c(d)$; \cr
\;=\; -2(d-1)\log q - \log p, & if $p>p_c(d)$. }
\end{equation}
In particular, $\C{F}(x,y)$ is non-analytic at $p=p_c(d)$
where $x = (1-p)^{-2d}$ and $y=(1-p)^{-1}$, in
which case the animals are weighted as critical percolation clusters
and the model undergoes a \textit{collapse phase transition} which may 
be interpreted as a model for gelation of a random medium.  
In this phase both $x$ and $y$ are large,
and the animals are rich in both cycles and contacts, resulting in compact
clusters.

\subsection{Proof of theorem \ref{theorem5-12}}
\label{section6-2}

Our strategy is to bound $P_{p,\sigma}^I(|C| {=} n)$
from below by $P_{p}^H(|C|{=}n{-}1)$, and then to use the
lower bound from homogeneous percolation (see equation
\Ref{eqn444z}).

Let $\Lattice_+$ be the positive half-lattice, consisting of
vertices $\{z\in \IntN^d:z_d\geq 1\}$ and all induced edges.
Let $\C{A}_+$ be the set of animals $D$ that are contained in 
$\Lattice_+$ and rooted at the vertex $e_d=(0,\ldots,0,1)$.   
Then each $D\in \C{A}_+$ is the translation of exactly $\|D\|$ 
animals which are rooted at the origin in $\Lattice$, and conversely 
every animal containing the origin is the translation of
at least one animal in $\C{A}_+$.  

For $n,m,t\geq 0$, let $a[n,m,t]$ (respectively, $a_+[n,m,t]$) be the number of
animals rooted at the origin (respectively, the number of animals 
in $\C{A}_+$ rooted at $e_d$) which have $n$ vertices, 
$m$ edges, and perimeter size $t$.  Then the 
preceding paragraph shows that $a_+[n,m,t]\,\geq \sfrac{1}{n} \,a[n,m,t]$.

For each $D\in \C{A}_+$, let $\hat{D}$ be the animal in $\Lattice_+$
obtained by 
% translating $D$ into $\Lattice_+$ and rooting it at $e_d$ and then 
adding the edge $0{\sim}e_d$ to $D$.   
If $D$ has $m$ edges, and perimeter size $t$,  
then $\hat{D}$ has $m+1$ edges (all in $\Edges\setminus\Edges_0$) and 
perimeter size $t+2(d-1)$ with exactly $2s$ perimeter edges in $\Edges_0$.

Thus we have (with $q=1-p$) 
\begin{eqnarray*}
  P^I_{p,\sigma}(|C|=n)  & \geq & \sum_{D\in \C{A}_+: |D|=n-1} P^I_{p,\sigma}(C=\hat{D}) \\
    & = & \sum_{m,t\geq 0} a_+[n-1,m,t]\, p^{m+1}q^{t+2(d-s-1)}\sigma^{2s}    \\
    & \geq & \frac{1}{n-1}\,p\,q^{2(d-s-1)}\sigma^{2s} \sum_{m,t\geq 0} a[n-1,m,t]\,p^m q^t   \\
    & = & \frac{1}{n-1}\,p\,q^{2(d-s-1)}\sigma^{2s} P_p^H(|C|=n-1) \,.
\end{eqnarray*}
Theorem \ref{theorem5-12} now follows from equation \Ref{eqn444z}.

\subsection{Lattice animals, adsorption and inhomogeneous percolation}

In this section our aim is to make a link between inhomogeneous
percolation and a model of lattice animals, similar in nature to the 
association made in section \ref{sec6.1} for homogeneous percolation.  

Our goal is to prove existence of the limits
\begin{equation}
%\fl
\hspace{5mm}
\zeta^I(p,\sigma) \;=\; - \lim_{n\to\infty} \sfrac{1}{n}
P^I_{p,\sigma}(|C| {=} n) \;\hbox{and}\;
\psi^I(p,\sigma) \;=\; - \lim_{n\to\infty} \sfrac{1}{n}
P^I_{p,\sigma}(\|C\| {=} n) 
\label{eqn67}     %%ZXZ[eqn67]
\end{equation}
and to relate these to singular points in the free energies
of lattice animals.

We first show existence of the limits in equation
\Ref{eqn67}. 

Let $\C{A}$ be the set of lattice animals in $\Lattice$ containing the origin.
Let $a_{n,m} (t,r)$ be the number of animals in $\C{A}$
having $n$ edges, of which $m$ are in
$\Edges_0$, and whose perimeter consists of $t$ edges in $\Edges\setminus\Edges_0$
and $r$ edges in $\Edges_0$.  

Define the partition function of these animals by
\begin{equation}
Z_n^I (x,y,z) \;=\; 
\sum_{\sstack{t\geq 0}{r\geq 0}} \sum_{m\geq 0}
a_{n,m} (t,r) \, x^t y^r z^m .
\label{eqn33z}   %%ZXZ[eqn33z]
\end{equation}
Then the probability that the cluster at the origin has
$n$ edges is given by
\begin{equation}
P^I_{p,\sigma} (\| C \| {=} n)
\;=\; p^n\,Z_n^I(q,\tau,\sigma/p) 
\label{eqn72z}   %%ZXZ[eqn72z]
\end{equation}
where $q=1-p$ and $\tau=1-\sigma$.
This shows that if the limiting free energy $\C{F}(x,y,z) 
\;=\; \lim_{n\to\infty} \sfrac{1}{n} \log Z_n^I(x,y,z)$
exists, then the limit $\psi^I(p,\sigma)$ in equation \Ref{eqn67} 
also exists.  Existence of $\zeta^I(p,\sigma)$ is done using 
a similar approach, but counting animals in a different
ensemble (number of vertices).

The basic construction for showing the existence of
$\C{F}(x,y,z)$ is illustrated in figure \ref{figure7}.
Consider two animals $\omegas1$ and $\omegas2$,  each intersecting
$\Lattice_0$ at vertices we call \textit{visits}.
An \textit{edge-visit} in these animals is an edge of the animal which is
also in $\Edges_0$.  
Observe that translations parallel to $\Lattice_0$ preserve visits and edge-visits.

%%%%%%%%%%%%%%%%%%%%%%
\begin{figure}[t!]
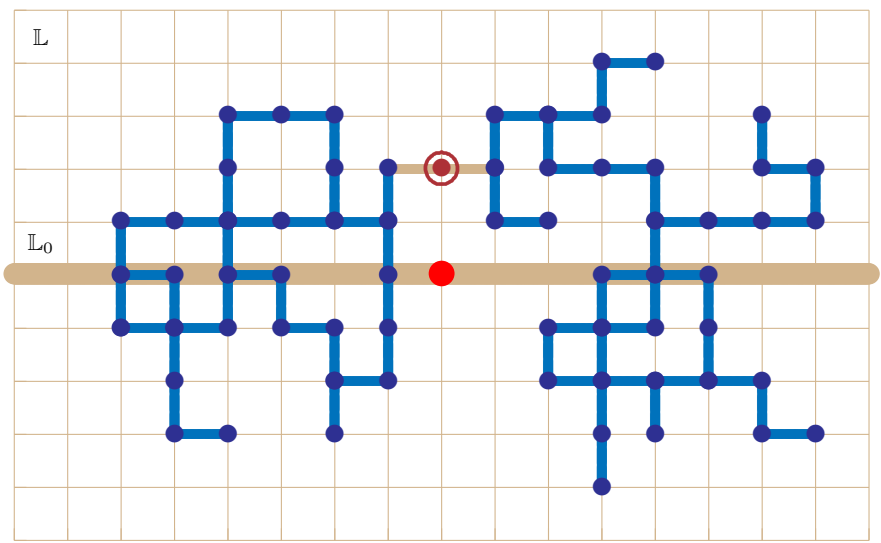

\centering\hfill
\input figure7.tex
\caption{Concatenation of two clusters in the inhomogeneous lattice.
Two animals are placed in a standard placing with a minimal distance of
two lattice steps separating them.  The animals are joined into a single
animal by adding two edges and a single new vertex (marked above)
along the lexicographic least path separating the two animals.}
\label{figure7}   %%ZXZ[figure7]
\end{figure}
%%%%%%%%%%%%%%%%%%%%%%

The goal is to concatenate $\omegas1$ and $\omegas2$ into one
animal from which the original pair of animals can be uniquely recovered.
  
A \textit{placing} $(\hat{\omega}_1,\hat{\omega}_2)$ of two animals $\omegas1$
and $\omegas2$ is a pair of translations (parallel to $\Lattice_0$)  $\hat{\omega}_1$ of
$\omega_1$ and $\hat{\omega}_2$ of $\omegas2$ such that the minimum 
distance between $\hat{\omega}_1$ and $\hat{\omega}_2$ is at least $2$ steps.   

There are infinitely many placings  $(\hat{\omega}_1,\hat{\omega}_2)$, but 
%up to translations parallel to $\Lattice_0$, 
there are only finitely many non-equivalent placings 
with a minimum distance of two
(where two placings are equivalent if they only differ by
an overall translation parallel to $\Lattice_0$).

Consider a placing $(\hat{\omega}_1,\hat{\omega}_2)$ with the following
properties:  (1) each visit in $\hat{\omega}_1$ is lexicographically less
than each visit in $\hat{\omega}_2$; (2) the shortest path in $\Lattice$
from a vertex in $\hat{\omega}_1$ to a vertex in $\hat{\omega}_2$ has length
two.  These two properties define a nonempty finite collection of
placings (up to equivalence), one of which is lexicographically least.  
This is the \textit{standard placing}.
%
%\textcolor{blue}{[[One has to select a standard representative
%from each set of placings to do the concatenation -- I think
%ordering them lexicographically may be one way.  This gives
%a unique outcome for each pair that is concatenated.  Is it
%necessary to describe the method of ordering vertices
%lexicographically by coordinates, and then how this 
%induces an ordering of animals?]]}

Observe that the total perimeter of the animals in a standard
placing is the sum of the perimeters of the two animals.

In each standard placing there is at least one path of length
two joining the two animals.  In the set of such paths, there
is a path $P$ which is lexicographically least.  The animals $\hat{\omega}_1$
and $\hat{\omega}_2$ may be concatenated by joining them into
a single animal by adding two edges along $P$.  This increases
the number of edges by $2$ and decreases the total perimeter 
of the animals by $2$.  Observe that the center vertex of $P$
is a cut-vertex in the concatenated animal.

Consider the possible arrangements of the two added edges
along $P$: (a) The two added edges are disjoint with $\Lattice_0$.
(b) One edge in $P$ is in $\Lattice_0$, and 
(c) both edges are in $\Lattice_0$.

Next, account for the change in the perimeter of the animals upon
concatenation.  Suppose that $\omegas1$ is an animal 
with $\n1$ edges and with $m-\m1$ edges in $\Lattice_0$, 
and with perimeter size $t+r-(t_{1}+r_{1})$, including $r-r_{1}$ 
perimeter edges in $\Lattice_0$.  

Similarly, suppose that $\omegas2$ is an animal with $\n2$ edges 
and with $\m1$ edges in $\Lattice_0$, and with perimeter
size $t_{1}+r_{1}$, including $r_{1}$ perimeter edges in $\Lattice_0$.

Putting these animals in a standard placing and concatenating them
gives an animal $\omega$ with $\n1+\n2+2$ edges in total, 
and there are either $m$ edge-visits (case (a)), or $(m+1)$ 
edge-visits (case (b)), or $(m+2)$ edge-visits (case (c)).
These different outcomes are due to the fact that new edges
may be created in $\Lattice_0$ when the concatenation introduces
two new edges.

It is necessary that $\omegas1$ and $\omegas2$ can be recovered
from the concatenated animal. Since the concatenation is done by adding 
two edges incident with one another in a new cut-vertex, these edges 
can be located in $\omega$ by colouring the new vertex red.
This gives an animal with one red vertex of degree 2 (and the remaining
vertices are all black).  Note that the maximum number
of vertices in $\omega$ is $n_1+n_2+3$.

By deleting the two edges incident on the red vertex, it is possible
to recover the two translated animals $\hat{\omega}_1$ and $\hat{\omega}_2$ 
in their  standard placing.  Observe
that there are at most $2d{-}2$ new perimeter edges associated with
the red vertex, and that at most $2s$ of these may be in the
defect lattice $\Lattice_0$.

We now account for the changes in the number of perimeter
edges.  The concatenation deletes two perimeter edges, but
the new red vertex creates new perimeter edges.  Thus, 
$\omega$ has perimeter between $t+r-2$ and $t+r-2+2d-2$
of which between $r$ and $r+2s$ are in $\Lattice_0$.

The roots of the animals $\omegas{i}$ are discarded when they
are put in standard placing, and so the number of choices for
each $\hat{\omega}_{i}$ is at least $a_{n_1,m-m_1}(t-t_1,r-r_1) / (\n1+1)$
for $\hat{\omega}_1$ and $a_{n_2,m_1}(t_1,r_1)/(\n2+1)$ for $\hat{\omega}_2$.

The concatenated animal $\omega$ is similarly unrooted,
and there are at most $(n_1+n_2+3)$ positions for the red vertex.   
Accounting for the different possible numbers of 
edge-visits and perimeter sizes then shows that
\begin{eqnarray}
&%\fl
 \sum_{m_1,t_1,r_1} 
 \L  \frac{a_{n_1,m-m_1}(t-t_1,r-r_1)}{\n1+1} \R 
  \L \frac{a_{n_2,m_1} (t_1,r_1)}{\n2+1} \R& \nonumber \\
&%\fl 
\q \;\leq\; (\n1+\n2+3) \sum_{i=0}^{2d-2}\sum_{j=0}^{2s}
\LH a_{\n1+\n2+2,m}(t-2+i,r+j) \right. & \nonumber \\
& \hspace{-2cm} \left.
+ a_{\n1+\n2+2,m+1}(t-1+i,r-1+j)
 + a_{\n1+\n2+2,m+2}(t+i,r-2-j) \RH , &
\label{eqn84}   %%ZXZ[eqn84]
\end{eqnarray} 
where the summation over $i$ and $j$ accounts for new perimeter
edges incident on the red vertex.

Define $\phi(x,y) = \sum_{i=0}^{2d-2} \sum_{j=0}^{2s}
x^{-i}y^{-j}$.  Multiply equation \Ref{eqn84} by 
$x^ty^rz^m$ and sum over $m$, $t$ and $r$.  This gives
\begin{eqnarray}
Z_{\n1}(x,y,z)\,Z_{\n2}(x,y,z) 
&\;\leq\;&
(\n1{+}\n2{+}1)^2(\n1{+}\n2{+}3) \phi(x,y) \nonumber \\
& &\times \LH  x^2 {+} z^{-1}xy {+} z^{-2}y^2\RH\, Z_{\n1+\n2+2}(x,y,z) .
\end{eqnarray} 
Define $\lambda(x,y,z) = \phi(x,y) \L x^2 + z^{-1}xy\ + z^{-2}y^2\R$.
Then the above simplifies to 
\begin{eqnarray}
Z_{\n1}(x,y,z)\,Z_{\n2}(x,y,z) &\;\leq\; &
(\n1{+}\n2{+}1)^2(\n1{+}\n2{+}3) \nonumber \\
& & \q \times \lambda (x,y,z) \, Z_{\n1+\n2+2}(x,y,z) .
\end{eqnarray}
This shows that the function $ Z_{n-2} (x,y,z)/\lambda(x,y,z)$
satisfies a generalised supermultiplicative inequality on $\NatN$, and by
references \cite{H62,HF57} one obtains the following theorem.

\begin{theorem} For  $x,y,z \in(0,\infty)$ the limit
\[
\C{F}^I(x,y,z)
\;=\; \lim_{n\to\infty} \sfrac{1}{n} \log Z_n^I (x,y,z)
\]
exists. Moreover, $\C{F}^I (x,y,z)$ is log-convex in each of its arguments. 
\qed
\label{theorem12A}   %%ZXZ[theorem12A]
\end{theorem}

Log-convexity follows because $Z_n^I(x,y,z)$ is a 
polynomial in $\{x,y,z\}$ with positive coefficients.

Comparison to equation \Ref{eqn72z} gives the following relationship between
$\C{F}^I(x,y,z)$ and $\zeta^I(p,\sigma)$:
\begin{equation}
\C{F}^I (q,\tau,\sigma/p) \;=\; - \log p - \psi^I(p,\sigma)
\label{eqn84A}   %%ZXZ[eqn84A]
\end{equation}
which is valid for $p,\sigma\in(0,1)$ and proves existence of the
limit definition of $\psi^I(p,\sigma)$ in equation \Ref{eqn67}.

Existence of $\zeta^I(p,\sigma)$ can be similarly shown, as follows. 

Let $A_{v,n,m}(t,r)$ be the number of edge animals at the origin
as above, but with $v$ vertices, $n$ edges of which $m$ are in
$\Edges_0$, and with perimeter having $t$ edges in $\Edges\setminus \Edges_0$ and
$r$ edges in $\Edges_0$.  Define the partition function
\begin{equation}
Y_v(a,x,y,z) \;=\; 
\sum_{\sstack{t\geq 0}{r\geq 0}} \sum_{\sstack{n\geq 0}{m\geq 0}}
A_{v,n,m} (t,r) \,a^n x^t y^r z^m .
\label{eqn33z}   %%ZXZ[eqn33z]
\end{equation}
Then the probability that the animal at the origin has size $v$
is given by
\begin{equation}
P_{p,\sigma}^I(|C|=v) \;=\;
Y_v(p,q,\tau,\sigma/p) .
\label{eqn86A}   %%ZXZ[eqn86A]
\end{equation}

Repeating the construction in figure \ref{figure7} in this 
ensemble gives an outcome similar to the above, but now with
\begin{eqnarray}
& &
 \sum_{n_1,m_1,t_1,r_1} 
 \L  \frac{A_{v_1,n-n_1,m-m_1}(t-t_1,r-r_1)}{v_1} \R 
  \L \frac{A_{v_2,n_1,m_1} (t_1,r_1)}{v_2} \R \nonumber \\
& &
  \;\leq\; (v_1+v_2+1) \sum_{i=0}^{2d-2}\sum_{j=0}^{2s}
\LH A_{v_1+v_2+1,n+2,m}(t-2+i,r+j) \right.   \\
& & 
\left. + A_{v_1+v_2+1,n+2,m+1}(t{-}1{+}i,r{-}1{+}j) 
              + A_{v_1+v_2+1,n+2,m+2}(t{+}i,r{-}2{+}j) \RH . 
\nonumber
\end{eqnarray} 
Multiply this by $a^n x^t y^r z^m$ and summing the left hand
side over $\{n,t,r,m\}$ gives
\begin{eqnarray}
%\fl
Y_{v_1}(a,x,y,z)\,
Y_{v_2}(a,x,y,z)
& \leq &v_1 v_2 (v_1+v_2+1) \nonumber \\
& &\hspace{-3.5cm} 
\LH a^{-2} \phi(x,y) \L x^2 + xy z^{-1} + y^2 z^{-2} \R\RH\,
Y_{v_1+v_2+1} (a,x,y,z) .
\end{eqnarray}
Similarly to theorem \ref{theorem12A}, $Y_v(a,x,y,z)$
satisfies a generalised supermultiplicative inequality 
on $\NatN$, and by references \cite{H62,HF57} the
following theorem is a result.

\begin{theorem} For  $a,x,y,z \in(0,\infty)$ the limit
\[
\C{G}(a,x,y,z)
\;=\; \lim_{v\to\infty} \sfrac{1}{v} \log Y_v (a,x,y,z)
\]
exists. Moreover, $\C{G} (a,x,y,z)$ is log-convex 
in each of its arguments. 
\qed
\label{theorem12B}   %%ZXZ[theorem12B]
\end{theorem}

Notice by equation \Ref{eqn86A} that $\zeta^I(p,\sigma) = 
- \C{G}(p,q,\tau,\sigma/p)$ so that the limit in equation
\Ref{eqn67} exists.

We claim that $\zeta^I(p,\sigma) = 0$ in $\C{R}_H$. To see this,
suppose $\zeta^I(p,\sigma) > 0$ at some $(p,\sigma)$ in $\C{R}_H$.
%In the regime $\C{R}_H$ it follows from theorem \ref{theorem5-12}
%that $\lim_{n\to\infty} \sfrac{1}{n} \log
%P_{p,\sigma}^I(\infty{>}|C|{\geq}n) = 0$. ))
Then there exists an $\eps>0$
and an $N_\eps \in\NatN$ such that $P^I_{p,\sigma}(|C|=n) 
\leq e^{-\eps n}$ for all $n \geq N_\eps$.  This shows that
\[P^I_{p\sigma} (\infty{>}|C|{\geq}n) \leq 
\sum_{m\geq n} e^{-\eps m} = \frac{e^{-\eps n}}{1-e^{-\eps}} \]
for any $n\geq N_\eps$.  
%Take logarithms, divide by $n$ and let $n\to\infty$. 
%keeping $\eps$ fixed.  
%This shows that $\lim_{n\to\infty} \sfrac{1}{n} \log
%P_{p,\sigma}^I(\infty{>}|C|{\geq}n) \leq - \eps < 0$
%in $\C{R}_H$,
But this contradicts theorem \ref{theorem5-12}.  Therefore $\zeta^I(p,\sigma) = 0$.

A similar argument using theorem \ref{thm999} shows that $\zeta^I(p,\sigma) = 0$ in $\C{R}_L$.

Since $\zeta^I(p,\sigma) = 0$ in $\C{R}_H\cup\C{R}_L$, 
it follows that $\psi^I(p,\sigma) = 0$ in $\C{R}_H\cup\C{R}_L$.

On the other hand, 
by theorem \ref{theoremZ6}, $\zeta^I(p,\sigma) \geq
\sfrac{1}{2\,\chi^H(p)\chi^I(p,\sigma)}>0$, provided that $p<p_c(d)$ and
$\sigma\in(0,\sigma^*(p))$ (this is in regime $\C{R}_0$ in figure
\ref{Diagram}). This shows that $\psi^I(p,\sigma) > 0$
in $\C{R}_0$.

In terms of the free energy $\C{F}^I(x,y,z)$
in equation \Ref{eqn84A}, this implies that
\begin{equation}
%\fl
\C{F}^I(q,\tau,\sigma/p) \cases{
\;=\; - \log p, & \hbox{in $\C{R}_H$ (i.e.\ when $p>p_c(d)$)}; \cr
\;=\; - \log p, & \hbox{in $\C{R}_L$ (when $p<p_c(d)$ and $\sigma>\sigma^*(p)$)}; \cr
\;<\; - \log p, & 
\hbox{in $\C{R}_0$ (when $p<p_c(d)$ and $\sigma<\sigma^*(p)$}).}
\label{eqn72XX}   %%ZXZ[eqn72XX]
\end{equation}
Thus, $\C{F}^I(q,\tau,\sigma/p)$ is non-analytic along
the line segment $p=p_c(d)$ and $\sigma\in(0,\sigma^{**})$ (where
$\sigma^{**}$ is the limit of $\sigma^*(p)$ as $p$ approaches $p_c(d)$ from the left),
as well as along the surface critical curve $\sigma^*(p)$
for $0 \leq 0 < p_c(d)$.

\section{Numerical results}
\label{sectionNum}     %%ZXZ[section7]

We performed a numerical study of inhomogeous percolation 
using the Newman-Ziff algorithm \cite{NZ00} to sample clusters
in the model.  To describe the implementation, let $B(L)$ be the $d$-dimensional 
hypercube of side length $2L$ defined by
$B(L) = \LH -L,L\RH^d \cap \Lattice$. 
The boundary of $B(L)$ is  $\partial B(L)$, and it has
a vertical and a horizontal component, similar to equation \Ref{eqnA61}:
\begin{eqnarray}
\partial_{\hbox{vert}}(L) &=& \LC {v}\in B(L) \vert
\hbox{$|v_i| = L$ for at least one $i \leq s$} \RC \\
\partial_{\hbox{hor}}(L) &=& \LC {v}\in B(L) \vert
\hbox{$|v_i| = L$ for at least one $i > s$} \RC .
\end{eqnarray}
The vertical component $\partial_{\hbox{vert}} (L)$
is composed of $2s$ $(d{-}1)$-dimensional hypercubes defined
by
\[
%\fl
A_i = \LC {v} \in \partial_{\hbox{vert}}(L)\vert
\hbox{$v_i = L$} \RC \:\hbox{and}\;
A_{-i} = \LC {v} \in \partial_{\hbox{vert}}(L)\vert
\hbox{$v_i = -L$} \RC 
\]
for $i=1,2,\ldots,s$.

Consider a realisation of open edges in $\Lattice$ at
densities $(p,\sigma)$. This realisation gives sets
of open edges in $B(L)$:  Denote the set of open edges
in $B(L)\setminus \Lattice_0$ by $\C{P}$ and the
set of open edges in the defect plane $B(L)\cap \Lattice_0$
by $\C{S}$.  Let 
$\indic_{\C{P},\C{S},L} \L A_1 \leftrightsquigarrow A_{-1} \R$
be the indicator function that there is an open path inside $B(L)$ 
between two opposite vertical faces $A_1$ and $A_{-1}$ in 
$\partial_{\hbox{vert}}(L)$.  

The average of 
$\indic_{\C{P},\C{S},L} \L A_1 \leftrightsquigarrow A_{-1} \R$
for all realisations of $\C{P}$ at density $p$ and
all $\C{S}$ with $|\C{S}|=s$ is denoted by
$Q_p(s)$.  That is, $Q_p(s)$ is the probability that there is
an open path in $B(L)$ between $A_1$ and $A_{-1}$ when
bulk edges are open with probability $p$ given that there are
exactly $s$ surface edges open in $B(L) \cap \Lattice_0$. 
Then $0\leq s \leq S$ where $S$ is the total
number of edges in $B(L) \cap \Lattice_0$.

Following Newman and Ziff \cite{NZ00}, let us construct
\begin{equation}
Q_L (p,\sigma) = \sum_{s=0}^S \Bi{S}{s}
\sigma^s (1-\sigma)^{S-s}\, Q_p(s) .
\label{Q_L_p_sigma}
\end{equation}

Clearly, $Q_L(p,\sigma)$ decreases to zero with $L$ 
in $\C{R}_0$ (see figure \ref{Diagram}).  On the other
hand, it should approach a positive probability 
with increasing $\sigma$ for fixed $p\leq p_c(d)$ 
when $\sigma>\sigma^*(p)$.  
That is, if $L\to\infty$, then
$\lim_{L\to\infty} Q_L(p,\sigma) = 0$ in $\C{R}_0$, and 
$\liminf_{L\to\infty} Q_L(p,\sigma) > 0$ in the surface regime
$\C{R}_L$.  Hence, one may estimate the critical curve
$\sigma^*(p)$ by estimating $Q_L(p,\sigma)$ for
finite values of $L$ and various $(p, \sigma)$.

In figure \ref{figure01-32} numerical estimates of $Q_L(\sigma)
\equiv Q_L(\sigma,p)$ as a function of $\sigma$ for $p=0.1$
for the model with $(d,s)=(3,2)$ is presented (with
$L\in\C{L} = \{5,10,15,20,25,30,35\}$).  $Q_L(\sigma)$ is small
for $\sigma$ small, and increases with increasing $\sigma$.
For different values of $L$ all the curves pass almost through
the same point at a critical value of $\sigma$.

The normal scaling assumption for a function like $Q_L(\sigma)$ is
\begin{equation}
Q_L (\sigma) \simeq F_p(L^\phi(\sigma{-}\sigma^*(p)))
\end{equation}
where $\phi$ is a crossover exponent and $F$ is a scaling function.
In the case that $s=2$ the surface percolation at $\sigma=\sigma^*(p)$
should be in the same universality class as homogeneous percolation in
two dimensions.  If $\sigma = \sigma^*(p)$ then this shows
that $Q_L(\sigma^*(p)) \simeq F_p(0)$ so that the value of $Q_L(\sigma)$
is independent of $L$ at the critical point.  This indicates that the
point where all the curves intersect in figure \ref{figure01-32} is 
an estimate of the location of the critical point.

%%%%%%%%%%%%%%%%%%%%%%
\begin{figure}[t!]
\centering\hfill
\input figure01-32.tex
\caption{Plots of $Q_L(p,\sigma)$ for $p=0.1$
as a function of $\sigma$ for $d=3$ and $s=2$.  
Each curve is the average taken over $30000$ realisations 
of bulk percolation clusters with $p=0.1$ and then 
determining $Q_L(0.1,\sigma)$ as a function of $\sigma$.
The value of $L$ increases in the set $\{5,10,15,20,25,30,35\}$.
The curves intersect to high accuracy in a single point for larger
values of $L$, which gives an estimate of the location of
the critical point $\sigma^*(0.1)$.}
\label{figure01-32} 
\end{figure}
%%%%%%%%%%%%%%%%%%%%%%

%%%%%%%%%%%%%%%%%%%%%%%%
\begin{figure}[b!]
\centering\hfill
\includegraphics[width=100mm,height=60mm]{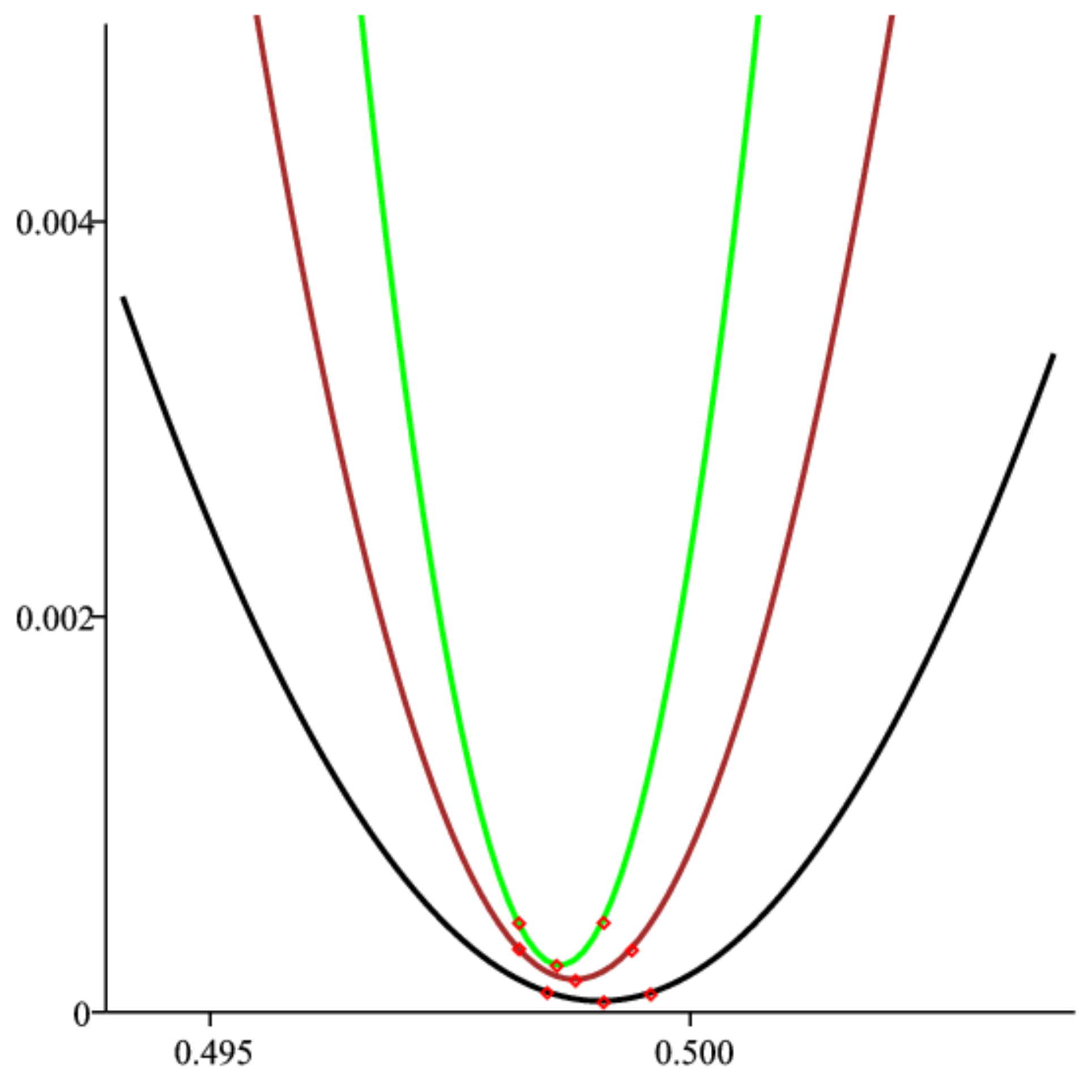}
\caption{A plot of $\C{E}^2(\sigma)$ against $\sigma$ for
$p=0.1$, $d=3$ and $s=2$.  The  top curve is for data
when $\C{L}=\{5,10,15,20,25,30\}$, the middle curve 
is obtained with $L=5$ dropped, and the bottom curve is
obtained when both $L=5$ and $L=10$ are dropped.
The minima in these curves are estimates of the location of the
narrow point in figure \ref{figure01-32}.  The width of 
$\C{E}^2(\sigma)$ at twice its minimum height, and also at
four times its minimum height, is indicated by the 
square symbols on each curve.}
\label{figureQLsigma01-32}
\end{figure}
%%%%%%%%%%%%%%%%%%%%%%%%
To find a numerical estimate of the crossing point, define the least
square width of the set of curves at surface density $\sigma$ by
\begin{equation}
\C{E}^2(\sigma) = \sum_{L\in\C{L}}\sum_{K\in\C{L}} 
\L Q_L(\sigma) - Q_K(\sigma) \R^2.
\end{equation}
$\C{E}^2(\sigma)$ is a measure of the square vertical width of the set of 
curves, and minimizing it gives an estimate of the location (the value of 
$\sigma$) of \textit{narrowest vertical waist} in the set of intersecting curves.
That is, this gives an estimate for $\sigma^*(p)$.  An error bar can 
be estimated by determining the values of $\sigma$ where
 $\C{E}^2(\sigma)$ is twice its minimum.  For example, the data in figure 
\ref{figure01-32} gives $\sigma^*(0.1) = 0.49859 \pm 0.00040$.   
A plot of $\C{E}^2(\sigma)$ against $\sigma$ for $p=0.1$ is 
given in figure \ref{figureQLsigma01-32}.

{\renewcommand{\baselinestretch}{1.5}
\begin{table}[t!]
\begin{center}
{
\begin{tabular}{|l|l||l|l|}
\hline\hline
\multicolumn{2}{|c||}{$d=3$, $s=2$} &
\multicolumn{2}{|c|}{$d=4$, $s=2$} \\
\hline
\q$p$ &\qq $\sigma^*(p)$ &\q $p$ &\qq $\sigma^*(p)$   \\
\hline
$0.000$ &$0.5$  &$0$ &$0.5$  \\
$0.050$ &$0.50043 \pm 0.00064$ &$0.050$&$0.5007 \pm 0.0013$   \\
$0.100$ &$0.4988  \pm 0.0011$   &$0.0750$&$0.4992 \pm 0.0010$  \\
$0.150$ &$0.4929  \pm 0.0018$   &$0.1100$&$0.4956 \pm 0.0010$  \\
$0.170$ &$0.4865  \pm 0.0020$   &$0.1200$&$0.4929 \pm 0.0019$  \\
$0.200$ &$0.4746  \pm 0.0024$   &$0.1250$&$0.4914 \pm 0.0014$  \\
$0.215$ &$0.4626  \pm 0.0026$   &$0.1300$&$0.4896 \pm 0.0019$  \\
$0.220$ &$0.4575 \pm 0.0020$    &$0.1350$&$0.4873 \pm 0.0026$  \\
$0.225$ &$0.4509 \pm 0.0010$    &$0.1400$&$0.4842 \pm 0.0036$  \\
$0.230$ &$0.4405 \pm 0.0021$    &$0.1450$&$0.4796 \pm 0.0063$  \\
$0.235$ &$0.4308 \pm 0.0019$    &$0.1500$&$0.4741 \pm 0.0052$  \\
$0.240$ &$0.4138 \pm 0.0046$    &$0.1525$&$0.4703 \pm 0.0030$  \\
$0.245$ &$0.3797 \pm 0.0080$    &$0.1550$&$0.4651 \pm 0.0040$ \\
$p_3^*$ &$0.2941 \pm 0.0091$   &$0.1575$&$0.4551 \pm 0.0025$ \\
$           $  &$                                    $    &$0.1585$&$0.4463 \pm 0.0032$ \\
$           $  &$                                    $    &$0.1595$&$0.4290 \pm 0.0045$ \\
$           $  &$                                    $    &$0.1600$&$0.4079 \pm 0.0062$ \\
$           $  &$                                    $    &$0.16013$&$0.3977\pm 0.0056$ \\
\hline\hline
\end{tabular}
}
{
\caption{Estimates of $\sigma^*(p)$ for $(d,s)=(3,2)$ and $(d,s)=(4,2)$}
}
\end{center}
\label{table32}  
\end{table}
}

%%%%%%%%%%%%%%%%%%%%%%
\begin{figure}[t!]
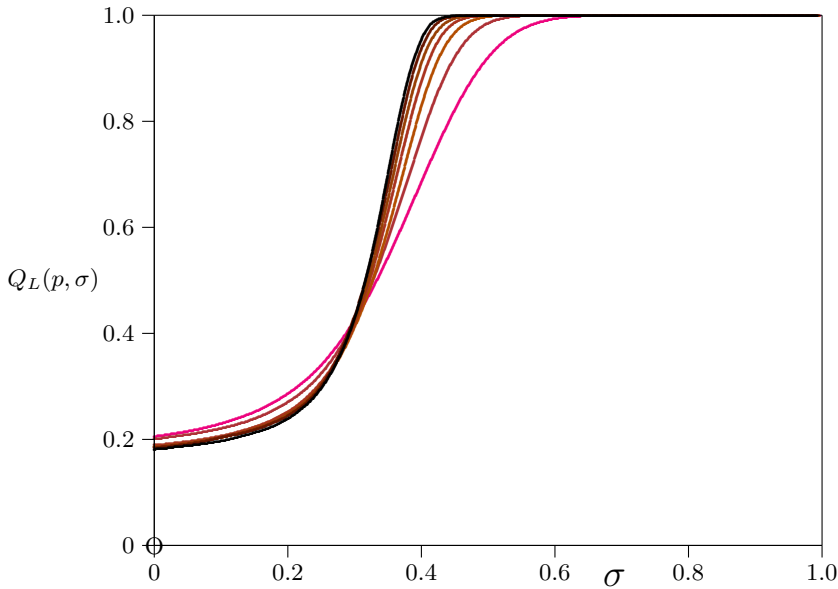

\centering\hfill
\input figure32.tex
\caption{Critical behaviour in the model for $p=p_c(3)=0.24881182(10)$
as a function of $\sigma$ for $d=3$ and $s=2$.  The curves 
are numerical estimates of the probability $Q_L(p,\sigma)$, for 
$L\in \{5,10,15,20,25,30,35\}$.  Each curve was computed
by taking the average of $30000$ realisations of bulk 
percolation clusters at $p=p_c(3)$ and then computing
the average $Q_L(p,\sigma)$ as a function of $\sigma$.
The intersections between the curves is a numerical
estimate of the location of the critical point $\sigma^*(p_c(3))$.}
\label{figure32} 
\end{figure}
%%%%%%%%%%%%%%%%%%%%%%

%%%%%%%%%%%%%%%%%%%%%%%%
\begin{figure}[t!]
\centering\hfill
\input figuresigma32.tex
\caption{Numerical estimates of the location of the
critical curve $\sigma^*(p)$ as a function of $p$ 
for $(d,s)=(3,2)$
are indicated by the data points ($\bullet$).  The solid
curve is an interpolation curve drawn through the data
points.  }
\label{figuresigma32}
\end{figure}
%%%%%%%%%%%%%%%%%%%%%%%%

If the data at $L=5$ are dropped, then a similar analysis show
that $\sigma^*(0.1) = 0.49879\pm 0.00059$.  Similarly, dropping both
$L=5$ and $L=10$ from the analysis gives 
$\sigma^*(0.1)=0.499081\pm0.00050$.
Comparing these results show that there is no improvement in the
statistical estimate by dropping data at small values of $L$, and so we
take as our best estimate the result when dropping $L=5$ from the
analysis, namely $\sigma^*(0.1) = 0.49879 \pm 0.00059$. 

The curves in figure \ref{figureQLsigma01-32} show a systematic 
drift towards the right with removing data at the smallest values
of $L$. We estimate a systematic error in the data by taking twice the
absolute difference between the estimate over all the data and
the estimate with the data at $L=5$ removed.  This gives
$\sigma^*(0.1) = 0.49879 \pm 0.00059 \pm 0.00042$
where the last error bar is an estimated systematic error.

By adding the two error bars our best estimate is obtained,
$\sigma^*(0.1) = 0.4988 \pm 0.0011$

A similar approach at $p=0$ gives the results
$\sigma^*(0) = 0.49986 \pm 0.00026$ over all the data and 
$\sigma^*(0) = 0.50003 \pm 0.00034$ if the data at $L=5$ is
ignored. This gives our best estimate
$\sigma^*(0) = 0.50003 \pm 0.00034 \pm 0.00034$ so that
by combining the error bars,
$\sigma^*(0)=0.50003 \pm 0.00068$
(consistent with the exact value for bond percolation in the
square lattice \cite{H60,K80}; see proposition \ref{prop-sigstar}(b)).

Similar analysis can be done at other values of $p$ and the
results are tabulated in table \ref{table32}.  The stated
error bar is the sum of the statistical and systematic error.
In figure \ref{figuresigma32} the results are plotted in the
$(p, \sigma)$-plane.  The critical curve varies slowly with $p$ for
small $p$, but decreases quickly for $p$ approaching
$p_c(3)$.

An interesting situation arises when $p=p_c(3)$. Simulations 
for $d=3$ and $s=2$ can be done with $p=0.24881182 
= p_3^* \approx p_c(3)$, very close to the critical point 
(the uncertainty is only in the last digit) for percolation in 
the cubic lattice ($d=3$) \cite{WZZGD13}, 
see reference \cite{LZ98}.   In
figure \ref{figure32} estimates of $Q_L(\sigma,p_c(d))$
are plotted against $\sigma$ for $L$ taking values in
$\{5,10,15,20,25,30,35\}$ for $d=3$ and $s=2$.

Minimizing $\C{E}^2(\sigma)$
over all the data gives $\sigma^*(p_3^*) = 0.2949 \pm 0.0053$ and 
if the data point at $L=5$ is dropped, then 
$\sigma^*(p_3^*) = 0.2941 \pm 0.0075$.  This gives the best
estimate $\sigma^*(p_3^*) = 0.2941 \pm 0.0075 \pm 0.0016$.
Combining the error bars give
$\sigma^*(p_3^*) = 0.2941 \pm 0.0091$.  

Notice that the numerical result for $\sigma^*(p_3^*)$ rules out the 
critical bulk percolation density at $p_c(3) = 0.24881182(2)$ in its
error bars.  Since we do not know that $\sigma^*(p)$ is 
left-continuous at $p=p_c(3)$ this result cannot be interpreted 
as evidence that $\lim_{p\to p_c(3)^-} \sigma^*(p)> \sigma^*(p_c(3))$
-- this is so in particular also because of the steepness of $\sigma^*(p)$
as $p$ approaches $p_c(3)$ from below, as seen in figure 
\ref{figuresigma32}.

Numerical simulations of the model with $d=4$ and $s=2$ 
were also done for $L\in\{5,8,10,12,15,20,25,30\}$ and a select
set of values of the bulk density $p$ approaching $p_c(4)$.
In the case that $p=0.05$ a plot of $Q_L(\sigma)$ against
$\sigma$ is similar to figure \ref{figure01-32}.  Minimizing 
$\C{E}^2(\sigma)$ gives an estimate of the critical point by 
locating the narrowest waist in the set of crossing curves.
Over all the data this gives $\sigma^*(0.05) = 0.50043\pm 0.00070$
and if the data point at $L=5$ is dropped, $\sigma^*(0.05)
=0.50070\pm 0.00067$.  Computing a systematic error as before
by doubling the absolute difference between the estimates gives
a best estimate of $\sigma^*(0.05) = 0.50070\pm 0.00067
\pm 0.00054$ and combining the error bars gives the 
$\sigma^*(0.05) = 0.5007 \pm 0.0013$.

%%%%%%%%%%%%%%%%%%%%%%
\begin{figure}[t!]
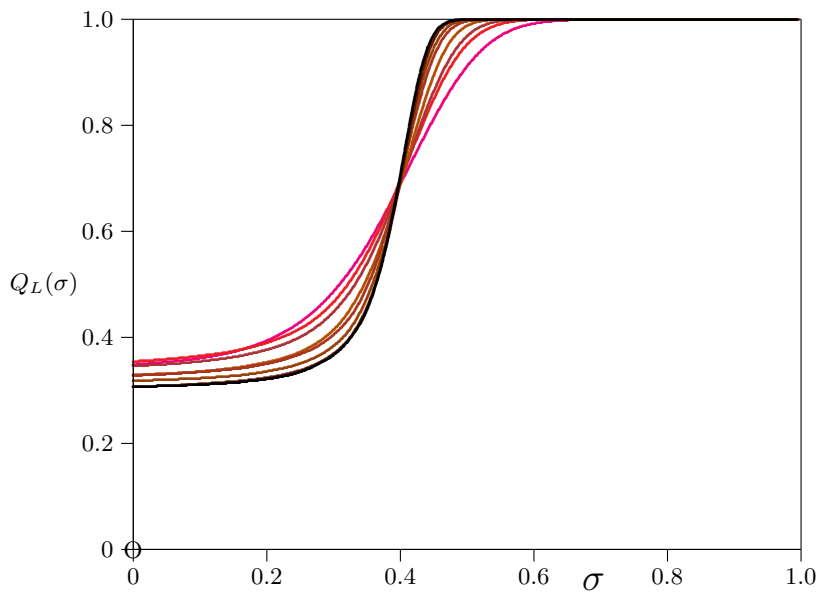

\centering\hfill
\input figure42.tex
\caption{Critical behaviour in the model for $p=p_c(4)=0.160130$
as a function of $\sigma$ for $d=4$ and $s=2$.  The curves 
are numerical estimates of the probability $Q_L(p,\sigma)$, for 
$L\in \{5,8,10,15,20,25,30,35\}$.  Each curve was computed
by taking the average of $20000$ realisations of bulk 
percolation clusters at critical density $p_c(4)$
then determining $Q_L(p,\sigma)$ as a function of $\sigma$.
The intersections of the curves for larger $L$ is a numerical
signal of critical behaviour in this model.}
\label{figure42} 
\end{figure}
%%%%%%%%%%%%%%%%%%%%%%%%%

%%%%%%%%%%%%%%%%%%%%%%
\begin{figure}[b!]
\centering\hfill
\input figuresigma42.tex
\caption{Numerical estimates of the location of the
critical curve $\sigma^*(p)$ as a function of $p$ 
for $(d,s)=(4,2)$
are indicated by the data points ($\bullet$).  The solid
curve is an interpolation curve drawn through the data
points. }
\label{figuresigma42} 
\end{figure}
%%%%%%%%%%%%%%%%%%%%%%

The estimates in table \ref{table32} for other values of
$p\in (0,p_c(4))$ were similarly estimated.  In each
case $Q_L(\sigma)$ was computed over $30000$
realisations of bulk clusters.

We have also performed simulations at $p_4^* = 0.16013$
which is the best estimate for the critical point $p_c(4)$
\cite{PZS01}. Minimizing $\C{E}^2(\sigma)$
over all the data gives $\sigma^*(p_4^*) = 0.3962 \pm 0.0031$ and 
if the data point at $L=5$ is dropped, then 
$\sigma^*(p_4^*) = 0.3977 \pm 0.0025$.  This gives the best
estimate $\sigma^*(p_3^*) = 0.3977 \pm 0.0025 \pm 0.0031$.
Combining the error bars give
$\sigma^*(p_4^*) = 0.2941 \pm 0.0056$.  

The critical curve $\sigma^*(p)$ against $p$ for $(d,s)=
(4,2)$ is plotted in figure \ref{figuresigma42}.

Similar to the case for $(d,s)=(3,2)$ the numerical data 
for $(d,s)=(4,2)$ suggest that $\sigma^*(p_4^*)
> p_c(4) \approx 0.160130\pm 0.000003$ \cite{PZS01}.
The estimate at $p_4^*$ is far larger than $p_c(4)$,
but as above this cannot be interpreted as evidence that
$\lim_{p\to p_c(4)^-} \sigma^*(p)> \sigma^*(p_c(4))$.

\section{Conclusions}
\label{sectionConc}     %%ZXZ[section7]

In this paper we have generalised homogeneous percolation in $\Lattice$ to
a model of inhomogeneous percolation in a $d$-dimensional $\Lattice$
with an $s$-dimensional defect plane.  We showed that there is a surface
transition in this model, as proposed by references
\cite{CCRS81,DB79,DBE81,DBE00,DBL93}.  There is a critical curve $\sigma^*(p)$
for $p\in[0,1]$, with the properties that $\sigma^*(p) > p_c(d)> 0$ for $p< p_c(d)$ while
$\sigma^*(p)=0$ if $p>p_c(d)$, and that $\sigma^*$ is a strictly decreasing function 
of $p$ for $p<p_c(d)$ (see propositions \ref{prop-sigstar} and \ref{prop.cubic}).  
It follows that $\sigma^*$ is discontinuous at $p=p_c(d)$.  We expect that $\sigma^*$
is continuous for $p<p_c(d)$, but we have not yet proven this.

We have also examined the nature of the surface transition in this model.
We investigated the three phases:  the subcritical phase $\C{R}_0$ in which all clusters are
finite, the surface supercritical phase $\C{R}_L$ in which the infinite cluster stays near the
defect surface, and the bulk supercritical phase $\C{R}_H$ in which the infinite cluster 
permeates the whole lattice.
We generalised the differential inequalities of homogeneous percolation
\cite{AB87} to the model here and showed (theorem \ref{theorem5}) 
that the susceptibility $\chi^{I}(p,\sigma)$ is infinite  if and only if 
% either $\theta^I(p,\sigma)>0$ or if $\theta^I(p,\sigma)=0$ and 
$\theta^I(p+\delta,\sigma+\delta)>0$
for all small $\delta>0$ (which happens whenever 
$(p,\sigma)$ is not in the interior of the subcritical phase $\C{R}_0$).

In section \ref{section4} we considered the exponential decay of the 
cluster size distribution in the subcritical phase.  We show that the
cluster size distribution decays exponentially 
(see theorem \ref{theoremZ6}) in the subcritical phase (i.e., when $p<p_c(d)$ 
and $\sigma<\sigma^*(p)$).   In contrast, theorems \ref{theorem5-12} and
\ref{thm999} prove subexponential decay of the cluster size distribution in 
the supercritical phases.  Our lower bound for $P^I_{p,\sigma}(|C|=n)$
in $\C{R}_H$ is $\exp(-c n^{(d-1)/d})$, the same as for supercritical
homogeneous percolation.  However, in $\C{R}_L$, where the infinite
cluster stays close to the defect plane and looks $s$-dimensional,
our lower bound (neglecting a logarithmic term) is $\exp(-c n^{(s-1)/s})$.
We expect that these lower bounds are essentially optimal in both
supercritical phases, although we have not attempted to prove the 
corresponding (more challenging) upper bounds.

We examined $\zeta^I(p,\sigma)$ and $\psi^I(p,\sigma)$, the exponential decay
rates of the cluster size distribution (where ``size'' is measured by vertices for
$\zeta^I$ and by edges for $\psi^I$).  We related these functions to
the free energy of a model of collapsing lattice animals
interacting with a defect plane.  We showed that the existence of the free energy 
in the animal model implies the existence of $\zeta^I$ and $\psi^I$, and we showed
that the percolation transition had implications about non analyticity of the free 
energy.

Finally we performed a numerical study of inhomogeneous percolation
using the Newman-Ziff algorithm.   We plotted the box crossing 
probability $Q_L(p,\sigma)$ as a function of $\sigma$ for various values
of $p\in[0,p_c(d)]$.  Table \ref{table32} shows these results for 
$(d,s)$=(3,2) and for $(d,s)$=(4,2), and includes an error bar associated 
with each estimated $\sigma^*(p)$ value.

For both $d=3$ and $d=4$ we find qualitatively similar phase boundaries.
In both cases the curves start at $\sigma_c(0)=\frac{1}{2}$ and 
decreases with increasing $p$.  On approach to $p_c(d)$ the
critical curve becomes sensitive to even small changes in $p$, and
$\sigma^*(p)$ is discontinuous at $p=p_c(d)$.   There are numerous
open question about the function $\sigma^*(p)$, regarding its continuity
and rate of decrease with increasing $p<p_c(d)$. 

\vspace{1cm}
\section*{Acknowledgements}
EJJvR and NM acknowledge support in the form of NSERC Discovery Grants from
the Government of Canada.  We also thank Geoffrey Grimmett for some helpful
email correspondence.

\bigskip

\appendix

\section{Differential Inequalities for Inhomogeneous Percolation}
\label{appendix1}

In this appendix our aim is to prove the following differential inequalities
\begin{eqnarray}
%\fl
\vec{q} \cdot \nabla \theta^I (p,\sigma,\gamma)
\;\leq\; 2d\, \chi^H(p)\,\theta^I(p,\sigma,\gamma) 
(1-\gamma) \frac{\partial \theta^I}{\partial \gamma} , & & 
\label{eqn37ZZ}    %%ZXZ[eqn37ZZ]
\\
%
%\fl 
\theta^I(p,\sigma,\gamma) 
\;\leq\; \gamma\, \frac{\partial\theta^I}{\partial \gamma}
+  \L\theta^I(p,\sigma,\gamma)\R^2
+ \chi^H(p)\, \theta^I(p,\sigma,\gamma)
\L \vec{p}\cdot \nabla \theta^I(p,\sigma,\gamma) \R  & &
\label{eqn38ZZ}    %%ZXZ[eqn38ZZ]
\end{eqnarray}
where $\nabla = (\sfrac{\partial}{\partial p},
\sfrac{\partial}{\partial \sigma})$ and $p\leq \sigma$.  These are equations
\Ref{eqn37} and \Ref{eqn38}.

These inequalities are due to Aizenman and Barsky \cite{AB87} 
for homogeneous bond percolation (the proofs can be found in
reference \cite{G99} as well\setplotsymbol ({\footnotesize $\bullet$})), and we 
adapt their proofs to the inhomogeneous model.

Let $B(N)$ be the box of side-length $2N$ centered at the origin
of $\Lattice$.  Denote by $V(N)$ the vertices in $B(N)$. 
We shall prove the above differential inequalities in $B(N)$, 
and then take $N\to\infty$.

Impose periodic boundary conditions on $B(N)$ by identifying its 
opposite faces. Denote this finite lattice by $\Lattice(N)$. Denote the intersection 
of $\Lattice(N)$ and $\Lattice_0$ by $\Lattice_0(N)$.  Denote the set 
of edges in $\Lattice_0(N)$ by $\Edges_0(N)$ and the set of edges in 
$\Lattice(N)$ by $\Edges(N)$.  

As before, edges in $\Edges(N)\setminus \Edges_0(N)$ are open with the bulk 
probability $p$, and edges in $\Edges_0(N)$ are open with surface 
probability $\sigma$.

The open cluster in $\Lattice(N)$ at the vertex $x$ is $C_N(x)$, 
and the open cluster at the origin is $C_N(0) \equiv C_N$.
For $y\in V(N)$, define the susceptibility $\chi_N^I (p,\sigma;y) 
= E^I_{p,\sigma}  |C_N(y)| $. 

Introduce the ghost vertex $g$ and edges $\Edges_g (N) = 
\{g{\sim}x\vert \hbox{ $x\in V(N)$}\}$. 
Edges in $\Edges_g(N)$ are open with probability $\gamma\in (0,1)$.
Define $G_N$ to be the collection of vertices in $\Lattice(N)$
adjacent to $g$ through open edges in $\Edges_g(N)$.

For inhomogeneous percolation on $\Edges(N) \cup \Edges_g(N)$ with parameters
$(p,\sigma,\gamma)$, let $P^I_{p,\sigma,\gamma}$ and $E^I_{p,\sigma,\gamma}$ be the
corresponding probability measure and expectation. For $y\in \Lattice(N)$, we define the
associated quantities 
\begin{eqnarray}
\theta^I_N(p,\sigma,\gamma;y) 
     &\;=\;& P^I_{p,\sigma,\gamma} (C_N(y)\cap G_N \not= \emptyset)\\
\chi^I_N(p,\sigma,\gamma;y) 
     &\;=\;& E^I_{p,\sigma,\gamma} \L |C_N (y) | 
                           \indic\{ C_N(y) \cap G_N = \emptyset \}\R .
\label{eqnAA4}     %%ZXZ[eqnAA4]
\end{eqnarray}
That is, $\theta^I_N(p,\sigma,\gamma;y)$ is the probability that there
is an open path from $y$ to $g$.  Also write
\[    \theta^I_N(p,\sigma,\gamma) \,:=\, \theta^I_N(p,\sigma,\gamma;0)
      \hspace{5mm}\mbox{and}\hspace{5mm}
      \chi^I_N(p,\sigma,\gamma) \,:=\, \chi^I_N(p,\sigma,\gamma;0)\,.
\]
Then $\theta_N^I(p,\sigma,\gamma) 
 = \theta_N^I(p,\sigma,\gamma;y)$
for all $y \in \Lattice_0(N)$, and similarly for $\chi^I_N$.  
We shall frequently simplify notation by
leaving out the arguments when there is no risk of ambiguity, 
e.g.\ $\theta_N^I \equiv \theta_N^I(p,\sigma,\gamma)$.
We have
%Edges in $\Edges_g(N)$ are open with probability $\gamma$, so it follows that
\begin{eqnarray*}
\theta^I_N (p,\sigma,\gamma) &\;=\;& 1 - \sum_{n=1}^\infty (1-\gamma)^n\,
P^I_{p,\sigma}(|C_N| = n)     \hspace{5mm}\hbox{and}
\\
\chi^I_N (p,\sigma,\gamma) &\;=\;& \sum_{n=1}^\infty n\,(1-\gamma)^n 
P^I_{p,\sigma} (|C_N|=n) .
\end{eqnarray*}
We immediately obtain the following analogue of Equation \Ref{eqn35}:
\begin{equation} \chi^I_N (p,\sigma,\gamma) = (1-\gamma) 
\frac{\partial \theta^I_N}{\partial \gamma} .
\label{eq.lemmaA2}   %%ZXZ[lemmaA2]
\end{equation}

The functions $\theta^I (p,\sigma,\gamma)$ and $\chi^I (p,\sigma,\gamma)$ 
are defined in the usual way for the infinite lattice $\Lattice$ with the ghost vertex
$g$ and edges $\Edges_g$, and with $\Lattice_0$ as defined before:
\begin{eqnarray*}
\theta^I(p,\sigma,\gamma) &\;=\;&
\theta^I(p,\sigma,\gamma;0)\;=\; P^I_{p,\sigma,\gamma} (|C| = \infty) \\
\chi^I(p,\sigma,\gamma) &\;=\;&
\chi^I(p,\sigma,\gamma;0) \;=\; E^I_{p,\sigma,\gamma}(|C|\indic\{C\cap G=\emptyset\})
\end{eqnarray*}
where $C$ is the cluster at the origin.  Note that for $\gamma>0$, 
$C \cap G$ is not empty with probability one when $|C|=\infty$.

The proof of the following lemma is similar to the proof
for homogeneous percolation in appendix I of reference \cite{G99}.

\begin{lemma}
For all $\gamma \in (0,1)$ and $p,\sigma\in(0,1)$,
\[ \lim_{N\to\infty} \theta^I_N (p,\sigma, \gamma) 
                                   \;=\; \theta^I(p,\sigma,\gamma),\]
and similarly,
\[ \lim_{N\to\infty} \frac{\partial \theta^I_N}{\partial p}
\;=\; \frac{\partial \theta^I}{\partial p},\q
\lim_{N\to\infty} \frac{\partial \theta^I_N}{\partial \sigma}
\;=\; \frac{\partial \theta^I}{\partial \sigma},\,\;\hbox{and}\,\;
\lim_{N\to\infty} \frac{\partial \theta^I_N}{\partial \gamma}
\;=\; \frac{\partial \theta^I}{\partial \gamma} .\]
\qed
\label{lemmaA1}   %%ZXZ[lemmaA1]
\end{lemma}

By equations  \Ref{eqn35} and \Ref{eq.lemmaA2} and lemma \ref{lemmaA1}, we have
\begin{equation}
\lim_{N\to\infty} \chi^I_N (p,\sigma,\gamma) \;=\; \chi^I(p,\sigma,\gamma) .
\label{eqnAchi}
\end{equation}

\subsection{Uniform bounds for vertex-dependent functions}

Let $P^H_{p,\gamma}$ and $E^H_{p,\gamma}$ be the probability measure and 
expectation for homogeneous percolation with a ghost field.
Let $\chi_N^H(p,\gamma) 
= E^H_{p,\gamma} (|C_N|\indic\{ C_N\cap G_N = \emptyset\})$ be the
corresponding susceptibility in $B(N)$.  
Also let $\chi_N^H(p) = \chi^H_N(p,0)$.
% Putting $\sigma = p$ then shows that $\chi_N^I(p,p,\gamma) 
% \to \chi^H(p,\gamma)$ as $N\rightarrow\infty$, which in turn
% converges to $\chi^{f,H}(p)$ as $\gamma\to 0^+$ (see for example
% equation \Ref{eqn35a}).  Notice that $\chi^{f,H}(p) \leq \chi^H(p)$. 

Proving the differential inequalities requires the following useful lemma. 

\begin{lemma}
Assume $p\leq \sigma$ and $y\in B(N)$.  Then 
\begin{itemize}
\item[(a)] $\theta_N^I (p,\sigma,\gamma;y)
 \;\leq\; \chi_N^H(p) \, \theta_N^I (p,\sigma,\gamma)$ for every 
$\gamma\in (0,1)$;
\item[(b)] $\chi_N^I(p,\sigma;y) \;\leq\; \chi_N^H(p)
 \,\chi_N^I(p,\sigma;0)$   (where $\chi_N^I(p,\sigma;y):=\chi_N^I(p,\sigma,0;y)$).
\end{itemize}
\label{lemmaA6}    %%ZXZ[lemmaA6]
\end{lemma}

\begin{proof}
Since $\chi_N^H(p) \geq 1$, it is only necessary to consider the
case that $y\not\in\Lattice_0(N)$.

\smallskip
\noindent{(a)}
Suppose that $C_N(y) \cap G_N \not= \emptyset$.  Then there is an 
open path from $y$ to a point of $G_N$.  This path either uses no 
edge of $\Edges_0(N)$, or there exists an open self-avoiding
path $\pi$ from $y$ to a vertex of $G_N$ which passes through an edge
of $\Edges_0(N)$.  In the latter case, let $z$ be the earliest point
of $\pi$ that is an endpoint of a bond in $\pi\cap \Lattice_0(N)$.  Then $z\in \Lattice_0(N)$,
and the part of $\pi$ from $y$ to
$z$ is disjoint from the part of the path from $z$ to a vertex
of $G_N$.

We formalize the above as follows.  For given fixed $y$ and for each 
$z\in\Lattice_0(N)$, define the events
\begin{itemize}
\item $\W{A}_N$ is the event that there is an open path from $y$ to
$G_N$ in $\Edges(N)\setminus \Edges_0(N)$;
\item $\W{D}_N(z)$ is the event that there is an open path from 
$y$ to $z$ in $\Edges(N)\setminus\Edges_0(N)$;
\item $D^*_N(z)$ is the event that there is an open path from
$z$ to $G_N$ in $B(N)$.
\end{itemize}
Observe that 
\begin{equation}
    \label{eq.ppthth}
   %\fl   \hspace{10mm}
     P^I_{p,\sigma,\gamma} (\W{A}_N) \;=\; P_{p,\gamma}^H(\W{A}_N) \; \leq\;
    \theta_N^I(p,p,\gamma) \;\leq\; \theta_N^I(p,\sigma,\gamma) 
    \hspace{5mm}\hbox{(since $p\leq \sigma$)} \,.
\end{equation}

Using standard percolation notation, 
$\W{D}_N(z) \circ D^*_N(z)$ is the event that $\W{D}_N(z)$ and
$D^*_N(z)$ \textit{occur disjointly}---that is, there exist two disjoint
sets of open edges such that the first set guarantees occurrence of
$\W{D}_N(z)$ and the second set guarantees occurrence of
$D^*_N(z)$.  

We then observe that  
\begin{equation}
  \label{eq.AUDD}
   \LC  C_N(y) \cap G_N \not= \emptyset \RC \subset
  \W{A}_N \cup \bigcup_{z\in\Lattice_0(N)} \L \W{D}_N(z) \circ D^*_N(z) \R
\end{equation}
since occurrence of of the event $\LC  C_N(y) \cap G_N \not= \emptyset \RC$
implies that either $\W{A}_N$ occurs or $\L \W{D}_N(z) \circ D^*_N(z) \R$
occurs for some $z\in \Lattice_0(N)$. 

From Equations \Ref{eq.AUDD} and \Ref{eq.ppthth} and the BK Inequality \cite{G99}, we see that
\begin{eqnarray*} 
& & \hspace{-1cm}
P^I_{p,\sigma,\gamma}\L  C_N(y) \cap G_N \not= \emptyset \R \\
&\;\leq\;& P^I_{p,\sigma,\gamma} \L \W{A}_N \R
     \;+ \sum_{z\in \Lattice_0(N)} P^I_{p,\sigma,\gamma}
         \L \W{D}_N(z) \circ D^*_N(z)  \R \\
&\;\leq\;& \theta^I_N (p,\sigma,\gamma) 
     \;+ \sum_{z\in\Lattice_0(N)}
        P^I_{p,\sigma,\gamma} ( \W{D}_N(z) ) \,
           P^I_{p,\sigma,\gamma} ( D^*_N(z) ) \\
&\;\leq\;&  \theta^I_N (p,\sigma,\gamma) 
     \;+ \sum_{z\in\Lattice_0(N)}
        P^H_{p} ( z\in C_N(y) ) \,
           P^I_{p,\sigma,\gamma} ( C_N(z) \cap G_N \not=\emptyset ) \\
&\;=\;&  \L 1 \;+ \sum_{z\in\Lattice_0(N)}
        P^H_{p} ( z\in C_N(y) )  \R \theta^I_N (p,\sigma,\gamma)  \\
&\;\leq\;&  \chi_N^H(p)\, \theta^I_N (p,\sigma,\gamma)  .
\end{eqnarray*}

\noindent{(b)}  Fix $y\in \Lattice(N)\setminus \Lattice_0(N)$.  
Let $\W{D}_N(z)$ be defined as in part (a).   For $w\in B(N)$, let 
$\LC y \leftrightarrow_N w \RC$ denote the event
that $y$ and $w$ are connected in $B(N)$ by a path of open edges.

Similarly to the proof of part (a), we see that for each $w\in B(N)$ 
\begin{eqnarray}
P^I_{p,\sigma}\L y \leftrightarrow_N w \R
&\;\leq\;&  P^I_{p,\sigma} \L \W{D}_N (w) \R
    \; + \sum_{z\in \Lattice_0(N)} P^I_{p,\sigma}
         \L \W{D}_N(z) \circ \LC z\leftrightarrow_N w\RC  \R 
    \label{eq.partb} \\
&\;\leq\;&  P^I_{p,\sigma} \L \W{D}_N (w) \R  
     \; + \sum_{z\in \Lattice_0(N)} 
      P^I_{p,\sigma} \L \W{D}_N(z) \R \,
        P^I_{p,\sigma} \L \LC z\leftrightarrow_N w\RC  \R 
 \nonumber \\
&\;\leq\;&  P^H_{p} \! \L \LC y \leftrightarrow_N w \RC \R
     \, + \!\!\! \sum_{z\in \Lattice_0(N)} \!
      P^H_{p} \L \LC y \leftrightarrow_N z \RC \R \,\!
        P^I_{p,\sigma} \L \LC z\leftrightarrow_N w\RC  \R .
\nonumber
\end{eqnarray}

%%%%%%%% I think we should keep this to make ...%%%But it was wrong!!
%Next, observe that since we assume $p \leq \sigma$, we have
%$\chi_N^H(p) \equiv \chi_N^H(p;0) \leq \chi_N^I(p,\sigma;0)
%\leq \chi^I(p,\sigma)$.
%Moreover, if $z\in \Lattice_0 (N)$, then
%$P^I_{p,\sigma} \L \LC z\leftrightarrow_N w\RC \R
%\leq P^I_{p,\sigma} \L \LC z\leftrightarrow w\RC \R$
%so that
%\[ \sum_{w\in B(N)} P^I_{p,\sigma} \L \LC z\leftrightarrow_N w\RC \R
%\leq \chi^I_N(p,\sigma) \]
%by using the periodic boundary conditions on $B(N)$.
%%%%%%%%% ... the stuff following below clear? %%%%%

Since $p\leq \sigma$, we have $\chi_N^H(p) \leq \chi_N^I(p,\sigma;0) =
\chi^I_N(p,\sigma;z)$ for every $z\in \Lattice_0(N)$.  Using this and
summing equation \Ref{eq.partb}  over $w$, we obtain
\begin{eqnarray*}
& & \hspace{-18mm}
\chi_N^I(p,\sigma;y) \;=\; \sum_{w\in B(N)} 
                     P^I_{p,\sigma}\L y \leftrightarrow_N w \R \\ 
& &  \hspace{-2cm}
   \;\leq\;  \sum_{w\in B(N)}    P^H_{p} \! \L \LC y \leftrightarrow_N w \RC \R
   %P^H_{p,\sigma} \L \W{D}_N (w) \R
     + \sum_{w\in B(N)} \sum_{z\in \Lattice_0(N)} 
       P^H_{p} \L \LC y \leftrightarrow_N z \RC \R \,
        P^I_{p,\sigma} \L \LC z\leftrightarrow_N w\RC \R \\
& &  \hspace{-2cm}
   \;=\;\chi_N^H(p) \;+ \sum_{z\in\Lattice_0(N)}
        P^H_{p} \L \LC y \leftrightarrow_N z \RC \R \,
          \chi_N^I(p,\sigma;z) \\
& & \hspace{-2cm}
  \leq \L 1+  \sum_{z\in\Lattice_0(N)}
        P^H_{p} \L \LC y \leftrightarrow_N z \RC \R \R 
          \chi_N^I(p,\sigma;0)   \;\leq\; \chi_N^H (p)\,\chi^I_N(p,\sigma;0) .
\end{eqnarray*}
This completes the proof.
\end{proof}

\subsection{The first differential inequality}

The first differential inequality is defined in terms of 
$\theta_N^I \equiv \theta^I_N(p,\sigma,\gamma)$ as follows.

\begin{lemma}
For $p,\sigma,\gamma\in(0,1)$, let
$\vec{p} = (p,\sigma)$ and $\nabla \equiv \L \sfrac{\partial}{\partial p},
\sfrac{\partial}{\partial \sigma}\R$. Define $\vec{q} = \vec{1}-\vec{p}
= (1-p,1-\sigma)$.  If $p \leq \sigma$, then  
% for $\theta_N^I \equiv \theta_N^I(p,\sigma,\gamma)$ 
we have 
\[\vec{q} \cdot \nabla \theta_N^I
\;\leq\; 2d\, \chi^H_N(p)\,
(1-\gamma) \,\theta^I_N\,\frac{\partial \theta_N^I}{\partial \gamma} .
\]
\label{lemmaA7}   %%ZXZ[lemmaA7]
\end{lemma}

\begin{proof}
The proof is similar to the proof for homogeneous percolation
(see for example reference \cite{G99}) and proceeds by applying
Russo's formula to the event $\{ C_N \cap G_N \not= \emptyset\}$, 
conditioned on $G_N$. 

Let $\Gamma$ be a realisation of $G_N$, 
i.e.\ a subset of vertices of $B(N)$ .
The event $A_N (\Gamma) = \{ C_N \cap \Gamma \not=\emptyset\}$
is increasing. Hence by Russo's formula \cite{G99},
\begin{eqnarray}
\sfrac{\partial }{\partial p} P_{p,\sigma}^I (A_N(\Gamma))
&\;=\;& \sum_{e\in\Edges(N)\setminus \Edges_0(N)}
P_{p,\sigma}^I (\hbox{$e$ is pivotal for $A_N(\Gamma)$}),
\label{eqn38A}   %%ZXZ[eqno38A] 
\\
\sfrac{\partial }{\partial \sigma} P_{p,\sigma}^I (A_N(\Gamma))
&\;=\;& \sum_{e\in\Edges_0(N)}
P_{p,\sigma}^I (\hbox{$e$ is pivotal for $A_N(\Gamma)$}).
\label{eqn38B}   %%ZXZ[eqno38B]
\end{eqnarray}

First consider equation (\ref{eqn38A}).  
An edge $e=x{\sim}y$ is pivotal for  $A_N (\Gamma)$
if and only if the following all occur in $\Edges(N) \setminus \{e\}$:
(1) there is no open path from the origin
to $\Gamma$, (2) exactly one of $x$ and $y$ is joined to the
origin by an open path, and (3) the other vertex is joined
to a vertex of $\Gamma$ by an open path.  Hence, 
\begin{eqnarray*}
& &(1-p)\sfrac{\partial }{\partial p} P_{p,\sigma}^I (A_N(\Gamma)) \\
&\;=\;& \sum_{e\in\Edges(N)\setminus \Edges_0(N)}
P_{p,\sigma}^I (\hbox{$e$ is closed})
P_{p,\sigma}^I (\hbox{$e$ is pivotal for $A_N(\Gamma)$}) \\
&\;=\;& \sum_{x{\sim}y\in\Edges(N)\setminus \Edges_0(N)} P_{p,\sigma}^I
(\hbox{$x\in C_N$, $C_N \cap \Gamma = \emptyset$, 
$C_N(y)\cap \Gamma \not=\emptyset$ }).
\end{eqnarray*}
where the last summation is over all ordered pairs $(x,y)$ of vertices
such that the (undirected) edge $x{\sim}y$ is
in $\Edges(N) \setminus \Edges_0(N)$.  Put $q=1-p$
and average the left hand side of the above over $\Gamma$:
\begin{eqnarray*}
\hspace{-2.4cm}
& & q \sum_{\Gamma} P_{p,\sigma,\gamma}^I (G_N = \Gamma)\,
\sfrac{\partial }{\partial p} P_{p,\sigma}^I (A_N(\Gamma)) \\
&\;=\;& q\sfrac{\partial}{\partial p}
\LH \sum_{\Gamma} P_{p,\sigma,\gamma}^I(C_N \cap\Gamma \not=\emptyset)\,
P_{p,\sigma,\gamma}^I(G_N=\Gamma) \RH \\
&\;=\;& q\sfrac{\partial}{\partial p}
 P_{p,\sigma,\gamma}^I(C_N \cap G_N \not=\emptyset) \\
&\;=\;& q\sfrac{\partial}{\partial p} \theta_N^I(p,\sigma,\gamma) .
\end{eqnarray*}
%This holds because $P_{p,\sigma,\gamma}^I(G_N=\Gamma)\,=\,(1-\gamma)^{|\Gamma|}$,
%and because 
Here it is important that the sum over $\Gamma$ has a finite number of terms.

This shows that
\begin{eqnarray}
q\sfrac{\partial}{\partial p} \theta_N^I & & \nonumber
\;=\; \sum_{x{\sim}y\in\Edges(N)\setminus \Edges_0(N)} 
P_{p,\sigma,\gamma}^I
(\hbox{$x\in C_N$, $C_N \cap G_N  = \emptyset$, 
$C_N(y)\cap G_N \not=\emptyset$ }). 
\label{eqn3.65}    %%ZXZ[eqn3.65]
\end{eqnarray}
Observe that $C_N$ and $C_N(y)$ must be disjoint on the
right hand side.

Exactly the same set of arguments applied to equation
\Ref{eqn38B} gives (with $\tau = 1-\sigma$)
\begin{eqnarray}
\tau\sfrac{\partial}{\partial \sigma} \theta_N^I
\;=\; \sum_{x{\sim}y\in \Edges_0(N)} P_{p,\sigma,\gamma}^I
(\hbox{$x\in C_N$, $C_N \cap G_N  = \emptyset$, 
$C_N(y)\cap G_N \not=\emptyset$ }). 
\label{eqn3.66}    %%ZXZ[eqn3.66]
\end{eqnarray}
Adding the last two equations together then produces
\begin{eqnarray}
\vec{q} \cdot \nabla \theta_N^I
\;=\; \sum_{x{\sim}y\in\Edges(N)} P_{p,\sigma,\gamma}^I
(\hbox{$x\in C_N$, $C_N \cap G_N  = \emptyset$, 
$C_N(y)\cap G_N \not=\emptyset$ }). 
\label{eqn3.67}    %%ZXZ[eqn3.67]
\end{eqnarray}

The right hand side of equation \Ref{eqn3.67}  must be bounded next.  This is done by
conditioning  on the cluster at the origin.  The 
last equation becomes
\begin{eqnarray}
& & 
 \vec{q} \cdot \nabla \theta_N^I 
\label{eqn39A}   %%ZXZ[eqn39A] 
\\
& &
\;=\;  \sum_{x{\sim}y}\LH \sum_{\Xi}  %\ni\{0,x\}}
P_{p,\sigma}^I(\hbox{$C_N = \Xi$})\,
P_{p,\sigma,\gamma}^I(\hbox{$C_N \cap G_N  = \emptyset$, 
$C_N(y)\cap G_N \not=\emptyset$ } \vert \hbox{$C_N = \Xi$}) \RH 
\nonumber
\end{eqnarray}
where the outer sum is over ordered pairs $(x,y)$ such that $x{\sim}y\in
\Edges(N)$, and the inner sum is over all connected graphs $\Xi$ 
containing $\{0,x\}$ and not containing $y$.

Conditioned on $C_N=\Xi$, the events $C_N \cap G_N  = \emptyset$ 
and $C_N(y)\cap G_N \not=\emptyset$ are independent (the first 
depends only on vertices of $\Xi$, and the second depends only 
on vertices and edges outside $\Xi$).
%% -- note also that there are no end-vertices in $\Edges_0(N)$ because of the choice of periodic boundary conditions.  (I don't see why we need this --- Neal)
Hence
\begin{eqnarray*}
& & \hspace{-1.5cm}
 P_{p,\sigma,\gamma}^I (\hbox{$C_N \cap G_N  = \emptyset$, 
$C_N(y)\cap G_N \not=\emptyset$ } \vert \hbox{$C_N = \Xi$}) \\
& &\hspace{-1cm}
\;=\; P_{p,\sigma,\gamma}^I 
       (\hbox{$C_N \cap G_N = \emptyset$}\vert \hbox{$C_N = \Xi$}) \,
P_{p,\sigma,\gamma}^I 
(\hbox{$C_N(y) \cap G_N \not=\emptyset$} \vert \hbox{$C_N = \Xi$}) .
\end{eqnarray*}
The condition $C_N = \Xi$ in the last factor restricts the set of possible open
paths from $y$ to a vertex in $G_N$ (since $y\not\in C_N$).  Hence
\[   
 P_{p,\sigma,\gamma}^I 
(\hbox{$C_N(y) \cap G_N \not=\emptyset$} \vert \hbox{$C_N = \Xi$}) 
\;\leq\; P_{p,\sigma,\gamma}^I (\hbox{$C_N(y) \cap G_N \not=\emptyset$})
\;=\; \theta_N^I(p,\sigma,\gamma; y) . \]
This shows that
\begin{eqnarray}
& & 
 \vec{q} \cdot \nabla \theta_N^I
\nonumber %\label{eqn39A}   %%ZXZ[eqn39A] 
\\
& &
\;\leq\; \sum_{x{\sim}y}\LH \sum_{\Xi}
P_{p,\sigma}^I(\hbox{$C_N = \Xi$})\,
P_{p,\sigma,\gamma}^I 
       (\hbox{$C_N \cap G_N = \emptyset$}\vert \hbox{$C_N = \Xi$})\,
\theta_N^I(p,\sigma,\gamma; y)  \RH \nonumber \\
& &
\;=\; \sum_{x{\sim}y}\LH 
P_{p,\sigma,\gamma}^I (\hbox{$x\in C_N$, $y\not\in C_N$, 
                        $C_N \cap G_N = \emptyset$}) \,
\theta_N^I(p,\sigma,\gamma; y) \RH
\nonumber
\\
%Use lemma \ref{lemmaA6}(a) on $\theta_N^I(q,\sigma,\gamma;y)$ to see that
%\[ \fl \vec{q} \cdot \nabla \theta_N^I (p,\sigma,\gamma)
& &
\;\leq\; \chi^H_N(p)\,\theta_N^I(p,\sigma,\gamma) 
\sum_{x{\sim} y}
P_{p,\sigma,\gamma}^I (\hbox{$x\in C_N$, $y\not\in C_N$, 
                        $C_N \cap G_N = \emptyset$})
        \label{eqn39AA}  
\end{eqnarray}
where the final inequality comes from lemma \ref{lemmaA6}(a).
It remains to bound the last summation.
\begin{eqnarray*}
& & \hspace{-2cm}
\sum_{x{\sim} y}
P_{p,\sigma,\gamma}^I (\hbox{$x\in C_N$, $y\not\in C_N$, 
                        $C_N \cap G_N = \emptyset$}) \\
% &\;\leq\; & \sum_{x{\sim} y} P_{p,\sigma,\gamma}^I (\hbox{$x\in C_N$,
   %                     $C_N \cap G_N = \emptyset$})  \\
&\;\leq\; & 2d \sum_{x} P_{p,\sigma,\gamma}^I (\hbox{$x\in C_N$,
                        $C_N \cap G_N = \emptyset$})  \\
&\;=\;& 2d\, E_{p,\sigma,\gamma}^I \L \LV C_N \RV \, 
                 \indic\{C_N \cap G_N = \emptyset\} \R \\
&\;=\;& 2d\, \chi_N^I(p,\sigma,\gamma) \;\;\;= \;\;\;
2d\,(1-\gamma) \frac{\partial \theta_N^I}{\partial \gamma}
\end{eqnarray*}
by equations  %(\ref{eq.lemmaA2}) and by equation 
\Ref{eqnAA4} and \Ref{eq.lemmaA2}.
Putting this all together then gives
% \[ \vec{q} \cdot \nabla \theta_N^I (p,\sigma,\gamma)
% \;\leq\; 2d\, \chi^H_N(p,0)\,\theta_N^I(p,\sigma,\gamma) 
% (1-\gamma) \frac{\partial \theta_N^I}{\partial \gamma} \]
%which is 
the desired inequality. 
\end{proof}

\subsection{The second differential inequality}

The second differential inequality is the following (again writing
$\theta_N^I$ for $\theta_N^I(p,\sigma,\gamma)$).

\begin{lemma}
For $p,\sigma,\gamma \in(0,1)$ let
$\vec{p} = (p,\sigma)$ and $\nabla \equiv \L \sfrac{\partial}{\partial p},
\sfrac{\partial}{\partial \sigma}\R$. If $p \leq \sigma$, then 
\[
\theta_N^I
\;\leq\; \gamma\, \frac{\partial\theta_N^I}{\partial \gamma}
+  \L\theta_N^I\R^2
+ \chi_N^H(p)\, \theta_N^I
\L \vec{p}\cdot \nabla \theta_N^I \R .
\]
\label{lemmaA8}   %%ZXZ[lemmaA8]
\end{lemma}

\begin{proof}
Observe that
\begin{eqnarray}
%\fl   \hspace{8mm}
\theta_N^I(p,\sigma,\gamma) 
&\;=\;& P_{p,\sigma,\gamma}^I(C_N \cap G_N \not=\emptyset) \nonumber \\
&\;=\;& P_{p,\sigma,\gamma}^I(\LV C_N \cap G_N \RV = 1)
  + P_{p,\sigma,\gamma}^I(\LV C_N \cap G_N \RV \geq 2) .
\label{eqnZ1}    %%ZXZ[eqnZ1]
\end{eqnarray}
The first term in equation \Ref{eqnZ1} can calculated:
\begin{eqnarray}
P_{p,\sigma,\gamma}^I(\LV C_N \cap G_N \RV = 1)
&\;=\;& \sum_{n=1}^\infty n\gamma(1-\gamma)^{n-1} 
  \,P_{p,\sigma,\gamma}^I(\LV C_N \RV=n) \nonumber  \\
&\;=\;& \frac{\gamma\,\chi_N^I(p,\sigma,\gamma)}{1-\gamma} \nonumber \\
&\;=\;& \gamma\, \frac{\partial\theta_N^I}{\partial \gamma}
\label{eqnZ2}    %%ZXZ[eqnZ2]
\end{eqnarray}
by equation \Ref{eq.lemmaA2}.

It remains to bound the second term in \Ref{eqnZ1}.
%Given $\Lattice(N)$, let $\Edges(N)$, $B(N)$ and the set of vertices $V(N)$ 
%be defined as before. 
Define the event
\[ A_x \;=\; \{ \hbox{$x\in G_N$ or $x$ is joined to $G_N$ by an open path} \}.\]
%Realisations of $A_x$ is defined on the set 
%$\prod_{e\in\Edges(N)} \{0,1\} \times \prod_{x\in V(N)}\{ 0,1\}$.
%
%$A_x\circ A_x$ is the set of realisations $\omega$ for which there 
%exist disjoint sets $V_1$ and $V_2$ of vertices in $G_N$, and
%disjoint sets $E_1$ and $E_2$ of open edges in $\Lattice(N)$, such that
%$A_x$ occurs for all realisations $\omega_1$ which agrees with $\omega$
%on $(V_1,E_1)$ and for all realisations $\omega_2$ which agree with
%$\omega$ on $(V_2,E_2)$.
%
Then $A_x \circ A_x$ is the event that there exist two distinct
vertices $v_1$ and $v_2$ in $G_N$ and two edge-disjoint open paths
joining these vertices to $x$.
If $x\in G_N$, then one these paths may be the singleton $x$.

It follows that
\begin{eqnarray}
%\fl
P_{p,\sigma,\gamma}^I (\LV C_N \cap G_N \RV \geq 2)
&\;=\; P_{p,\sigma,\gamma}^I (A_0\circ A_0) & 
  \label{eqnZ1A} \\
&+ P_{p,\sigma,\gamma}^I (\hbox{$\LV C_N \cap G_N \RV \;\geq\; 2 $, and
$A_0\circ A_0$ does not occur})&   .    \nonumber
\end{eqnarray}
By the BK inequality,
\begin{equation}
P_{p,\sigma,\gamma}^I (A_0\circ A_0) \;\leq\;
  \L P_{p,\sigma,\gamma}^I (A_0) \R^2 \;=\; \L\theta_N^I(p,\sigma,\gamma)\R^2.
\label{eqnZ3}    %%ZXZ[eqnZ3]
\end{equation}
The remaining term is the probability that $\LV C_N \cap G_N \RV \geq 2$
but there do not exist two edge-disjoint paths from the origin to distinct
vertices in $G_N$.  If this occurs, then there exists an edge $x{\sim}y$
in $\Edges(N)$ with the following properties:
\begin{itemize}
\item $x{\sim}y$ is open;
\item If $x{\sim}y$ is deleted in $\Lattice(N)$, then three events occur:
\begin{enumerate}
\item there is no open path from the origin to a vertex of $G_N$;
\item $x$ is joined to the origin by an open path;
\item the event $A_y\circ A_y$ occurs.
\end{enumerate}
\end{itemize}
%In addition, $x{\sim}y$ is either in $\Edges_0(N)$ or in 
%$\Edges(N)\setminus\Edges_0(N)$.

The probability that a particular edge $x{\sim}y$ has these properties 
is
\[ pq^{-1}\, P_{p,\sigma,\gamma}^I
(\hbox{$x{\sim}y$ is closed, $x\in C_N$, $C_N\cap G_N =\emptyset$, $A_y\circ A_y$}) \]
if $x{\sim}y \in \Edges(N)\setminus\Edges_0(N)$;  if $x{\sim}y \in \Edges_0(N)$, then we
get the above expression with $\sigma(1-\sigma)^{-1}$ instead of $pq^{-1}$.
Therefore we obtain the bound
\[    P_{p,\sigma,\gamma}^I (\hbox{$\LV C_N \cap G_N \RV \;\geq\; 2 $, and
$A_0\circ A_0$ does not occur})   \;\leq \; S_1\,+\,S_2, \]
where 
%the probability that an edge 
%$x{\sim}y\in \Edges(N)\setminus\Edges_0(N)$ exists is bound above by
\begin{equation}
\hspace{-1.5cm}
S_1 \,=\, pq^{-1}\, \sum_{x{\sim}y\in \Edges(N)\setminus\Edges_0(N)}
   P_{p,\sigma,\gamma}^I
(\hbox{$x\in C_N$, $C_N\cap G_N =\emptyset$, $A_y\circ A_y$}) 
\label{eqnAA12}     %%ZXZ[eqnAA12]
\end{equation}
%since $x{\sim}y$ is closed if $C_N(y)$ intersects $G_N$ and $C_N(x)$
%does not intersect $G_N$.
%
%Similarly, the probability that an edge
%$x{\sim}y\in \Edges_0(N)$ exists is bound by
\begin{equation}
\hspace{-2.2cm}  \mbox{and}\hspace{5mm}  S_2 \,=\,
\sigma \tau^{-1}\, \sum_{x{\sim}y\in \Edges_0(N)}P_{p,\sigma,\gamma}^I
(\hbox{$x\in C_N$, $C_N\cap G_N =\emptyset$, $A_y\circ A_y$}),
\label{eqnAA13}     %%ZXZ[eqnAA13]
\end{equation}
writing $\tau = 1-\sigma$.

Consider a summand from equation \Ref{eqnAA12} and \Ref{eqnAA13} conditioned
on $C_N=\Xi$, with $x\in\Xi$ and $y\not\in\Xi$.  Using conditional
independence of the events $A_y$ and 
$\{ C_N \cap G_N = \emptyset\}$, and the BK inequality, we obtain
\begin{eqnarray*}
& & \hspace{-2cm} P_{p,\sigma,\gamma}^I
(\hbox{$x\in C_N$, $C_N\cap G_N =\emptyset$, $A_y\circ A_y$} \vert C_N=\Xi) \\
& & \hspace{-1cm} \;=\; P_{p,\sigma,\gamma}^I
(\hbox{$C_N\cap G_N =\emptyset \vert C_N=\Xi$}) \, P_{p,\sigma,\gamma}^I
(\hbox{$A_y\circ A_y$} \vert C_N=\Xi) \\
& & \hspace{-1cm} \;\leq\; P_{p,\sigma,\gamma}^I
(\hbox{$C_N\cap G_N =\emptyset \vert C_N=\Xi$}) \, 
\L P_{p,\sigma,\gamma}^I (\hbox{$A_y \vert C_N=\Xi$}) \R^2 \\
& & \hspace{-1cm} \;\leq\; P_{p,\sigma,\gamma}^I
(\hbox{$C_N\cap G_N =\emptyset \vert C_N=\Xi$}) \, 
P_{p,\sigma,\gamma}^I (\hbox{$A_y \vert C_N=\Xi$})\,
P_{p,\sigma,\gamma}^I (\hbox{$A_y$})  \\
& & \hspace{-1cm} \;=\; P_{p,\sigma,\gamma}^I
(\hbox{$C_N\cap G_N =\emptyset$, $A_y$} \vert C_N=\Xi)\,
\theta_N^I(p,\sigma,\gamma;y) .
\end{eqnarray*}
Substitute this into equation \Ref{eqnAA12} and average over $\Xi$.  This
gives the upper bound
\[   %\fl  \hspace{1cm}
 S_1 \;\leq \; 
pq^{-1}\, \sum_{x{\sim}y \in \Edges(N)\setminus\Edges_0(N)}P_{p,\sigma,\gamma}^I
(\hbox{$x\in C_N$, $C_N\cap G_N =\emptyset$, $A_y$} )\,
\theta_N^I(p,\sigma,\gamma;y) \,.
\]
%where the summation $x{\sim}y$ is in $\Edges(N)\setminus\Edges_0(N)$.
Next, by equation \Ref{eqn3.65} and lemma \ref{lemmaA6}(a), we obtain
(with $\theta_N^I(p,\sigma,\gamma)\equiv  \theta_N^I$)
\[ S_1 \;\leq \; pq^{-1}\, \L q\, \frac{\partial\theta_N^I }{\partial p}\R \,  
 \chi_N^H(p) \, \theta_N^I. \]
The analogous bound for equation \Ref{eqnAA13} is
\[    S_2 \;\leq \; \sigma \tau^{-1}\, 
\L \tau \, \frac{\partial \theta_N^I }{\partial \sigma}
\R  \, \chi_N^H(p) \, \theta_N^I. \]
Hence,
\begin{eqnarray*}
&P_{p,\sigma,\gamma}^I &(\LV  C_N \cap  G_N \RV \geq 2, \hbox{ and
  $A_0\circ A_0$ does not occur}) \\
&\;\leq\;&   \L p \,\frac{\partial \theta_N^I }{\partial p}
+ \sigma \,\frac{\partial \theta_N^I }{\partial \sigma}
\R \, \chi_N^H(p) \, \theta_N^I \,.  
% &\;\leq\;& \chi_N^H(p,0)\, \theta_N^I(p,\sigma,\gamma)
% \L \vec{p}\cdot \nabla \theta_N^I(p,\sigma,\gamma) \R
\end{eqnarray*}
Putting this together with equations
\Ref{eqnZ1}, \Ref{eqnZ1A}, \Ref{eqnZ2} and  \Ref{eqnZ3} completes the proof of the 
desired inequality.
% \[ \fl
% \theta_N^I(p,\sigma,\gamma) 
% \;\leq\; \gamma\, \frac{\partial\theta_N^I}{\partial \gamma}
% +  \L\theta_N^I(p,\sigma,\gamma)\R^2
% + \chi_N^H(p,0)\, \theta_N^I(p,\sigma,\gamma)
% \L \vec{p}\cdot \nabla \theta_N^I(p,\sigma,\gamma) \R .
% \]
% This is the desired inequality. 
\end{proof}

\subsection{The final differential inequalities}

To complete the proof of the two differential inequalities \Ref{eqn37} and
\Ref{eqn38}, we take the $N\to\infty$
limit in lemmas \ref{lemmaA7} and \ref{lemmaA8}.  Using
lemma \ref{lemmaA1} and equation \Ref{eqnAchi}, the result is the following theorem.

\begin{theorem}
For $p,\sigma,\gamma\in(0,1)$, write 
$\vec{p} = (p,\sigma)$, $\nabla \equiv \L \sfrac{\partial}{\partial p},
\sfrac{\partial}{\partial \sigma}\R$, $\vec{q} = \vec{1}-\vec{p}
= (1{-}p,1{-}\sigma)$, and $\theta^I \equiv \theta^I(p,\sigma,\gamma)$.  
If $p \leq \sigma$, then 
% the following differential inequalities hold:
\begin{eqnarray*}
\vec{q} \cdot \nabla \theta^I  \;\;
& \leq & \;\; 2d\, (1-\gamma) \,\chi^H(p)\, \theta^I \,
\frac{\partial \theta^I}{\partial \gamma}  \hspace{10mm}\hbox{and}  \\
   \hspace{3mm} \theta^I  
& \leq & \;\; \gamma\, \frac{\partial\theta^I}{\partial \gamma}
+  \L\theta^I\R^2
+ \chi^H(p)\, \theta^I
\L \vec{p}\cdot \nabla \theta^I \R \,.
\end{eqnarray*}
% where $\theta^I \equiv \theta^I(p,\sigma,\gamma)$.
\label{theoremA9}    %%ZXZ[theoremA9]
\end{theorem}

\bibliographystyle{plain}
\bibliography{percobib}

\end{document}